\pdfoutput=1
\documentclass[acmsmall,10pt,screen]{acmart}
\acmJournal{PACMPL}
\acmNumber{POPL} 
\acmYear{2024}

\usepackage[utf8]{inputenc} 
\usepackage{amsmath, amsthm}
\usepackage{thmtools}
\usepackage{thm-restate}
\usepackage[english]{babel}
\usepackage{pdfpages}
\usepackage{stmaryrd}
\usepackage{graphicx}
\usepackage{subcaption}
\usepackage{caption}
\usepackage{longtable}
\usepackage{booktabs}
\usepackage{tabu}
\usepackage[referable]{threeparttablex}
\usepackage{varwidth}
\usepackage{multirow}
\usepackage{wrapfig}
\usepackage{listings}
\usepackage{microtype}
\usepackage{stackengine}
\usepackage{csquotes}
\usepackage{circuitikz}
\tikzset{>=stealth}
\ctikzset{tripoles/american and port/width=0.8, tripoles/american and port/height=.65}
\ctikzset{tripoles/american or port/width=0.8, tripoles/american or port/height=.65}

\usepackage{enumitem} 

\usepackage{pifont}

\usepackage[ruled, linesnumbered, noend]{algorithm2e}
\usepackage{parskip}
\usepackage{multicol}

\usepackage{xparse}

\usepackage{thm-restate}
\usepackage[capitalise]{cleveref}

\crefname{figure}{Figure}{Figure}
\crefformat{footnote}{#2\footnotemark[#1]#3}
\usepackage{tikz}
\usepackage{comment}
\usepackage[normalem]{ulem}
\usepackage{todonotes}
\usepackage{tikz}
\usetikzlibrary{positioning,arrows,automata,shapes,decorations,decorations.markings,calc, matrix,decorations.pathmorphing, patterns,backgrounds}

\usepackage{bm}
\usepackage{pgfplots}
\usepackage{pbox}

\usepackage{appendix}
\usepackage{calc}
\usepackage{fp}
\usepackage{multirow}
\usepackage{pbox}

\usepackage{etoolbox}

\theoremstyle{plain}

\DeclareMathAlphabet{\mathpzc}{OT1}{pzc}{m}{it}

\newcommand{\Nats}{\mathbb{N}}

\newcommand{\ov}{\overline}

\newcommand{\Path}{\rightsquigarrow}
\newcommand{\DPath}[1]{\mathrel{\stackon[1pt]{$\rightsquigarrow$}{$\scriptscriptstyle#1$}}}
\newcommand{\DTo}[1]{\xrightarrow{#1}}

\newcommand{\Queue}{\mathcal{Q}}
\newcommand{\Stack}{\mathcal{S}}

\newcommand{\Label}{\lambda}
\newcommand{\Paragraph}[1]{\smallskip\noindent{\bf #1}}
\newcommand{\SubParagraph}[1]{\smallskip\noindent{\em #1}}
\newcommand{\Alphabet}{\Sigma}
\newcommand{\OpenParenthesis}{\alpha}
\newcommand{\CloseParenthesis}{\ov{\alpha}}

\newcommand{\CloseParenthesisBeta}{\ov{\beta}}

\newcommand{\CloseParenthesisGamma}{\ov{\gamma}}
\newcommand{\Open}{\operatorname{Open}}
\newcommand{\Close}{\operatorname{Close}}

\newcommand{\StartNonTerminal}{\mathcal{I}}

 \newcommand{\Dyck}{\mathcal{D}}

\newcommand{\DSCC}{\operatorname{DSCC}}
\newcommand{\PDSCC}{\operatorname{PDSCC}}
\newcommand{\DSCCRep}{\operatorname{DSCCRepr}}
\newcommand{\PDSCCRep}{\operatorname{PDSCCRepr}}
\newcommand{\Sequence}{\mathcal{S}}
\newcommand{\InsertEdge}{\operatorname{insert}}
\newcommand{\DeleteEdge}{\operatorname{delete}}
\newcommand{\OutEdges}{\textcolor{teal}{\operatorname{OutEdges}}}
\newcommand{\InEdges}{\textcolor{teal}{\operatorname{InPrimary}}}

\newcommand{\DisjointSets}{\textcolor{teal}{\operatorname{DisjointSets}}}
\newcommand{\Union}{\operatorname{Union}}
\newcommand{\Find}{\operatorname{Find}}
\newcommand{\MakeSet}{\operatorname{MakeSet}}

\newcommand{\MakePrimary}{\operatorname{MakePrimary}}

\newcommand{{\Count}}{\textcolor{teal}{\operatorname{Count}}}
\newcommand{\PrimCompDS}{\textcolor{teal}{\operatorname{PrimCompDS}}}

\newcommand{\SetClosedParenthesis}{\Sigma^C}
\newcommand{\Edges}{\textcolor{teal}{\operatorname{Edges}}}
\newcommand{\OfflineAlgo}{\operatorname{OfflineAlgo}}
\newcommand{\DynamicAlgo}{\operatorname{DynamicAlgo}}
\newcommand{\Initialization}{\operatorname{Initialization}}
\newcommand{\Fixpoint}{\operatorname{Fixpoint}}

\newcommand{\AffectedDSCCs}{\mathcal{Z}}
\newcommand{\NewRoots}{\mathcal{R}}
\newcommand{\NewIncomingRoots}{\mathcal{L}}

\SetKwProg{MyFunction}{function}{}{}
\SetKwFunction{FunctionInitialization}{{Initialization}}
\SetKwFunction{FunctionFixpoint}{{Fixpoint}}

\SetFuncSty{textnormal}
\SetNlSty{bfseries}{\color{black}}{}
\SetInd{0.4em}{0.4em}

\usetikzlibrary{arrows.meta,calc,decorations.markings,math,arrows.meta,shapes.gates.logic.US,fit}

\preto\tabular{\setcounter{magicrownumbers}{0}}
\newcounter{magicrownumbers}

\newlist{compactenum}{enumerate}{3} 
\setlist[compactenum]{label=(\arabic*), nosep,leftmargin=*}
\crefname{compactenumi}{Item}{Items}
\crefname{compactenumii}{Item}{Items}

\newlist{compactitem}{itemize}{3} 
\setlist[compactitem]{label=\textbullet, nosep,leftmargin=*}
\crefname{compactitemi}{Item}{Items}

\newlist{compactdesc}{description}{3} 
\setlist[compactdesc]{nosep,leftmargin=1em}
\crefname{compactdesci}{Item}{Items}

\pgfdeclarelayer{bg}    
\pgfsetlayers{bg,main}  

\definecolor{mybluecolor}{rgb}{0.2549019607843137, 0.2117647058823529, 0.9823529411764706}
\def \darkred {black!20!red}
\def \darkgreen {black!40!green}

\def \darkorange{black!10!orange}
\SetKw{Continue}{continue}

\def\ColorAlpha{\darkorange}
\def\ColorBeta{\darkgreen}
\def\ColorGamma{olive}

\SetCommentSty{mycommfont}

\makeatletter
\newcommand*{\centerfloat}{%
\parindent \z@
\leftskip \z@ \@plus 1fil \@minus \textwidth
\rightskip\leftskip
\parfillskip \z@skip}
\makeatother

\definecolor{mGreen}{rgb}{0,0.6,0}
\definecolor{mGray}{rgb}{0.5,0.5,0.5}
\definecolor{mPurple}{rgb}{0.58,0,0.82}
\definecolor{backgroundColour}{rgb}{0.95,0.95,0.92}

\lstdefinestyle{CStyle}{
tabsize = 2, 
showstringspaces = false, 
backgroundcolor=\color{backgroundColour},   
numbers = left, 
commentstyle = \color{green}, 
keywordstyle = \color{blue}, 
stringstyle = \color{red}, 
rulecolor = \color{black}, 
basicstyle = \small \ttfamily , 
breaklines = true, 
numberstyle = \tiny,
language=Java
}

\newcounter{listtotal}\newcounter{listcntr}%
\NewDocumentCommand{\names}{o}{%
  \setcounter{listtotal}{0}\setcounter{listcntr}{0}%
  \renewcommand*{\do}[1]{\stepcounter{listtotal}}%
  \expandafter\docsvlist\expandafter{\namesarray}%
  \IfNoValueTF{#1}
    {\namesarray}
    {
     \renewcommand*{\do}[1]{\stepcounter{listcntr}\ifnum\value{listcntr}=#1\relax##1\fi}%
     \expandafter\docsvlist\expandafter{\namesarray}}%
}

\makeatletter
\g@addto@macro\bfseries{\boldmath}
\makeatother
                
\begin{document}

\title{On-The-Fly Static Analysis via Dynamic Bidirected Dyck Reachability}

\author{Shankaranarayanan Krishna}
\orcid{0000-0003-0925-398X}
\affiliation{
  \institution{IIT Bombay}            
  \country{India}                    
}
\email{krishnas@cse.iitb.ac.in} 

\author{Aniket Lal}
\orcid{}
\affiliation{
  \institution{IIT Bombay}            
  \country{India}                    
}
\email{aniketlal@cse.iitb.ac.in} 

\author{Andreas Pavlogiannis}
\orcid{0000-0002-8943-0722}             
\affiliation{
\institution{Aarhus University}            
\streetaddress{Aabogade 34}
\city{Aarhus}
\postcode{8200}
\country{Denmark}                    
}
\email{pavlogiannis@cs.au.dk}          

\author{Omkar Tuppe}
\orcid{}
\affiliation{
  \institution{IIT Bombay}            
  \country{India}                    
}
\email{omkarvtuppe@cse.iitb.ac.in}

\begin{abstract}
Dyck reachability is a principled, graph-based formulation of a plethora of static analyses.
Bidirected graphs are used for capturing dataflow through mutable heap data, and are usual formalisms of demand-driven points-to and alias analyses.
The best (offline) algorithm runs in $O(m+n\cdot \alpha(n))$ time, where $n$ is the number of nodes and $m$ is the number of edges in the flow graph, which becomes $O(n^2)$ in the worst case.

In the everyday practice of program analysis, the analyzed code is subject to continuous change, with source code being added and removed.
On-the-fly static analysis under such continuous updates  gives rise to \emph{dynamic Dyck reachability}, where reachability queries run on a dynamically changing graph, following program updates.
Naturally, executing the \emph{offline} algorithm in this \emph{online} setting is inadequate,
as the time required to process a single update is prohibitively large.

In this work we develop a novel dynamic algorithm for bidirected Dyck reachability that has $O(n\cdot \alpha(n))$ worst-case performance per update, thus beating the $O(n^2)$ bound, and is also optimal in certain settings.
We also implement our algorithm and evaluate its performance on on-the-fly data-dependence and alias analyses, and compare it with two best known alternatives, namely
(i)~the optimal offline algorithm, and
(ii)~a fully dynamic Datalog solver.
Our experiments show that our dynamic algorithm is consistently, and by far, the top performing algorithm,
exhibiting speedups in the order of 1000X.
The running time of each update is almost always unnoticeable to the human eye, making it ideal for the on-the-fly analysis setting.
\end{abstract}



\maketitle

\section{Introduction}\label{sec:intro}

\Paragraph{Dyck reachability in static analysis.}
Dyck reachability is an elegant and widespread graph-based formulation of a plethora of static analyses.
The reachability problem is phrased on labeled directed graphs $G=(V,E)$, where $V$ is a set of nodes and $E$ is a set of edges, labeled with opening and closing parenthesis symbols.
Given two nodes $u$ and $v$, the task is to determine whether $v$ is reachable from $u$ by means of a path $P$, such that the sequence of labels of the edges of $P$ produce a parenthesis string that is properly balanced~\cite{Yannakakis90}.
In the static analysis domain, $G$ serves as the program model, where nodes represent basic program constructs such as program variables, pointers, or statements, and edges capture data flow between these constructs.
Effectively, program executions carrying data flow between distant program points are captured by paths in $G$.
Naturally, as $G$ is an approximate model of the program, it may have paths that do not correspond to any valid program execution, leading to spurious data flow and thus to false positive warnings.
In order to increase the precision of the analysis, parenthesis symbols are used to model certain restrictions in the data flow, such as sensitivity to calling contexts and field accesses~\cite{Reps97,Reps00,Spath2019}.
Focusing on Dyck reachability on $G$ (as opposed to plain reachability) thus filters out paths that are guaranteed to not correspond to valid program executions, thereby increasing the precision of the analysis.

As a modeling formalism, Dyck reachability strikes a remarkable balance between simplicity and expressiveness, and has been used to drive a wide range of static analyses, such as interprocedural data-flow analysis~\cite{Reps95}, slicing~\cite{Reps94}, 
shape analysis~\cite{Reps1995b}, impact analysis~\cite{Arnold96}, type-based flow analysis~\cite{Rehof01}, taint analysis~\cite{Huang2015},
data-dependence analysis~\cite{Tang2017},
alias/points-to analysis~\cite{Lhotak06,Zheng2008,Xu09}, to only name a few.
In practice, widely-used analysis tools, such as Wala~\cite{Wala} and Soot~\cite{Bodden12} equip Dyck-reachability techniques to perform the analysis.
From a complexity perspective, answering a single Dyck-reachability query ``is $v$ reachable from $u$?'' takes $O(n^3)$ time, where $n$ is the number of nodes of $G$, and although some slight improvements are possible~\cite{Chaudhuri2008}, this cubic dependency is often considered prohibitive~\cite{Heintze97}.

\Paragraph{Bidirectedness.}
One important and practically motivated variant of Dyck reachability is that of \emph{bidirectedness}. 
Intuitively, a bidirected graph $G$ has the property that every (directed) edge is present in both directions with complementary labels:~an edge $x\DTo{(}y$ is present in $G$  if and only if $y\DTo{)}x$ is also present in $G$.
Bidirectedness makes reachability an \emph{equivalence relation}:~if $v$ is reachable from $u$ via a properly balanced path, the same path backwards (i.e., from $v$ to $u$) is also properly balanced, thereby witnessing the reachability of $u$ from $v$. Thus, the nodes of $G$ can be partitioned into equivalence classes of inter-reachable nodes, called \emph{Dyck Strongly Connected Components} (or \emph{DSCCs}, for short). 
For example, in \cref{fig:illustrative_example} the bidirected graph $G$ (top right) corresponds to the input program,
and nodes $c$ and $d$ are inter-reachable via the paths $c\DTo{L}f\DTo{\ov{L}} d$ and $d\DTo{L}f\DTo{\ov{L}} c$, and thus they belong to the same DSCC.

From a \emph{semantics} perspective, bidirectedness has been a standard approach to handle mutable heap data~\cite{Sridharan2006,Xu09,Lu2019,Zhang2017}
-- though it can sometimes be relaxed for read-only accesses~\cite{Milanova2020},
and the de-facto formulation of demand-driven points-to analyses~\cite{Sridharan2005,Zheng2008,Shang2012,Yan11,Vedurada2019}.
From an \emph{algorithmic} perspective, bidirectedness allows Dyck reachability to be computed more efficiently~\cite{Yuan09,Zhang13}, with the fastest algorithm running in $O(m+n\cdot \alpha(n))$ time~\cite{Chatterjee18}, where $n$ is the number of nodes, $m$ is the number of edges, and $\alpha(n)$ is the inverse Ackermann function\footnote{For all practical purposes, $\alpha(n)\leq 5$, i.e., the function behaves as a constant.}.
Compared to the cubic bound of directed Dyck reachability, bidirectedness thus allows for a large speedup in the order of $n^2$.
Due to this algorithmic benefit, bidirectedness also serves as an overapproximation to directed reachability, and has been suggested as a mechanism to speed up challenging static analysis and verification problems~\cite{Li2020,Ganardi2022,Ganardi2022a}.

\input{figures/illustrative-example}
\Paragraph{On-the-fly static analysis.}
In the everyday practice of program analysis, the analyzed code base is not static but rather subject to perpetual change; source-code lines are added and removed following patches, new modules and libraries, and bug fixes.
In an even more demanding setting, lightweight static analyzers are embedded inside the IDE and run on-the-fly, so as to enable faster, more robust and more productive development.
As has been observed in the literature, a major computational challenge in static analyses is in their apparent inability to adapt to these continuous changes efficiently~\cite{Zadeck1984,Arzt2014}.
Existing works approach this challenge mostly by tweaking the offline analyses and adapting them to the dynamic setting.
Typically, such adaptations only focus on incremental updates (i.e., only the addition of code lines) and are  based on some form of caching~\cite{Burke1990,Arzt2014,Szabo2016,Liu2019,Pacak2020}.
Decremental updates (i.e., the deletion of code lines) is much more difficult, intuitively, due to the fixpoint nature of static analyses.
As such, analyses that fully follow all code changes (both incremental and decremental) thus far offer no efficiency guarantees.
This paper tackles this challenge of on-the-fly, full dynamic analyses formulated as bidirected Dyck reachability in a rigorous and provably efficient way.
\cref{fig:illustrative_example} illustrates the use of dynamic bidirected Dyck reachability for on-the-fly field-sensitive alias analysis on a small example.

\Paragraph{Our contributions.}
We consider the problem of maintaining the DSCCs of a dynamic graph $G$, representing the on-the-fly static analysis of a program, where source code changes are both incremental and decremental.
In particular, an incremental update inserts an edge in $G$, while a decremental update deletes edges from $G$, and the task is to restore the DSCCs of $G$ after each such update.
To this end, we develop an algorithm $\DynamicAlgo$, with guarantees as stated in the following theorem.

\begin{restatable}{theorem}{thmmaintheorem}\label{thm:main_theorem}
Given a bidirected graph $G$ of $n$ nodes, $\DynamicAlgo$ correctly maintains the DSCCs of $G$ across edge insertions and edge deletions, and uses at most $O(n\cdot \alpha(n))$ time for each update.
\end{restatable}

Observe that in general this is much faster than the optimal offline algorithm of~\cite{Chatterjee18}, as the latter spends time that is at least proportional to the number of edges $m$ (which can be up to $m= \Theta(n^2)$).
On the other hand, it can be easily shown that a single update (either edge insertion or deletion) can affect $\Theta(n)$ DSCCs.
Thus $\DynamicAlgo$ is effectively optimal, at least in the natural setting where  DSCCs have to be explicitly maintained at all times.

We have implemented our dynamic algorithm and evaluated its performance on Dyck graphs that arise on data-dependence analysis and alias analysis, on standard benchmarks.
Our experiments show that our dynamic algorithm is consistently, and by far, the top performing algorithm compared to 
(i)~the optimal offline algorithm for the problem, and
(ii)~a fully dynamic Datalog solver,
arguably the two most relevant approaches to our problem setting.
In particular, our dynamic algorithm exhibits speedups that are between 100X-1000X for data-dependence analysis and 1000X for alias analysis (which is also the more demanding of the two), compared to the best alternative.
Although its theoretical worst-case complexity is linear in $n$, its average running time is barely (if at all) noticeable for all practical purposes, making it suitable for the on-the-fly analysis setting that runs continuously during development.

\Paragraph{A note on earlier approaches.}
Dynamic bidirected Dyck reachability was studied recently in~\cite{Li2022}.
To this end, that work claims a dynamic algorithm with $O(n\cdot \alpha(n))$ running time per update operation, similarly to our \cref{thm:main_theorem}.
Unfortunately, both the complexity and the correctness of~\cite{Li2022} overlook certain cases.
In particular, that algorithm exhibits $\Omega(n^2)$ running time, as opposed to the claimed (nearly) linear bound.
This quadratic behavior arises even on sparse graphs (i.e., with $m=O(n)$ edges), for which the vanilla offline algorithm runs in $O(n\cdot \alpha(n))$ time.
Hence on such graphs, processing a single line deletion in the source code is $n$ times slower than performing the whole analysis from scratch.
Moreover, that algorithm also suffers from correctness issues, as it fails to detect that certain nodes become unreachable after updates.
In such cases, the on-the-fly static analysis returns wrong results.
We provide further details in \cref{sec:mistakes}.

\Paragraph{Intuition behind our approach.}
In high-level, our approach works as follows.
First, we observe that the component graph (i.e., the graph with a single node representing each DSCC) is sparse.
Each edge insertion can then be handled by executing the offline algorithm of~\cite{Chatterjee18} not from scratch, but directly on the already constructed component graph, yielding $O(n\cdot \alpha (n))$ running time.
The main difficulty arises in edge deletions, as each deletion seemingly requires repeating all the computation from scratch, i.e., operating on the initial graph, which might be dense and thus incur $O(n^2)$ running time.
We leverage techniques from dynamic undirected connectivity, together with some sparsification ideas that are specific to our richer bidirected setting, to start the recomputation not from scratch, but from a suitable preliminary component graph, called  the \emph{primary component graph}, that is already sparse, thereby recovering the $O(n\cdot \alpha(n))$ running time.
Maintaining the primary component graph efficiently is the main technical challenge.

\section{Preliminaries}\label{sec:preliminaries}

In this section we develop general notation on labeled graphs, Dyck reachability, and fully-dynamic reachability queries.
As this is somewhat standard material, our exposition follows closely that of related works (e.g.,~\cite{Zhang13,Chatterjee18,Kjelstrom2022}).

\subsection{Dyck Reachability on Bidirected Graphs}\label{subsec:dyck_reachability}

\Paragraph{Dyck Languages.}
Hereinafter, we fix a natural number $k\in \Nats$. Let $[k]$ denote $\{1,\dots,k\}$. 
A Dyck alphabet is a set $\Alphabet=\{\alpha_i, \ov{\alpha}_i\}_{i\in [k]}$ of $k$ matched symbols, usually referred to as parentheses, where
$\Open(\Alphabet)=\{ \alpha_i \}_{i\in k}$ and $\Close(\Alphabet)=\{ \ov{\alpha}_i \}_{i\in k}$
are the sets of opening parenthesis symbols and closing parenthesis symbols of $\Alphabet$, respectively\footnote{We use Greek letters such as $\alpha, \beta$ to represent opening parenthesis symbols, and $\ov{\alpha}, \ov{\beta}$ for their matching closing parentheses.}.
The Dyck language $\Dyck$ is the set of strings over $\Alphabet^*$ produced by the following grammar with initial non-terminal $\StartNonTerminal$:
\[
\StartNonTerminal \to \StartNonTerminal ~ \StartNonTerminal ~ | ~ \OpenParenthesis_1 \StartNonTerminal\CloseParenthesis_1 ~ | ~ \cdots ~ | ~\OpenParenthesis_k \StartNonTerminal\CloseParenthesis_k ~ | ~ \epsilon
\]
For example, $\alpha_1\alpha_2\ov{\alpha}_2 \alpha_3\ov{\alpha}_3\ov{\alpha}_1\in \Dyck$, but
$\alpha_1\alpha_2 \alpha_3\ov{\alpha}_3\ov{\alpha}_1 \ov{\alpha}_2 \not \in \Dyck$.
Finally, given a string $\sigma=\gamma_1\dots \gamma_m\in \Alphabet^*$, we let 
$\ov{\sigma}=\ov{\gamma}_1\dots \ov{\gamma}_m\in \Alphabet^*$,  
where $\ov{\gamma}_i$ is a closing parenthesis symbol if $\gamma_i$ is an opening parenthesis symbol, and vice versa.
For example, if $\sigma=\alpha_1\alpha_2\ov{\alpha}_2 \alpha_3\ov{\alpha}_3\ov{\alpha}_1$ then
$\ov{\sigma}=\ov{\alpha}_1\ov{\alpha}_2\alpha_2 \ov{\alpha}_3\alpha_3\alpha_1$.

\Paragraph{Graphs and Dyck reachability.}
We consider labeled directed graphs $G=(V,E)$ where $V$ is a set of nodes and $E$ is a set of labeled edges $(u,v, l)$, where $u,v\in V$ and $l\in \Alphabet\cup\{\epsilon\}$.
For notational convenience, given an edge $e$, we will refer to its label as $\Label(e)$.
We occasionally are not interested in the label of an edge, in which case we will simply denote it with its endpoints, e.g., $e=(u,v)$.
We also adopt the pictorial notation $u\DTo{l}v$ to denote the existence of an edge $(u,v,l)$, and write $u\DTo{l}\cdot$ to denote the existence of an edge outgoing $u$ and labeled with $l$.
Given a set $X$, we sometimes write $u\DTo{l}X$ to denote that there exists some $v\in X$ such that $u\DTo{l}v$.
We also use similar notation for incoming, instead of outgoing edges (e.g., $X\DTo{l}v$).
A path is a sequence of edges $P=(u_1, v_1, l_1), \dots, (u_r, v_r, l_r)$ such that, for each $i\in[r-1]$, we have $u_{i+1}=v_i$.
The label of $P$ is $\Label(P)=l_1\dots l_r$, that is, it is the concatenation  of the labels of its edges.
A path can also be an empty sequence of edges, in which case its label is $\epsilon$.
We often write $P\colon u\Path v$ to denote a path from node $u$ to node $v$.
Naturally, we say that $v$ is reachable from $u$ if such a path exists.
Moreover, we say that $v$ is Dyck-reachable from $u$ if there exists a path $P\colon u\Path v$ such that $\Label(P)\in \Dyck$,
in which case we call $P$ a Dyck path.

In alignment to existing literature (\cite{Zhang13,Chatterjee18}), we consider that the number of parenthesis types  is constant compared to the size of $G$ (i.e., $k=O(1)$).
From a theoretical standpoint this is not an oversimplification, as any Dyck graph with arbitrary $k$ can be transformed to one with $k=O(1)$, while preserving the reachability relationships (see e.g., ~\cite{Chistikov2022}).

\Paragraph{Bidirected Graphs and Dyck SCCs.}
We call a graph $G=(V,E)$ bidirected if every edge in $E$ appears in both ways with complementary labels.
Formally, let $\ov{\epsilon}=\epsilon$, and we have:
\[
(u,v,s)\in E \qquad  \iff \qquad (v,u,\ov{s})\in E
\]
Bidirectedness turns reachability into an equivalence relation, similarly to plain undirected graphs.
To realize this, notice that every path $P\colon u\Path v$ has a reverse version $\ov{P}\colon v\Path u$ with $\Label(P)=\Label(\ov{P})$.
Hence, if $P$ witnesses the Dyck reachability of $v$ from $u$, then $\ov{P}$ witnesses the Dyck reachability of $u$ from $v$.
Naturally, the nodes of $G$ are partitioned into strongly-connected components (SCCs), which are maximal sets of pairwise Dyck-inter-reachable nodes.
We will call such sets Dyck strongly connected components, or DSCCs, for short.
Note that, in contrast to standard SCCs, paths witnessing the Dyck reachability of two nodes in the same DSCC might have to traverse nodes outside the DSCC.
For example,  in Figure  \ref{fig:illustrative_example} (top right), the paths witnessing the Dyck reachability of nodes $c$ to $d$  and $d$ to $c$ go via node $f$, which is outside the DSCC $\{c,d\}$.
Given a node $u$, we let $\DSCC(u)$ be the DSCC in which $u$ appears.

For notational convenience, hereinafter we represent bidirected graphs by only explicitly denoting their edges labeled with closing parentheses, with the understanding that the reverse edges (labeled with an opening parentheses) are also present.
In our figures, when a graph contains edges of different labels, we color-code the edges for easier visual representation.

\Paragraph{The efficiency of DSCC computation.}
In Dyck-reachability formulations of static analyses, the analysis queries are phrased as Dyck-reachability queries in the underlying graph.
As we have seen above, the nodes of bidirected graphs are partitioned into DSCCs.
Thus, reachability queries can be handled by first computing these DSCCs.
Then, given a reachability query on two nodes $u,v$, we answer that they are inter-reachable iff they are found in the same DSCC, which costs a simple lookup of constant time.
Computing Dyck reachability on general (non-bidirected) graphs takes cubic time $O(n^3)$ even for a single pair. 
On the other hand, it is known that the DSCCs of bidirected graphs can be computed efficiently.

\begin{theorem}[\cite{Chatterjee18}]\label{thm:offline}
Given a graph $G$ of $n$ nodes and $m$ edges, the DSCCs of $G$ can be computed in $O(m+n\cdot \alpha(n))$ time,
where $\alpha(n)$ is the inverse Ackermann function.
\end{theorem}

\subsection{Dynamic Reachability}\label{subsec:dynamic_reachability}

\Paragraph{Dynamic reachability on bidirected graphs.}
As the program source is being developed, the underlying graph model changes shape due to the addition and removal of nodes and edges.
The goal of the fully dynamic setting is to maintain Dyck-reachability information on the fly under such changes, so that analysis queries can be performed fast.
More formally, we are processing a sequence of operations $\Sequence=(o_1,o_2,\dots, o_r)$, where each operation $o_i$ has one of the following types.
\begin{compactenum}
\item $\InsertEdge(u,v,\CloseParenthesis)$, which inserts an edge $u\DTo{\CloseParenthesis} v$ (as well as the implicit inverse edge $v\DTo{\OpenParenthesis} u$).
\item $\DeleteEdge(u,v,\CloseParenthesis)$, which deletes the edge $u\DTo{\CloseParenthesis} v$ (as well as the implicit inverse edge $v\DTo{\OpenParenthesis} u$).
\item $\DSCCRep(u)$, which returns a representative node of the DSCC of $u$ that is the same for all nodes $v\in \DSCC(u)$.
\end{compactenum}
This sequence of operations creates a sequence of graphs $(G_i=(V, E_i))_{i\in[r]}$, where $G_1$ is the initial graph, and $G_{i+1}$ is obtained from $G_i$ by performing the operation $o_i$ (naturally, if $o_i=\DSCCRep(u)$ is a query operation, $G_{i+1}=G_i$).
Note that we have kept the node set identical across all $G_i$.
Indeed, this is a standard approach for studying such problems:~we may simply assume that $V$ contains all nodes that are ever added to the graph, and instead of deleting a node, we can delete all its adjacent edges.
Finally, the reachability of $v$ from $u$ can be checked with two successive queries, i.e., testing whether $\DSCCRep(u)=\DSCCRep(v)$.

Although in general there might be a spectrum of trade offs in the time taken between operations modifying the graph  ($\InsertEdge(u,v,\CloseParenthesis)$, $\DeleteEdge(u,v,\CloseParenthesis)$) and reachability queries ($\DSCCRep(u)$), the linear upper bounds we establish here for updates means that query operations can be done in constant time $O(1)$.
This holds because, in theory, after every update operation we can easily, in $O(n)$ time (which is within the upper bound we establish for processing the update operations), retrieve the DSCC of each node to a simple lookup table that will lead to subsequent queries taking $O(1)$ time each.
In practice, of course, this approach is to be avoided, as dynamic updates cost much less than $O(n)$ time, hence building the full lookup table causes an unnecessary cost.
Moreover, even without a lookup table, queries cost $O(\alpha(n))$ time, which is constant for all practical purposes.

\Paragraph{Dynamic reachability on undirected graphs.}
Problems of dynamic reachability have been studied extensively on plain graphs (i.e., without the restriction on Dyck paths), both for directed and undirected variants.
However, the richer semantics of Dyck reachability makes the Dyck setting quite more intricate.
Our solution to the problem uses as a black box a data structure for dynamic reachability on undirected graphs.
Although there exist several results in this direction,  \cref{tab:connectivity_data_structures} shows two standard data structures that are relevant to our work.
The one developed in~\cite{Eppstein1997} offers standard worst-case guarantees for the time to perform each update (i.e., inserting/deleting an undirected edge) and query (i.e., obtaining the component of a node $u$).
The one developed in~\cite{Holm2001} offers amortized guarantees for the respective tasks.
That is, some operations might exceed the stated time bounds, but this excess is balanced in later operations, so that the average time per operation is the one stated.
As we will see later, we will be using the former technique to establish our bound in \cref{thm:main_theorem}, but the latter in our experiments.
\begin{table}
\caption{
\label{tab:connectivity_data_structures}
Classic results on dynamic undirected reachability.
}
\setlength\tabcolsep{5pt}
\begin{tabular}{cccc}
Reference & Update time & Query time & Guarantee\\
\hline
\cite{Eppstein1997} & $O(\sqrt{n})$ & $O(1)$ & Worst-case\\
\cite{Holm2001} & $O(\log^2 n)$ & $O(\log n)$ & Amortized
\end{tabular}
\end{table}

\section{Efficient Dynamic Bidirected Dyck Reachability}\label{sec:data_structure}

In this section we present our dynamic algorithm for bidirected Dyck reachability.
As this is the main technical section of our paper, we provide here an outline of its structure to guide the reader.
\begin{compactenum}
\item In \cref{subsec:offline} we present the key algorithmic intuition for solving Dyck reachability on bidirected graphs, and revisit the offline algorithm for the problem, as presented in~\cite{Chatterjee18}.
Some aspects of the offline algorithm (in particular, its fixpoint computation) will also be used later in our dynamic algorithm.
\item In \cref{subsec:intuition} we present the main technical challenges for handling dynamic updates efficiently, and the key concepts our algorithm uses for handling each challenge.
This section is filled with examples that illustrate how each concept is used, and what data structure is used to support it.
\item In \cref{subsec:main_data_structure} we present $\DynamicAlgo$ in detail, with pseudocode followed by a high-level description of each of its blocks.
The pseudocode is also heavy in comments to guide the reader.
\item In \cref{subsec:example} we give a step-by-step execution of the algorithm on the running example of  \cref{fig:illustrative_example}.
\item Finally, in \cref{subsec:analysis}, we establish the correctness and complexity of $\DynamicAlgo$.
We state its main invariants, together with some intuition around them and how they are used towards the correctness and complexity lemmas, while we delegate the proofs to \cref{sec:app_proofs}.
\end{compactenum}

\subsection{Offline Reachability}\label{subsec:offline}

We start with the optimal offline algorithm for computing bidirected Dyck reachability~\cite{Chatterjee18}, i.e., for obtaining \cref{thm:offline}.
This will allow us to build some intuition about the problem.
Moreover, our data structure for handling graph updates later will use some components of the offline algorithm.

\Paragraph{Intuitive description.}
The principle of operation of the algorithm behind \cref{thm:offline} (initially observed in~\cite{Zhang13}) is as follows.
Assume that we have computed a DSCC $S_1$ that has two nodes $u, v\in S_1$, and further, $G$ has two edges $u\DTo{\CloseParenthesis} x$ and $v\DTo{\CloseParenthesis} y$.
Here $u$ and $v$ are not necessarily distinct, meaning that possibly $u=v$.
Then we can conclude that $x$ and $y$ also belong to the same DSCC, witnessed by the path
\[
x \DTo{\OpenParenthesis} u \Path v\DTo{\CloseParenthesis} y
\]
where the intermediate path $u \Path v$ is one witnessing the reachability of $v$ from $u$ -- such a path exists since  $u$ and $v$ belong to the same DSCC.
In fact, any node $x'$ in the DSCC of $x$ can reach any node $y'$ in the DSCC of $y$, via the path
\[
x'\Path x \DTo{\OpenParenthesis} u \Path v\DTo{\CloseParenthesis} y \Path y'
\]
Thus the DSCCs of $x$ and $y$ can be merged into one DSCC $S_2$.
We can now repeat the above process  to discover Dyck paths that go through $S_2$, and so on, up to a fixpoint.
\cref{fig:principle_of_operation} illustrates this process on a small example.

\begin{figure}
\scalebox{0.85}{%
\begin{tikzpicture}[thick, >=latex, node distance=0.3cm and 1cm,
pre/.style={<-,shorten >= 1pt, shorten <=1pt,},
post/.style={->,shorten >= 1pt, shorten <=1pt,},
und/.style={very thick, draw=gray},
node/.style={circle, minimum size=4.5mm, draw=black!100, fill=white!100, thick, inner sep=0},
virt/.style={circle,draw=black!50,fill=black!20, opacity=0}]

\newcommand{\xdisposition}{4.18}
\newcommand{\ydisposition}{0}
\newcommand{\xstep}{0.65}
\newcommand{\ystep}{1.2}
\def\bend{20}

\begin{scope}[shift={(0*\xdisposition,0)}]
\node	[node]		(u)	at (0*\xstep,0*\ystep) {$u$};
\node	[node]		(z)	at (2*\xstep,0*\ystep) {$z$};
\node	[node]		(w)	at (3.25*\xstep,0*\ystep) {$w$};
\node	[node]		(x)	at (-1*\xstep,-1*\ystep) {$x$};
\node	[node]		(y)	at (1*\xstep,-1*\ystep) {$y$};
\node	[node]		(v)	at (3*\xstep,-1*\ystep) {$v$};

\draw[post, color=\ColorAlpha, ultra thick] (u) to node[left, pos=0.1] {$\CloseParenthesis$} (x);
\draw[post, color=\ColorAlpha, ultra thick] (u) to node[right, pos=0.1] {$\CloseParenthesis$} (y);
\draw[post, color=\ColorBeta] (z) to node[left, pos=0.1] {$\CloseParenthesisBeta$} (y);
\draw[post, color=\ColorBeta] (z) to node[right, pos=0.1] {$\CloseParenthesisBeta$} (v);

\draw[post, color=\ColorGamma, out=120, in=120, looseness=1.1] (x) to node[above left] {$\CloseParenthesisGamma$} (z);
\draw[post, color=\ColorGamma] (v) to node[right, pos=0.1] {$\CloseParenthesisGamma$} (w);
\end{scope}

\begin{scope}[shift={(1*\xdisposition,0)}]
\node	[node]		(u)	at (0*\xstep,0*\ystep) {$u$};
\node	[node]		(z)	at (2*\xstep,0*\ystep) {$z$};
\node	[node]		(w)	at (3.25*\xstep,0*\ystep) {$w$};
\node	[node]		(x)	at (-1*\xstep,-1*\ystep) {$x$};
\node	[node]		(y)	at (1*\xstep,-1*\ystep) {$y$};
\node	[node]		(v)	at (3*\xstep,-1*\ystep) {$v$};

\draw[post, color=\ColorAlpha] (u) to node[left, pos=0.1] {$\CloseParenthesis$} (x);
\draw[post, color=\ColorAlpha] (u) to node[right, pos=0.1] {$\CloseParenthesis$} (y);
\draw[post, color=\ColorBeta, ultra thick] (z) to node[left, pos=0.1] {$\CloseParenthesisBeta$} (y);
\draw[post, color=\ColorBeta, ultra thick] (z) to node[right, pos=0.1] {$\CloseParenthesisBeta$} (v);

\draw[post, color=\ColorGamma, out=120, in=120, looseness=1.1] (x) to node[above left] {$\CloseParenthesisGamma$} (z);
\draw[post, color=\ColorGamma] (v) to node[right, pos=0.1] {$\CloseParenthesisGamma$} (w);

\begin{pgfonlayer}{bg}
\node[box, ultra thick, rounded corners, draw=mybluecolor, dotted, fill=mybluecolor!10, minimum height=0.9cm, fit=(x)(y)] (S1) {};
\end{pgfonlayer}

\end{scope}

\begin{scope}[shift={(2*\xdisposition,0)}]
\node	[node]		(u)	at (0*\xstep,0*\ystep) {$u$};
\node	[node]		(z)	at (2*\xstep,0*\ystep) {$z$};
\node	[node]		(w)	at (3.25*\xstep,0*\ystep) {$w$};
\node	[node]		(x)	at (-1*\xstep,-1*\ystep) {$x$};
\node	[node]		(y)	at (1*\xstep,-1*\ystep) {$y$};
\node	[node]		(v)	at (3*\xstep,-1*\ystep) {$v$};

\draw[post, color=\ColorAlpha] (u) to node[left, pos=0.1] {$\CloseParenthesis$} (x);
\draw[post, color=\ColorAlpha] (u) to node[right, pos=0.1] {$\CloseParenthesis$} (y);
\draw[post, color=\ColorBeta] (z) to node[left, pos=0.1] {$\CloseParenthesisBeta$} (y);
\draw[post, color=\ColorBeta] (z) to node[right, pos=0.1] {$\CloseParenthesisBeta$} (v);

\draw[post, color=\ColorGamma, ultra thick,  out=120, in=120, looseness=1.1] (x) to node[above left] {$\CloseParenthesisGamma$} (z);
\draw[post, color=\ColorGamma, ultra thick] (v) to node[right, pos=0.1] {$\CloseParenthesisGamma$} (w);

\begin{pgfonlayer}{bg}
\node[box, ultra thick, rounded corners, draw=mybluecolor, dotted, fill=mybluecolor!10, minimum height=0.9cm, fit=(x)(y)(v)] (S1) {};
\end{pgfonlayer}

\end{scope}

\begin{scope}[shift={(3*\xdisposition,0)}]
\node	[node]		(u)	at (0*\xstep,0*\ystep) {$u$};
\node	[node]		(z)	at (2*\xstep,0*\ystep) {$z$};
\node	[node]		(w)	at (3.25*\xstep,0*\ystep) {$w$};
\node	[node]		(x)	at (-1*\xstep,-1*\ystep) {$x$};
\node	[node]		(y)	at (1*\xstep,-1*\ystep) {$y$};
\node	[node]		(v)	at (3*\xstep,-1*\ystep) {$v$};

\draw[post, color=\ColorAlpha] (u) to node[left, pos=0.1] {$\CloseParenthesis$} (x);
\draw[post, color=\ColorAlpha] (u) to node[right, pos=0.1] {$\CloseParenthesis$} (y);
\draw[post, color=\ColorBeta] (z) to node[left, pos=0.1] {$\CloseParenthesisBeta$} (y);
\draw[post, color=\ColorBeta] (z) to node[right, pos=0.1] {$\CloseParenthesisBeta$} (v);

\draw[post, color=\ColorGamma, out=120, in=120, looseness=1.1] (x) to node[above left] {$\CloseParenthesisGamma$} (z);
\draw[post, color=\ColorGamma] (v) to node[right, pos=0.1] {$\CloseParenthesisGamma$} (w);

\begin{pgfonlayer}{bg}
\node[box, ultra thick, rounded corners, draw=mybluecolor, dotted, fill=mybluecolor!10, minimum height=0.9cm, fit=(x)(y)(v)] (S1) {};
\node[box, ultra thick, rounded corners, draw=mybluecolor, dotted, fill=mybluecolor!10, minimum height=0.9cm, fit=(z)(w)] (S2) {};
\end{pgfonlayer}

\end{scope}

\end{tikzpicture}
}
\caption{
Illustration of the fixpoint computation of DSCCs.
Initially, every node forms its own DSCC.
We have $u\DTo{\CloseParenthesis}x$ and $u\DTo{\CloseParenthesis}y$, forming a DSCC $S_1=\{x,y\}$,
witnessed by the paths $x\DTo{\OpenParenthesis}u\DTo{\CloseParenthesis}y$ and $y\DTo{\OpenParenthesis}u\DTo{\CloseParenthesis}x$.
In turn, the two edges $z\DTo{\CloseParenthesisBeta}y$ and $z\DTo{\CloseParenthesisBeta}v$ enlarge this DSCC to $S_1=\{x,y, v\}$.
Finally, the two edges $x\DTo{\CloseParenthesisGamma}z$ and $v\DTo{\CloseParenthesisGamma}w$ form another DSCC $S_2=\{z,w\}$.
\label{fig:principle_of_operation}
}
\end{figure}

\Paragraph{Algorithm $\OfflineAlgo$.}
The above insights are made formal in \cref{algo:offlinealgo}, which is an adaptation of the algorithm as appeared initially in~\cite{Chatterjee18}.
The first insight that $\OfflineAlgo$ rests on is that, since DSCCs are equivalence classes that only merge together during the fixpoint computation, they can be maintained efficiently using a disjoint-sets data structure that supports $\Union/\Find$ operations.
The second insight is to use linked lists for storing the edges outgoing each DSCC and labeled with a given symbol $\CloseParenthesis$.
When two DSCCs are merged into one, the algorithm also merges the corresponding linked lists, which can be done very efficiently, in $O(1)$ time (using simple pointers).
In particular, $\OfflineAlgo$ performs two big steps.
First, it initializes a disjoint-sets data structure $\DisjointSets$ and a worklist $\Queue$ (function $\FunctionInitialization{}$ in \cref{algo:offline_initialize}).
Then, it uses $\DisjointSets$ and $\Queue$ to compute the fixpoint (function $\FunctionFixpoint{}$ in \cref{algo:offline_fixpoint}), achieving the running time of $O(m+n\cdot \alpha(n))$.
As the correctness and complexity are established in~\cite{Chatterjee18}, we do not re-establish them here.

\begin{algorithm}
\small
\DontPrintSemicolon
\caption{$\OfflineAlgo$}\label{algo:offlinealgo}
\KwIn{A bidirected graph $G=(V,E)$}
\KwOut{A $\DisjointSets$ map of the DSCCs of $G$}
\BlankLine
$\Initialization()$\tcp*[f]{Initialize the fixpoint computation}\label{algo:offline_initialize}\\
$\Fixpoint()$\tcp*[f]{Compute the fixpoint}\label{algo:offline_fixpoint}\\
\BlankLine
\MyFunction{\FunctionInitialization{}}{
$\Queue\gets $ an empty queue over $V^*$\label{line:init_begin}\tcp*[f]{A worklist over edges for computing the fixpoint}\\
$\Edges\gets$ a map $V\times \Close(\Alphabet)\to V^*$ as a linked list\tcp*[f]{The outgoing edges of each node}\\
$\DisjointSets\gets$ a disjoint-sets data structure over $V$\\
\ForEach{$u\in V$}{
$\DisjointSets.\MakeSet(u)$\tcp*[f]{constructs the singleton set $\{u\}$}\\
\ForEach{label $\CloseParenthesis\in \Close(\Alphabet)$}{
$\Edges[u][\CloseParenthesis]\gets (v:~(u,v,\CloseParenthesis)\in E)$\\
\lIf(\tcp*[f]{The two edges must merge their endpoints}){$|\Edges[u][\CloseParenthesis]|\geq 2$}{
Insert $(u, \CloseParenthesis)$ in $\Queue$\label{line:queue_insert1}
}
}
}
}
\BlankLine
\MyFunction{\FunctionFixpoint{}}{
\While(\tcp*[f]{Repeat until fixpoint}){$\Queue$ is not empty}{\label{line:loop_outer}
Extract $(u, \CloseParenthesis)$ from $\Queue$\label{line:extract}\tcp*[f]{We merge components connected via $u$'s component}\\
\uIf(\tcp*[f]{$u$ is the root of its component}){$u=\DisjointSets.\Find(u)$}{\label{line:if_parent}
$S\gets \{\DisjointSets.\Find(w):w\in \Edges[u][\CloseParenthesis]\}$\label{line:construct_s}\tcp*[f]{The components to be merged}\\
\eIf{$|S|\geq 2$}{\label{line:if_two}
$x\gets $ some arbitrary element of $S\setminus\{u\}$\label{line:choose_x}\tcp*[f]{$x$ will be the root of the merged component}\\
Make $\DisjointSets.\Union(S,x)$\label{line:update_scc}\tcp*[f]{All components merge to $x$'s component}\\
\ForEach(\tcp*[f]{Process each label separately}){label $\CloseParenthesisBeta\in \Close(\Alphabet)$}{\label{line:loop_append}
\ForEach{$v\in S\setminus\{x\}$}{\label{line:loop_inner}
\eIf(\tcp*[f]{$(v, \CloseParenthesisBeta)$ is not the $(u, \CloseParenthesis)$ we extracted from $\Queue$}){$u\neq v$ or $\CloseParenthesis\neq \CloseParenthesisBeta$}{\label{line:if_move}
Move $\Edges[v][\CloseParenthesisBeta]$ to $\Edges[x][\CloseParenthesisBeta]$\label{line:append}\tcp*[f]{In $O(1)$ time, using pointers}\\
}
{
Append $(x)$ to $\Edges[x][\CloseParenthesisBeta]$\label{line:append2}\tcp*[f]{All $\Edges[u][\CloseParenthesis]$ merged to a self-loop $x\DTo{\CloseParenthesis}x$}
}
}
\lIf(\tcp*[f]{More components to be merged}){$|\Edges[x][\CloseParenthesisBeta]|\geq 2$}{
Insert $(x, \CloseParenthesisBeta)$ in $\Queue$\label{line:queue_insert2}
}
}
}
{
Let $x\gets$ the single node in $S$
}
\lIf{$u\not \in S$ or $|S|=1$}{\label{line:if_update}
$\Edges[u][\CloseParenthesis]\gets (x)$\label{line:new_edges}\tcp*[f]{$u$'s component now points to $x$'s component}
}
}
}
}
\end{algorithm}

\subsection{Main Concepts and Intuition}\label{subsec:intuition}
Having outlined the basic algorithm for offline reachability, we now turn our attention to the data structure for dynamic reachability.
In high level, our data structure will use the same components as $\OfflineAlgo$, in particular, the representation of DSCCs using the $\DisjointSets$ data structure, the $\Edges$ linked lists for storing the outgoing edges of each component, and the $\Fixpoint()$ function for computing the DSCCs as a fixpoint after every graph update.
However, instead of starting the computation from scratch, $\DisjointSets$ and $\Edges$ (as well as the worklist $\Queue$ of $\Fixpoint()$) will be at a state that represent some DSCCs that are guaranteed to exist in the updated graph, and thus need not be recomputed.
The key technical challenge we have to solve is in computing a non-trivial such state, and doing so efficiently.

In this section we introduce the main concepts of our dynamic algorithm and provide the necessary intuition around them.
The precise algorithm is presented in the next section.

\Paragraph{High-level description of the dynamic algorithm.}
The main components that are new to our dynamic algorithm compared to the offline algorithm of \cref{subsec:offline} are geared towards handling edge deletions.
Assume that we have computed the DSCCs of a graph $G$, and we now have to process an operation $\DeleteEdge(u,v,\CloseParenthesis)$.
This will, in general, result in splitting some DSCCs into smaller ones.
In high level, we handle such edge deletion in three steps.
\begin{compactenum}
\item\label{item:intuition_step1} We compute a sound overapproximation of the DSCCs that are affected by the edge deletion, i.e., those that might have to be split into smaller components.
We compute this overapproximation by effectively performing a forward search starting from $\DSCC(v)$ and repeatedly proceeding to neighboring DSCCs of $G$.
In particular given a current DSCC $S$, we proceed to those DSCCs $S'$ that have at least two incoming edges of the form $x\DTo{\CloseParenthesisBeta}y$,
with $x\in S$ and $y\in S'$, and some label $\CloseParenthesisBeta$.
This is because $S'$ might have been formed through $S$ and the corresponding edges $x\DTo{\CloseParenthesisBeta}y$.
Although this traversal might end up touching $\Theta(n)$ DSCCs, we expect that in practice it will perform much better,
as the effect of an average edge deletion is usually very local.
\item\label{item:intuition_step2} At this point, it would suffice to split all nodes $t$ in the previously discovered DSCCs $S'$ into singleton components,
gather in the worklist $\Queue$ (see \cref{algo:offlinealgo}) 
all pairs $(s,\CloseParenthesisBeta)$ corresponding to 
incoming edges $s\DTo{\CloseParenthesisBeta}t$ where $t$ is the unique node of such a singleton component, 
and restart the fixpoint computation using  $\Fixpoint()$. 
Unfortunately, this would require $\Theta(n^2)$ time, which is beyond our target bound, as there can be $\Theta(n^2)$ such edges $s\DTo{\CloseParenthesisBeta}t$, and \cref{line:construct_s} of \cref{algo:offlinealgo} would iterate over all of them (where $u=s$ and $w=t$ in the algorithm).  
To circumvent this difficulty, we introduce the novel notion of \emph{primary DSCCs} (or \emph{PDSCCs}),
which, intuitively, are connected components of small size.
We keep track of the formation of PDSCCs dynamically by leveraging techniques from undirected dynamic reachability (see e.g., \cref{tab:connectivity_data_structures}).
Now, instead of splitting each previously discovered DSCC $S'$ into \emph{singleton components}, we split it into its \emph{PDSCCs}.
This allows us to re-initiate the fixpoint computation of the function $\Fixpoint()$ after only traversing $O(n)$ edges of the form $s\DTo{\CloseParenthesisBeta}t$ (where $t$ now is a node in a PDSCC).
In practice and due to the previous step, we typically traverse much fewer edges than $n$.
In turn, this implies that $\Fixpoint()$ will converge after $O(n)$ iterations, leading to the desired time bound.
\item We execute the fixpoint computation, similarly to \cref{algo:offlinealgo}.
\end{compactenum}

We now describe the above concepts in detail.

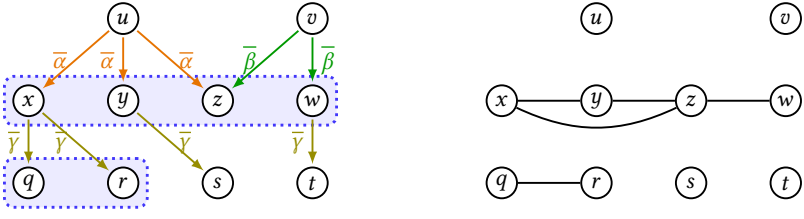
\begin{figure}
\scalebox{0.9}{%
\begin{tikzpicture}[thick, >=latex, node distance=0.3cm and 1cm,
pre/.style={<-,shorten >= 1pt, shorten <=1pt,},
post/.style={->,shorten >= 1pt, shorten <=1pt,},
und/.style={very thick, draw=gray},
node/.style={circle, minimum size=4.5mm, draw=black!100, fill=white!100, thick, inner sep=0},
virt/.style={circle,draw=black!50,fill=black!20, opacity=0}]

\newcommand{\xdisposition}{7}
\newcommand{\ydisposition}{0}
\newcommand{\xstep}{1.4}
\newcommand{\ystep}{1.2}
\def\bend{20}

\begin{scope}[shift={(0*\xdisposition,0)}]
\node	[node]		(u)	at (0*\xstep,0*\ystep) {$u$};
\node	[node]		(v)	at (2*\xstep,0*\ystep) {$v$};

\node	[node]		(x)	at (-1*\xstep,-1*\ystep) {$x$};
\node	[node]		(y)	at (0*\xstep,-1*\ystep) {$y$};
\node	[node]		(z)	at (1*\xstep,-1*\ystep) {$z$};
\node	[node]		(w)	at (2*\xstep,-1*\ystep) {$w$};

\node	[node]		(q)	at (-1*\xstep,-2*\ystep) {$q$};
\node	[node]		(r)	at (0*\xstep,-2*\ystep) {$r$};
\node	[node]		(s)	at (1*\xstep,-2*\ystep) {$s$};
\node	[node]		(t)	at (2*\xstep,-2*\ystep) {$t$};

\draw[post, color=\ColorAlpha] (u) to node[left, ] {$\CloseParenthesis$} (x);
\draw[post, color=\ColorAlpha] (u) to node[left, ] {$\CloseParenthesis$} (y);
\draw[post, color=\ColorAlpha] (u) to node[right, ] {$\CloseParenthesis$} (z);
\draw[post, color=\ColorBeta] (v) to node[left, ] {$\CloseParenthesisBeta$} (z);
\draw[post, color=\ColorBeta] (v) to node[right, ] {$\CloseParenthesisBeta$} (w);

\draw[post, color=\ColorGamma] (x) to node[left, ] {$\CloseParenthesisGamma$} (q);
\draw[post, color=\ColorGamma] (x) to node[left, ] {$\CloseParenthesisGamma$} (r);
\draw[post, color=\ColorGamma] (y) to node[right, ] {$\CloseParenthesisGamma$} (s);
\draw[post, color=\ColorGamma] (w) to node[left, ] {$\CloseParenthesisGamma$} (t);

\begin{pgfonlayer}{bg}
\node[box, very thick, rounded corners, draw=mybluecolor, dotted, fill=mybluecolor!10, fit=(x)(y)(z)(w)] (S1) {};
\node[box, very thick, rounded corners, draw=mybluecolor, dotted, fill=mybluecolor!10, fit=(q)(r)] (S2) {};
\end{pgfonlayer}

\end{scope}

\begin{scope}[shift={(1*\xdisposition,0)}]

\begin{scope}[shift={(0*\xdisposition,0)}]
\node	[node]		(u)	at (0*\xstep,0*\ystep) {$u$};
\node	[node]		(v)	at (2*\xstep,0*\ystep) {$v$};

\node	[node]		(x)	at (-1*\xstep,-1*\ystep) {$x$};
\node	[node]		(y)	at (0*\xstep,-1*\ystep) {$y$};
\node	[node]		(z)	at (1*\xstep,-1*\ystep) {$z$};
\node	[node]		(w)	at (2*\xstep,-1*\ystep) {$w$};

\node	[node]		(q)	at (-1*\xstep,-2*\ystep) {$q$};
\node	[node]		(r)	at (0*\xstep,-2*\ystep) {$r$};
\node	[node]		(s)	at (1*\xstep,-2*\ystep) {$s$};
\node	[node]		(t)	at (2*\xstep,-2*\ystep) {$t$};

\draw[-, thick] (x) to (y);
\draw[-, thick] (y) to (z);
\draw[-, thick, bend right=25] (x) to (z);
\draw[-, thick] (z) to (w);

\draw[-, thick] (q) to (r);

\end{scope}

\end{scope}

\end{tikzpicture}
}
\caption{
A bidirected graph $G$ (left) and the corresponding primal graph $H$ (right).
There are two non-singleton PDSCCs in $G$ (marked), corresponding to non-singleton connected subgraphs of $H$.
\label{fig:primal_graph}
}
\end{figure}
\Paragraph{Primal graphs and primary DSCCs.}
Consider a bidirected graph $G=(V,E)$.
The \emph{primal} graph $H=(V,L)$ is an unlabeled, undirected graph with the same node set, and edge set defined as:
\[
L=\{(x,y)\colon \exists u\in V.~\exists \CloseParenthesis\in\SetClosedParenthesis.~u\DTo{\CloseParenthesis}x, u\DTo{\CloseParenthesis}y\in E \}
\]
In words, $x$ and $y$ share an edge in $H$ if they are connected in $G$ via a Dyck path of length $2$.
A \emph{primary DSCC (PDSCC)} of $G$ is a (maximal) connected component of the primal graph $H$.
It is not hard to see that the PDSCC partitioning of $G$ is a refinement of its DSCC partitioning, i.e., each DSCC contains one or more PDSCCs.
See \cref{fig:primal_graph} for an illustration.
Similarly to the case of DSCCs, given a node $u$, we let $\PDSCC(u)$ be the PDSCC in which $u$ belongs,
while $\PDSCCRep(u)$ returns a representative of $\PDSCC(u)$ that is common for all nodes $v\in \PDSCC(u)$.
Hence, two nodes $u$, $w$ are in the same PDSCC iff $\PDSCCRep(u)=\PDSCCRep(v)$.

Since the PDSCCs of $G$ have a direct representation as connected components in the undirected graph $H$, we can leverage existing techniques from dynamic undirected connectivity (see, e.g., \cref{tab:connectivity_data_structures}) to maintain them efficiently.
PDSCCs serve the following function.
When an edge is deleted, some DSCCs $S$ of $G$ have to be split to smaller DSCCs.
This can lead to effectively repeating the fixpoint computation from scratch in $G$.
However, this approach would require $\Theta(n^2)$ running time, which is beyond our complexity bound.
Instead, maintaining the PDSCCs allows us to split $S$ to its PDSCCs (as opposed to individual nodes), and thus avoid recomputing the PDSCCs from scratch.
In turn, this allows us to achieve the desired $O(n\cdot \alpha(n))$ bound for edge deletions.

\begin{figure}
\scalebox{0.9}{%
\begin{tikzpicture}[thick, >=latex, node distance=0.3cm and 1cm,
pre/.style={<-,shorten >= 1pt, shorten <=1pt,},
post/.style={->,shorten >= 1pt, shorten <=1pt,},
und/.style={very thick, draw=gray},
node/.style={circle, minimum size=4.5mm, draw=black!100, fill=white!100, thick, inner sep=0},
virt/.style={circle,draw=black!50,fill=black!20, opacity=0}]

\newcommand{\xdisposition}{6}
\newcommand{\ydisposition}{2.5}
\newcommand{\xstep}{1}
\newcommand{\ystep}{1.1}
\def\bend{20}
\def\pad{0.19}

\begin{scope}[shift={(0*\xdisposition,0*\ydisposition)}]

\draw[very thick, ->] (0.3*\xdisposition,-0.5*\ystep) to node[above] {$\InsertEdge(u,x,\CloseParenthesis)$} (0.7*\xdisposition,-0.5*\ystep);

\node[] at (-1.2*\xstep,0.1*\ystep) {$G_1$};

\node	[node]		(u)	at (0*\xstep,0*\ystep) {$u$};

\node	[node]		(x)	at (-1*\xstep,-1*\ystep) {$x$};
\node	[node]		(y)	at (0*\xstep,-1*\ystep) {$y$};
\node	[node]		(z)	at (1*\xstep,-1*\ystep) {$z$};

\draw[post] (u) to node[left, ] {$\CloseParenthesis$} (y);
\draw[post] (u) to node[right, ] {$\CloseParenthesis$} (z);

\begin{pgfonlayer}{bg}
\node[box, very thick, rounded corners, draw=mybluecolor, dotted, fill=mybluecolor!10, fit=(y)(z)] (S1) {};
\end{pgfonlayer}

\end{scope}

\begin{scope}[shift={(0*\xdisposition,-1*\ydisposition)}]

\node[] at (-1.2*\xstep,0.1*\ystep) {$H_1$};

\node	[node]		(u)	at (0*\xstep,0*\ystep) {$u$};

\node	[node]		(x)	at (-1*\xstep,-1*\ystep) {$x$};
\node	[node]		(y)	at (0*\xstep,-1*\ystep) {$y$};
\node	[node]		(z)	at (1*\xstep,-1*\ystep) {$z$};

\draw[-, thick] (y) to (z);

\end{scope}

\begin{scope}[shift={(1*\xdisposition,0*\ydisposition)}]

\draw[very thick, ->] (0.3*\xdisposition,-0.5*\ystep) to node[above] {$\DeleteEdge(u,y,\CloseParenthesis)$} (0.7*\xdisposition,-0.5*\ystep);

\node[] at (-1.2*\xstep,0.1*\ystep) {$G_2$};

\node	[node]		(u)	at (0*\xstep,0*\ystep) {$u$};

\node	[node]		(x)	at (-1*\xstep,-1*\ystep) {$x$};
\node	[node]		(y)	at (0*\xstep,-1*\ystep) {$y$};
\node	[node]		(z)	at (1*\xstep,-1*\ystep) {$z$};

\draw[post] (u) to node[left, ] {$\CloseParenthesis$} (x);
\draw[post] (u) to node[left, ] {$\CloseParenthesis$} (y);
\draw[post] (u) to node[right, ] {$\CloseParenthesis$} (z);

\begin{pgfonlayer}{bg}
\node[box, very thick, rounded corners, draw=mybluecolor, dotted, fill=mybluecolor!10, fit=(x)(y)(z)] (S1) {};
\end{pgfonlayer}

\end{scope}

\begin{scope}[shift={(1*\xdisposition,-1*\ydisposition)}]

\node[] at (-1.2*\xstep,0.1*\ystep) {$H_2$};

\node	[node]		(u)	at (0*\xstep,0*\ystep) {$u$};

\node	[node]		(x)	at (-1*\xstep,-1*\ystep) {$x$};
\node	[node]		(y)	at (0*\xstep,-1*\ystep) {$y$};
\node	[node]		(z)	at (1*\xstep,-1*\ystep) {$z$};

\draw[-, thick] (x) to (y);
\draw[-, thick] (y) to (z);
\draw[-, thick, bend right=25, dashed, draw=gray] (x) to (z);

\end{scope}

\begin{scope}[shift={(2*\xdisposition,0*\ydisposition)}]

\node[] at (-1.2*\xstep,0.1*\ystep) {$G_3$};

\node	[node]		(u)	at (0*\xstep,0*\ystep) {$u$};

\node	[node]		(x)	at (-1*\xstep,-1*\ystep) {$x$};
\node	[node]		(y)	at (0*\xstep,-1*\ystep) {$y$};
\node	[node]		(z)	at (1*\xstep,-1*\ystep) {$z$};

\draw[post] (u) to node[left, ] {$\CloseParenthesis$} (x);
\draw[post] (u) to node[right, ] {$\CloseParenthesis$} (z);

\begin{pgfonlayer}{bg}
\draw[-, very thick, rounded corners, draw=mybluecolor, dotted, fill=mybluecolor!10,] ($ (x) + (-2*\pad, 2*\pad) $) to ($ (x) + (-2*\pad, -2*\pad) $) to ($ (z) + (2*\pad, -2*\pad)$) to ($ (z) + (2*\pad, 2*\pad)$) to ($ (z) + (-2*\pad, 2*\pad)$) to ($ (z) + (-2*\pad, -1.3*\pad)$) to ($ (x) + (2*\pad, -1.3*\pad)$) to ($ (x) + (2*\pad, 2*\pad)$) --cycle;
\end{pgfonlayer}

\end{scope}

\begin{scope}[shift={(2*\xdisposition,-1*\ydisposition)}]

\node[] at (-1.2*\xstep,0.1*\ystep) {$H_3$};

\node	[node]		(u)	at (0*\xstep,0*\ystep) {$u$};

\node	[node]		(x)	at (-1*\xstep,-1*\ystep) {$x$};
\node	[node]		(y)	at (0*\xstep,-1*\ystep) {$y$};
\node	[node]		(z)	at (1*\xstep,-1*\ystep) {$z$};

\draw[-, thick, bend right=25] (x) to (z);

\end{scope}

\end{tikzpicture}
}
\caption{
Sparsification for the maintenance of the PDSCCs of $G_i$ across edge insertions and deletions (top),
by maintaining connected components in the corresponding primal graphs $H_i$ (bottom).
\label{fig:pdsccs}
}
\end{figure}
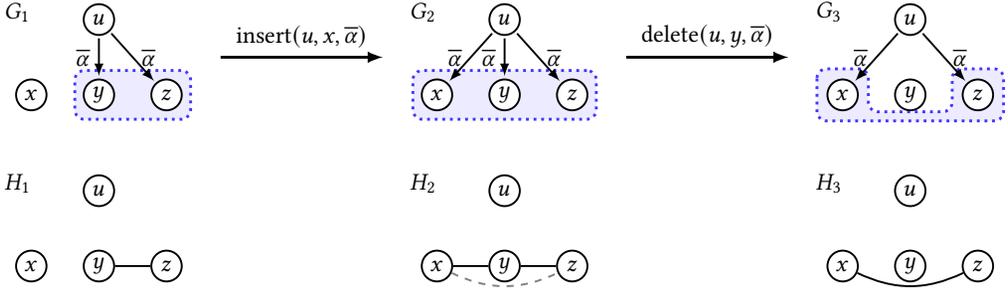
\Paragraph{A sparsification approach for maintaining the PDSCCs.}
Since each PDSCC corresponds to a connected component of the undirected primal graph $H$, we will maintain PDSCCs dynamically, by using any data structure for dynamic reachability on undirected graphs (see \cref{tab:connectivity_data_structures}).
We call this data structure for undirected reachability $\PrimCompDS$.

An operation $o=\InsertEdge(u,v,\ov{\alpha})$ inserts between $0$ and $n-1$ undirected edges in the primal graph $H$.
Indeed, we have an edge $(x,v)$ for every node $x\neq v$ for which already $u\DTo{\ov{\alpha}}x$,
and there may be $n-1$ such nodes $x$.
Note, however, that if there already exist two such nodes, $x,y$, with $u\DTo{\ov{\alpha}}x$ and $u\DTo{\ov{\alpha}}y$, then $x$ and $y$ already appear in the same PDSCC of $G$.
Hence, we can faithfully maintain the PDSCCs of $G$ by only adding \emph{one} of the two primal edges $(v,x)$ and $(v,y)$ in $\PrimCompDS$.
This has the desirable effect of maintaining PDSCCs after edge insertions by inserting $O(1)$ undirected edges in $\PrimCompDS$, as opposed to $O(n)$.
This kind of technique is known as \emph{sparsification}~\cite{Eppstein1997}, in that we manage to represent the full connectivity of the primal graph $H$ while only storing a subset of its edges.
In particular, we achieve this effect by maintaining the outgoing edges of $u$ labeled with $\ov{\alpha}$ in a linked list $\OutEdges[u][\CloseParenthesis]=(x_1,\dots, x_r)$, and only inserting in $\PrimCompDS$ edges $(x_{i}, x_{i+1})$, i.e., between consecutive nodes in $\OutEdges[u][\CloseParenthesis]$.
Distant nodes become connected transitively in $\PrimCompDS$.

Conversely, an operation $o=\DeleteEdge(u,v,\ov{\alpha})$ deletes between $0$ and $n-1$ undirected edges in the primal graph $H$, with reasoning similar to the one above.
However, because of the above invariant, we can restore the PDSCCs of $G$ by removing the edges $(a,v)$ and $(v,b)$ in $\PrimCompDS$, where $a$ and $b$ are the neighbors of $v$ in $\OutEdges[u][\CloseParenthesis]$. 
Moreover, as $a$ and $b$ are no longer connected in $\PrimCompDS$, we restore their connectivity by inserting the edge $(a,b)$ in $\PrimCompDS$.
\cref{fig:pdsccs} illustrates this sparsification technique on a small example. The update $\InsertEdge(u,x,\CloseParenthesis)$ creates two edges in the primal graph $H_2$, namely $(x,y)$ and $(x,z)$.
However, the dashed edge $(x,z)$ is not stored explicitly in $\PrimCompDS$, as $x$ and $y$ are already connected via $z$.
The update $\DeleteEdge(u,y,\CloseParenthesis)$ removes the edges $(x,y)$ and $(y,z)$ from the primal graph $H_3$.  
To preserve the connectivity of $x$ and $z$, the dashed edge $(x,z)$ is restored in $\PrimCompDS$.

\Paragraph{Maintaining the neighbors of PDSCCs.}
The maintenance of the PDSCCs allows us to efficiently identify the incoming neighbors of a PDSCC, in $O(n)$ worst-case time, although the graph might have $\Theta(n^2)$ edges.
This fact stems from the following observation.
For a given node $u$ and label $\CloseParenthesis$, if there are two edges $u\DTo{\CloseParenthesis}x$ and $u\DTo{\CloseParenthesis}y$, then $x$ and $y$ are in the same PDSCC.
In other words, for every such $u$ and $\CloseParenthesis$, there is at most one PDSCC that $u$ is a neighbor of via a $\CloseParenthesis$-labeled edge.
Thus, given a set of PDSCCs (as described in \cref{item:intuition_step2} above), we can obtain all edges incoming to them in $O(n)$ time, by iterating over all nodes $u$ and labels $\CloseParenthesis$.
Though this is sufficient towards our linear time bound, it can become unnecessarily slow when the number of such PDSCCs is small.
Instead, we follow a different approach, which retains the $O(n)$ worst-case behavior but is faster in practice.

For every node $x$ and label $\CloseParenthesis$, we maintain a  set $\InEdges[x][\CloseParenthesis]$, which stores nodes $u$ that have an edge $u\DTo{\CloseParenthesis}x$.
We maintain the invariant that $x$ will be the \emph{only} node in $\PDSCC(x)$ with $u\in \InEdges[x][\CloseParenthesis]$, even though there might be other nodes $y\in\PDSCC(x)$ also having an edge $u\DTo{\CloseParenthesis}y$.
To identify the neighbors of each PDSCC, it now suffices to iterate over its nodes $x$, and for each label $\CloseParenthesis$, collect the nodes from $\InEdges[x][\CloseParenthesis]$.
The invariant guarantees that every neighbor $u$ will be accessed exactly once per label $\CloseParenthesis$, retaining the $O(n)$ worst-case running time.

The sets $\InEdges[x][\CloseParenthesis]$, together with their invariant, can be maintained as follows.
Upon inserting an edge $\InsertEdge(u,x,\CloseParenthesis)$, if $u$ does not have any other outgoing edges labeled with $\CloseParenthesis$,
we insert $u$ in $\InEdges[x][\CloseParenthesis]$.
Otherwise, there exists another edge $u\DTo{\CloseParenthesis}y$, with already $u\in \InEdges[y][\CloseParenthesis]$.
The existence of the two edges $u\DTo{\CloseParenthesis}x$ and $u\DTo{\CloseParenthesis}y$ implies that $\PDSCC(x)=\PDSCC(y)$, 
hence it is sound to not insert $u$ in $\InEdges[x][\CloseParenthesis]$, as $u$ is retrievable via $y$.
Similarly, upon deleting an edge $\DeleteEdge(u,x,\CloseParenthesis)$, if $u\in \InEdges[x][\CloseParenthesis]$,
we move $u$ to another set $\InEdges[y][\CloseParenthesis]$ for which there exists an edge $u\DTo{\CloseParenthesis}y$.  \cref{fig:pdsccs_neighbors} illustrates the maintenance of $\InEdges$ on a small example.

\begin{figure}
\scalebox{0.93}{%
\begin{tikzpicture}[thick, >=latex, node distance=0.3cm and 1cm,
pre/.style={<-,shorten >= 1pt, shorten <=1pt,},
post/.style={->,shorten >= 1pt, shorten <=1pt,},
und/.style={very thick, draw=gray},
node/.style={circle, minimum size=4.5mm, draw=black!100, fill=white!100, thick, inner sep=0},
virt/.style={circle,draw=black!50,fill=black!20, opacity=0}]

\newcommand{\xdisposition}{4.15}
\newcommand{\ydisposition}{3}
\newcommand{\xstep}{0.75}
\newcommand{\ystep}{1.3}
\def\bend{50}
\def\pad{0.19}

\begin{scope}[shift={(0*\xdisposition,0*\ydisposition)}]

\draw[very thick, ->] (0.325*\xdisposition,-0.5*\ystep) to node[above] {\small $\InsertEdge(u,x,\CloseParenthesis)$} (0.675*\xdisposition,-0.5*\ystep);

\node[] at (-1.2*\xstep,0.1*\ystep) {$G_1$};

\node	[node]		(u)	at (0*\xstep,0*\ystep) {$u$};

\node	[node]		(x)	at (-1*\xstep,-1*\ystep) {$x$};
\node	[node]		(y)	at (0*\xstep,-1*\ystep) {$y$};
\node	[node]		(z)	at (1*\xstep,-1*\ystep) {$z$};

\begin{pgfonlayer}{bg}
\end{pgfonlayer}

\end{scope}

\begin{scope}[shift={(1*\xdisposition,0*\ydisposition)}]

\draw[very thick, ->] (0.325*\xdisposition,-0.5*\ystep) to node[above] {\small $\InsertEdge(u,y,\CloseParenthesis)$} node[below] {\small $\InsertEdge(u,z,\CloseParenthesis)$} (0.675*\xdisposition,-0.5*\ystep);

\node[] at (-1.2*\xstep,0.1*\ystep) {$G_2$};

\node	[node]		(u)	at (0*\xstep,0*\ystep) {$u$};

\node	[node]		(x)	at (-1*\xstep,-1*\ystep) {$x$};
\node	[node]		(y)	at (0*\xstep,-1*\ystep) {$y$};
\node	[node]		(z)	at (1*\xstep,-1*\ystep) {$z$};

\draw[post] (u) to node[left]{$\CloseParenthesis$}(x);
\draw[post, dashed, color=\darkorange, bend left=\bend] (x) to node[left] {$\CloseParenthesis$} (u);

\begin{pgfonlayer}{bg}
\end{pgfonlayer}

\end{scope}

\begin{scope}[shift={(2*\xdisposition,0*\ydisposition)}]

\draw[very thick, ->] (0.325*\xdisposition,-0.5*\ystep) to node[above] {\small $\DeleteEdge(u,x,\CloseParenthesis)$}  (0.675*\xdisposition,-0.5*\ystep);

\node[] at (-1.2*\xstep,0.1*\ystep) {$G_3$};

\node	[node]		(u)	at (0*\xstep,0*\ystep) {$u$};

\node	[node]		(x)	at (-1*\xstep,-1*\ystep) {$x$};
\node	[node]		(y)	at (0*\xstep,-1*\ystep) {$y$};
\node	[node]		(z)	at (1*\xstep,-1*\ystep) {$z$};

\draw[post] (u) to node[left]{$\CloseParenthesis$}(x);
\draw[post] (u) to node[left]{$\CloseParenthesis$}(y);
\draw[post] (u) to node[left]{$\CloseParenthesis$}(z);

\draw[post, dashed, color=\darkorange, bend left=\bend] (x) to node[left] {$\CloseParenthesis$} (u);

\begin{pgfonlayer}{bg}
\node[box, very thick, rounded corners, draw=mybluecolor, dotted, fill=mybluecolor!10, fit=(x)(y)(z)] (S1) {};
\end{pgfonlayer}

\end{scope}

\begin{scope}[shift={(3*\xdisposition,0*\ydisposition)}]

\node[] at (-1.2*\xstep,0.1*\ystep) {$G_4$};

\node	[node]		(u)	at (0*\xstep,0*\ystep) {$u$};

\node	[node]		(x)	at (-1*\xstep,-1*\ystep) {$x$};
\node	[node]		(y)	at (0*\xstep,-1*\ystep) {$y$};
\node	[node]		(z)	at (1*\xstep,-1*\ystep) {$z$};

\draw[post] (u) to node[left]{$\CloseParenthesis$}(y);
\draw[post] (u) to node[left]{$\CloseParenthesis$}(z);

\draw[post, dashed, color=\darkorange, bend left=\bend] (y) to node[left] {$\CloseParenthesis$} (u);

\begin{pgfonlayer}{bg}
\node[box, very thick, rounded corners, draw=mybluecolor, dotted, fill=mybluecolor!10, fit=(y)(z)] (S1) {};
\end{pgfonlayer}

\end{scope}

\end{tikzpicture}
}
\caption{
Maintenance of the sets $\InEdges$ (shown in orange dashed) along edge insertions and deletions.
The first edge insertion $u\DTo{\CloseParenthesis}x$ 
leads to $u\in \InEdges[x][\CloseParenthesis]$.
The following two edge insertions $u\DTo{\CloseParenthesis}y$ and $u\DTo{\CloseParenthesis}z$ do not modify $\InEdges$, as $x$, $y$ and $z$ belong to the same PDSCC,
thus $u$ can be retrieved as a $\CloseParenthesis$-neighbor via the set $\InEdges[x][\CloseParenthesis]$.
When processing deleting the edge $u\DTo{\CloseParenthesis}x$, we move $u$ to $\InEdges[y][\CloseParenthesis]$, thus $u$ can still be retrieved as a $\CloseParenthesis$-neighbor of the PDSCC $\{y,z\}$.
\label{fig:pdsccs_neighbors}
}
\end{figure}

\Paragraph{Processing a delete operation.}
We are now in position to outline our approach for processing a delete operation $\DeleteEdge(u,v,\CloseParenthesis)$.
As the edge $u\DTo{\CloseParenthesis}v$ might be used to connect $v$ to other nodes in $\DSCC(v)$, this operation might split $\DSCC(v)$ into smaller DSCCs. 
As nodes become disconnected in $\DSCC(v)$, other DSCCs might also be split, if the paths connecting their nodes passed through $\DSCC(v)$, resulting in a cascading effect. 
For instance, consider \cref{fig:illustrative_example} (middle), having two $\DSCC$s $\{c,d,e,f\}$ and $\{g,h\}$. 
On deleting the edge $f\DTo{\overline{L}}d$, the node $d$ splits from the $\DSCC$ $\{c,d,e,f\}$. 
In turn, the $\DSCC$ $\{g,h\}$ splits into $\{g\},\{h\}$ since they were held together by $c,d$ being in the same $\DSCC$. 
The split of $\{g,h\}$ results in pulling out $f$ from $\DSCC(e)$. Thus, deleting the edge $f\DTo{\overline{L}}d$ results in splitting $\DSCC$s $\{c,d,e,f\}$ and $\{g,h\}$ resulting 
in $\DSCC$s $\{c,e\}$, $\{d\}, \{f\}, \{g\}, \{h\}$.

Thus, a single edge deletion may have an effect which propagates to the whole graph; however, we can obtain a sound overapproximation of the DSCCs affected by the initial edge deletion using the following observation.
If the split of a DSCC $S$ leads to the immediate split of another DSCC $S'$, then there exist two distinct edges $x_i\DTo{\CloseParenthesisBeta}y_i$, for $i\in [2]$, such that $x_i\in S$ and $y_i\in S'$.
Thus, for $\DeleteEdge(u,v,\CloseParenthesis)$, we can obtain our overapproximation by starting a forward search from $\DSCC(v)$:~given a current DSCC $S$, we iterate over all labels $\CloseParenthesisBeta$ such that $S$ has at least two $\CloseParenthesisBeta$-labeled outgoing edges, and proceed to the resulting DSCC $S'$ (note that there can be at most one such $S'$).

After the overapproximation of the set of potentially splitting DSCCs has been performed, each such DSCC $S$ is split to its constituent PDSCCs, using the undirected connectivity data structure $\PrimCompDS$.
Note that this splitting is an underapproximation of the set of the final DSCCs, as some PDSCCs have to be merged again.
We will discover this by running the fixpoint computation again, up to convergence.
For this, we re-insert in the worklist $\Queue$ (see function $\Fixpoint()$ in \cref{algo:offlinealgo}) node-label pairs $(\DSCCRep(y), \CloseParenthesisBeta)$ that represent edges $y\DTo{\CloseParenthesisBeta}\cdot$ that might lead to further component merging. Due to the efficient maintenance of neighbors of PDSCCs (see the previous paragraph), this can be achieved by iterating over all nodes $x$ of the PDSCCs, and for each label $\CloseParenthesisBeta\in\SetClosedParenthesis$, insert in $\Queue$, the pair 
$(\DSCCRep(w), \CloseParenthesisBeta)$
corresponding to $w\in \InEdges[x][\CloseParenthesisBeta]$  if 
$|\Edges[\DSCCRep(w)][\CloseParenthesisBeta]| \geq 2$. 
Finally we continue with the fixpoint computation as in $\OfflineAlgo$ with the guarantee that upon convergence, $\DisjointSets$  will contain the DSCCs of the new graph after deleting the edge $u\DTo{\CloseParenthesis}v$.

\input{figures/delete_overview}
\SubParagraph{Example.}
\cref{fig:delete_overview} illustrates the above process on a small example, processing an update $\DeleteEdge(c, d ,\CloseParenthesisGamma)$.
Before the deletion (top), the graph has four DSCCs, namely $S_1, S_2, S_3, S_4$ (marked).
In the first step, we perform a forward search from $S_2$, following pairs of same-label edges, which discovers $S_2$, $S_3$ and $S_4$ as the set of DSCCs that are potentially affected by the delete operation.
Observe that $S_1$ does not have to be split, and the algorithm avoids recomputation on $S_1$.
Then, $S_2$, $S_3$ and $S_4$ are split to their constituent PDSCCs (middle).
Observe that $S_4$ is only partly split, as its subset $\{u,v,y\}$ is a PDSCC (with $y=\PDSCCRep(u)$) in the new graph (after deletion). 
Hence the algorithm avoids recomputing $S_4$ from scratch.
Afterwards, the algorithm uses the $\InEdges$ data structure to collect  
the set of pairs 
$
\NewIncomingRoots = \{(g,\CloseParenthesisGamma), 
(h,\CloseParenthesisGamma), 
(c,\CloseParenthesisGamma),
(d,\CloseParenthesis), 
(e,\CloseParenthesis), 
(d,\CloseParenthesisBeta), 
(y,\CloseParenthesis), 
(y,\CloseParenthesisBeta), 
(y,\CloseParenthesisGamma),
(f,\CloseParenthesisBeta)\}
$,
i.e., the pairs $(\DSCCRep(x), \CloseParenthesisBeta)$ such that $\DSCC(x) \DTo{\CloseParenthesisBeta} \PDSCC(r)$, where  
 $r$ is a node in the newly formed $\PDSCC$s.  
Observe that the pair $(f,\CloseParenthesisGamma)$ is not  collected in $\NewIncomingRoots$ even though $f\DTo{\CloseParenthesisGamma}c$, as this edge is incoming to DSCC $S_1$ which was deemed as not affected by the edge deletion.
Note that some of these $\PDSCC$s can be merged during the $\Fixpoint$ computation to obtain 
 the new set of $\DSCC$s post deletion. For this, a pair  $(x,\CloseParenthesis) \in \NewIncomingRoots$ is added to 
$\Queue$,
if there are  two outgoing edges on $\CloseParenthesis$ from  $x$'s component. 
In this example, only $(y,\CloseParenthesisGamma)$ is added in $\Queue$ since there are 
two edges on $\CloseParenthesisGamma$ from $y$'s component, namely $u\DTo{\CloseParenthesisGamma}e$ and $v\DTo{\CloseParenthesisGamma}f$. $\Fixpoint$ takes this $\Queue$ 
and converges to the final DSCCs (bottom), merging the $\PDSCC$s $\{e\}$ and $\{f\}$, thereby partially restoring $S_2$.

\subsection{The Dynamic Algorithm}\label{subsec:main_data_structure}

We now make the concepts of the previous section formal, by developing precise algorithms for handling each insert and delete operation. 
In a static analysis context, the edges of the underlying graph correspond to statements in the analyzed source code.
This means that a certain edge might be inserted several times in the graph, if, for example, it comes from a repeating program statement. 
Note, however, that only the presence of the edge impacts the reachability analysis and not its multiplicity.
To handle multiple edges, we maintain simple counters in a list data structure $\Count$, so that $\Count[u][v][\CloseParenthesis]$ holds the number of $u\DTo{\CloseParenthesis}v$ edges currently in the graph.

\smallskip
\begin{algorithm}
\small
\DontPrintSemicolon
\SetInd{0.5em}{0.5em}
\caption{$\InsertEdge(u,v,\CloseParenthesis)$}\label{algo:insert_edge}
\BlankLine
$\Count[u][v][\CloseParenthesis]\gets \Count[u][v][\CloseParenthesis] + 1$\tcp*[f]{Update the edge counter}\\
\lIf(\tcp*[f]{Edge exists already, no change in reachability}){$\Count[u][v][\CloseParenthesis]\geq 2$}
{
$\Return$
}
\eIf(\tcp*[f]{$v$ is the first $\CloseParenthesis$-neighbor of $u$}){$|\OutEdges[u][\CloseParenthesis]|= 0$}{\label{line:algo_insert_pdscc}
Insert $u$ in $\InEdges[v][\CloseParenthesis]$ \tcp*[f]{The fact that $u\DTo{\CloseParenthesis}\PDSCC(v)$ is retrievable via $\InEdges[v][\CloseParenthesis]$}\label{line:algo_insert_inedges}\\
}{
Let $y\gets$ the head of $\OutEdges[u][\CloseParenthesis]$\tcp*[f]{The most recent $\CloseParenthesis$-neighbor of $u$}\label{line:algo_insert_gethead}\\
$\PrimCompDS.\InsertEdge(v,y)$ \tcp*[f]{Update the PDSCCs with the edge $(v,y)$}\label{line:algo_insert_insert_primal_edge}\\
}
Insert $v$ as the head of $\OutEdges[u][\CloseParenthesis]$\tcp*[f]{$v$ is a node in the PDSCC formed by $u\DTo{\CloseParenthesis}\cdot$ edges}\label{line:algo_insert_outedges}\\
$x\gets\DisjointSets.\Find(u)$\tcp*[f]{$x$ is the root of the component of $u$}\label{line:algo_insert_root_of_u}\\ 
Insert $v$ in $\Edges[x][\CloseParenthesis]$\tcp*[f]{The component of $u$ has an extra outgoing edge}\label{line:algo_insert_insert_edges}\\
\uIf(\tcp*[f]{Potentially some component merging occurs, compute new fixpoint}){$|\Edges[x][\CloseParenthesis]|\geq 2$ }{
Insert $(x, \CloseParenthesis)$ in $\Queue$\tcp*[f]{Prepare the new fixpoint computation}\\
$\Fixpoint()$\tcp*[f]{Compute the new fixpoint due to the new edge}\label{line:algo_insert_fixpoint}\\
}
\end{algorithm}

\Paragraph{Operation $\InsertEdge(u,v,\ov{\alpha})$.}
\cref{algo:insert_edge} handles each $\InsertEdge(u,v,\ov{\alpha})$ operation.
In words, the algorithm first implements the logic for sparsely maintaining the PDSCCs (\cref{line:algo_insert_pdscc} to \cref{line:algo_insert_outedges}).
In particular, if $v$ is the first $\CloseParenthesis$-neighbor of $u$, then $u$ is marked as a $\CloseParenthesis$-neighbor of $\PDSCC(v)$ via $v$, by inserting $u$ in $\InEdges[v][\CloseParenthesis]$ (\cref{line:algo_insert_inedges}).
Otherwise, $u$ is already marked via some existing $\CloseParenthesis$-neighbor $y$ of $u$, which is also in $\PDSCC(v)$.
Then the algorithm inserts an edge $(v,y)$ in the $\PrimCompDS$ data structure (which maintains undirected connectivity) to reflect the change in $v$'s PDSCC.

The second step of \cref{algo:insert_edge} (\cref{line:algo_insert_root_of_u} to \cref{line:algo_insert_fixpoint}) prepares the fixpoint computation  that might be triggered due to the new edge.
For this, it marks that $v$ is a new $\CloseParenthesis$-neighbor of $\DSCC(u)$ (\cref{line:algo_insert_root_of_u}, as the DSCC is represented by its root node $x$).
If there already exists another $\CloseParenthesis$-labeled edge out of $\DSCC(u)$, then the pair $(x,\CloseParenthesis)$ is inserted in $\Queue$, and  the $\Fixpoint()$ is called to continue the fixpoint computation 
for merging the DSCCs (in $\DisjointSets$)
(note that $\Fixpoint()$ is also used by $\OfflineAlgo$ and defined in \cref{algo:offlinealgo}).

\smallskip
\begin{algorithm}
\small
\DontPrintSemicolon
\caption{$\DeleteEdge(u,v,\ov{\alpha})$}\label{algo:delete_edge}
\BlankLine
$\Count[u][v]\CloseParenthesis]\gets \Count[u][v][\CloseParenthesis] - 1$\tcp*[f]{Update the edge counter}\\
\lIf(\tcp*[f]{Edge still exists, no change in reachability}){$\Count[u][v][\CloseParenthesis]\geq 1$}{ $\Return$
}
\uIf(\tcp*[f]{$v$ was the earliest node to have an edge $u\DTo{\CloseParenthesis}\cdot$}){$v$ is the tail of $\OutEdges[u][\CloseParenthesis]$}{\label{line:algo_delete_pdscc}
Remove $u$ from $\InEdges[v][\CloseParenthesis]$\\
\uIf(\tcp*[f]{There is another $u\DTo{\CloseParenthesis}\cdot$ edge}){$|\OutEdges[u][\CloseParenthesis]|\geq 2$}{
Let $w\gets$ the penultimate node in $\OutEdges[u][\CloseParenthesis]$\label{line:penultimate}\\
Insert $u$ in $\InEdges[w][\CloseParenthesis]$\tcp*[f]{$u\DTo{\CloseParenthesis}\PDSCC(w)$ is now retrievable via $\InEdges[w][\CloseParenthesis]$}\label{line:algo_delete_inedges}
}
}
Let $a, b\gets \bot, \bot$\\
\uIf{$v$ is neither the head nor the tail of $\OutEdges[u][\CloseParenthesis]$}{
Let $a, b\gets$ the two neighbors of $v$ in $\OutEdges[u][\CloseParenthesis]$\tcp*[f]{We will have to reconnect $a$ and $b$}
}
\ForEach(\tcp*[f]{At most two neighbors}){neighbor $x$ of $v$ in $\OutEdges[u][\CloseParenthesis]$}{\label{line:algo_delete_delete_primal_edge}
$\PrimCompDS.\DeleteEdge(v,x)$\tcp*[f]{The primal graph loses the edge $(v,x)$}
}
Remove $v$ from $\OutEdges[u][\CloseParenthesis]$\tcp*[f]{$v$ no longer has undirected edges due to $\CloseParenthesis$ via $u$}\label{line:algo_delete_remove_from_outedges}\\
\uIf(\tcp*[f]{$v$ had two neighbors in $\OutEdges[u][\CloseParenthesis]$}){$a\neq \bot$ and $b\neq \bot$}{
$\PrimCompDS.\InsertEdge(a,b)$\tcp*[f]{$a$ and $b$ no longer connected via $v$, connect them directly}\label{line:algo_delete_reinsert_pdscc}
}
$\MakePrimary()$\tcp*[f]{Prepare the state for the dynamic fixpoint}\label{line:algo_delete_makeprimary}\\
$\Fixpoint()$\tcp*[f]{Compute the fixpoint from the current state}\label{line:algo_delete_fixpoint}
\end{algorithm}
\Paragraph{Operation $\DeleteEdge(u,v,\ov{\alpha})$.}
\cref{algo:delete_edge} handles each $\DeleteEdge(u,v,\ov{\alpha})$ operation.
Similarly to \cref{algo:insert_edge}, \cref{algo:delete_edge} first implements the logic for sparsely maintaining the PDSCCs (\cref{line:algo_delete_pdscc} to \cref{line:algo_delete_reinsert_pdscc}). 
In particular, if $v$ is the tail of $\OutEdges[u][\CloseParenthesis]$ then $u$ is currently retrievable as a $\CloseParenthesis$-neighbor of $\PDSCC(v)$ via $v$.
Since $u$ will no longer be a $\CloseParenthesis$-neighbor of $v$, the algorithm removes $u$ from $\InEdges[v][\CloseParenthesis]$,
and makes $u$ retrievable via $w$ (\cref{line:algo_delete_inedges}), which is the new last $\CloseParenthesis$-neighbor of $u$ (after removing $v$ in \cref{line:penultimate}). 
Then, the algorithm deletes the edge $(a,v)$ and $(b,v)$ in the $\PrimCompDS$ data structure, where $a$ and $b$ are the neighbors of $v$ in $\OutEdges[u][\CloseParenthesis]$ to reflect the change in $\PDSCC(v)$ (\cref{line:algo_delete_delete_primal_edge}).
After removing $v$ from $\OutEdges[u][\CloseParenthesis]$ (\cref{line:algo_delete_remove_from_outedges}), $a$ and $b$ become neighboring nodes in $\OutEdges[u][\CloseParenthesis]$, and thus the algorithm inserts a new edge $(a,b)$ in $\PrimCompDS$, to reflect the fact that $a$ and $b$ are now directly connected in the sparse representation of PDSCCs (\cref{line:algo_delete_reinsert_pdscc}).
As the edge deletion might lead to the splitting of some DSCCs, the second step of \cref{algo:delete_edge} splits some of the previously computed DSCCs and re-initiates the fixpoint computation (\cref{line:algo_delete_makeprimary} and \cref{line:algo_delete_fixpoint}).

\smallskip
\begin{algorithm}
\small
\DontPrintSemicolon
\caption{$\MakePrimary()$}\label{algo:make_primary}
\BlankLine
\tcp{1. A forward search from $\DSCC(u)$ to gather the DSCCs potentially affected by $\DeleteEdge(u,v,\ov{\alpha})$}
$\Stack\gets$ an empty queue over $V$\tcp*[f]{A simple queue to perform the forward search}\label{line:algo_makeprimary_emptyqueue}\\
$\AffectedDSCCs\gets$ an empty set over $V$\tcp*[f]{Stores the roots of DSCCs potentially affected by $\DeleteEdge(u,v,\ov{\alpha})$}\\
Insert $\DisjointSets.\Find(v)$ in $\AffectedDSCCs$ and $\Stack$\tcp*[f]{Mark $\DSCC(v)$ as potentially affected by the delete}\label{line:algo_makeprimary_mark_dscc_v_affected}\\
\While(\tcp*[f]{Start the forward search from $\DSCC(v)$}){$\Stack$ is not empty}{\label{line:algo_makeprimary_fwdsearchloop}
Extract $x$ from $\Stack$\tcp*[f]{Proceed the search to $\DSCC(x)$}\\
\ForEach{label $\CloseParenthesisBeta\in\SetClosedParenthesis$}{
\uIf(\tcp*[f]{Note that $\DSCC(s)=\DSCC(t)$}){exist nodes $s \neq t$ with $\DSCC(x)\DTo{\CloseParenthesisBeta}s$ and $\DSCC(x)\DTo{\CloseParenthesisBeta}t$}{\label{line:algo_makeprimary_iftwoedgestot}
$z\gets \DisjointSets.\Find(t)$\tcp*[f]{The root representative of $t$'s component}\\
\uIf(\tcp*[f]{This is the first time we encounter $\DSCC(z)$}){$z\not \in \AffectedDSCCs$}{
Insert $z$ in $\AffectedDSCCs$ and $\Stack$\tcp*[f]{Mark $\DSCC(z)$ as potentially affected by the delete}\\
}
}
}
}\label{line:algo_makeprimary_firststepend}
\BlankLine
\tcp{2. Every DSCC potentially affected is split to its PDSCCs}
$\NewRoots\gets$ an empty set over $V$\tcp*[f]{Gathers the roots of the PDSCCs that DSCCs are split to}\label{line:algo_makeprimary_newroots}\\
\ForEach{$z\in \AffectedDSCCs$}{
\ForEach(\tcp*[f]{Iterate over the nodes of the affected DSCCs}){node $t$ in the component of $z$}{\label{line:algo_makeprimary_iterate_nodes_of_breaking_dsccs}
\ForEach{label $\CloseParenthesisBeta\in\SetClosedParenthesis$}{\label{line:algo_makeprimary_iterate_labels_of_breaking_dsccs}
\ForEach(\tcp*[f]{Incoming neighbor to the affected $\DSCC(z)$}){$s\in \InEdges[t][\CloseParenthesisBeta]$}{\label{line:algo_makeprimaryiterateoverpdsccinedgesofunaffected}
$x\gets\DisjointSets.\Find(s)$\tcp*[f]{The root representative of $\DSCC(s)$}\\
\uIf(\tcp*[f]{$\DSCC(x)$ is unaffected by the delete}){$x\not \in \AffectedDSCCs$}{
Reinitialize $\Edges[x][\CloseParenthesisBeta]$ to an empty list\tcp*[f]{Will be re-populated later in \cref{line:algo_makeprimary_insert_t_in_edges_x}}\label{line:algo_makeprimary_unaffected}
}
}
}
}
Split the nodes of $z$'s component in $\DisjointSets$ to its PDSCCs, using $\PrimCompDS$\label{line:algo_makeprimary_split_to_pdsccs}\\
Insert the roots of the new components in $\NewRoots$\tcp*[f]{Some of the PDSCCs might merge again}\label{line:newroots}\\
\ForEach{label $\CloseParenthesisBeta\in\SetClosedParenthesis$}{
Reinitialize $\Edges[z][\CloseParenthesisBeta]$ to an empty list \tcp*[f]{$z$ no longer represents a component in $\DisjointSets$}\label{line:makeprimary_reinitialize_edges_of_affected_nodes}
}
}\label{line:algo_makeprimary_secondstepend}
\BlankLine
\tcp{3. Gather edges incoming to the newly initialized components for new fixpoint computation}
$\NewIncomingRoots\gets$ an empty set over $V\times \SetClosedParenthesis$ \tcp*[f]{Gathers pairs $(x,\CloseParenthesisBeta)$ for which $\DSCC(x)\DTo{\CloseParenthesisBeta}\PDSCC(r)$, for $r\in \NewRoots$} \label{line:algo_makeprimary_newincoming}\\
\ForEach(\tcp*[f]{Iterate over the roots of the newly split components}){$r\in \NewRoots$}{\label{line:algo_makeprimary_iterateoverpdsccs}
\ForEach(\tcp*[f]{Iterate over the nodes of each newly split component}){node $t$ in the component of $r$}{\label{line:algo_makeprimary_iterateoverpdsccnodes}
\ForEach{label $\CloseParenthesisBeta\in\SetClosedParenthesis$}{\label{line:algo_makeprimary_iterateoverpdsccnodeslabels}
\ForEach(\tcp*[f]{Retrieve that $s\DTo{\CloseParenthesisBeta}\PDSCC(t)$}){$s\in \InEdges[t][\CloseParenthesisBeta]$}{\label{line:algo_makeprimary_iterateoverpdsccinedges}
$x\gets \DisjointSets.\Find(s)$\tcp*[f]{The root representative of $s$'s component}\label{line:algo_makeprimary_rootoft}\\
Insert $t$ in $\Edges[x][\CloseParenthesisBeta]$\tcp*[f]{The component of $s$ has an outgoing edge to $t$}\label{line:algo_makeprimary_insert_t_in_edges_x}\\
Insert $(x, \CloseParenthesisBeta)$ in $\NewIncomingRoots$\tcp*[f]{Mark that we have inserted edges in $\Edges[x][\CloseParenthesisBeta]$}\label{line:populate_L}\\
}
\uIf(\tcp*[f]{$t$ has a $\CloseParenthesisBeta$ out-neighbor}){$|\OutEdges[t][\CloseParenthesisBeta]|\geq 1$}{\label{line:algo_makeprimary_edgelist_of_pdsccs}
$y\gets$ the first node in $\OutEdges[t][\CloseParenthesisBeta]$\tcp*[f]{An arbitrary $\CloseParenthesisBeta$ out-neighbor of $t$}\label{line:first}\\
\uIf(\tcp*[f]{The component of $y$ is not splitting}){$\DisjointSets.\Find(y)\not \in \NewRoots$}{
Insert $y$ in $\Edges[r][\CloseParenthesisBeta]$\tcp*[f]{The component of $t$ has a $\CloseParenthesisBeta$ edge to $y$}\label{line:algo_makeprimary_insert_y_in_edges_r}\\
}
}
}
}
}
\ForEach(\tcp*[f]{Re-initialize the queue $\Queue$ for the new fixpoint}){$(x,\CloseParenthesisBeta)\in \NewIncomingRoots$}{\label{line:algo_makeprimary_populateQstart}
\uIf(\tcp*[f]{There are two outgoing  $\CloseParenthesisBeta$-labeled edges from $x$'s component}){$|\Edges[x][\CloseParenthesisBeta]|\geq 2$}{
Insert $(x, \CloseParenthesisBeta)$ in $\Queue$\label{line:algo_makeprimary_inserttoqueue}\\
}
}\label{line:algo_makeprimary_thirdstepend}
\end{algorithm}
\Paragraph{Function $\MakePrimary()$.}
\cref{algo:make_primary} prepares the new fixpoint computation by partially splitting some DSCCs to their constituent PDSCCs.
In words, the algorithm first computes an overapproximation of the splitting DSCCs by running a forward graph search from $\DSCC(v)$ (\cref{line:algo_makeprimary_emptyqueue} to \cref{line:algo_makeprimary_firststepend}), every time transitioning to new DSCCs by following pairs of same-label edges outgoing the current DSCC.
This transitioning reflects the intuition that such pairs of edges might have used to form the neighboring DSCC -- by splitting the current one, we might have to split the neighbor as well. It begins the forward search by  populating 
$\DSCCRep(v)$ in a set $\AffectedDSCCs$  and a queue $\Stack$.  Eventually,  after \cref{line:algo_makeprimary_emptyqueue} to \cref{line:algo_makeprimary_firststepend},  
$\AffectedDSCCs$ will contain the roots of all affected $\DSCC$s (including $\DSCC(v)$) while $\Stack$ helps with the forward search of affected $\DSCC$s.

The second step (\cref{line:algo_makeprimary_newroots} to \cref{line:algo_makeprimary_secondstepend}) of \cref{algo:make_primary} splits these potentially affected DSCCs to their constituent PDSCCs.
In \cref{line:algo_makeprimary_split_to_pdsccs}, the algorithm performs this split in $\DisjointSets$ by iterating over the nodes of the DSCCs stored in $\AffectedDSCCs$.
$\NewRoots$ holds the roots of all the newly formed $\PDSCC$s after splitting (\cref{line:newroots}).
Then  
$\Edges[z][\cdot]$ is reset, where $z$ is either the root node of an affected DSCC (\cref{line:algo_makeprimary_secondstepend})  
or is the root node of an unaffected  DSCC 
which has an outgoing edge entering one of the affected $\DSCC$s (\cref{line:algo_makeprimary_unaffected}). 
The re-initialization is done keeping in mind the  split of the affected  $\DSCC$s. 
These will be re-populated in the third step with edges outgoing or entering the nodes of the newly formed $\PDSCC$s.  
Although we are guaranteed that these $\DSCC$s do not have to be split below their $\PDSCC$s, it may be the case that certain $\PDSCC$s  have to merge again.

In its third step, \cref{algo:make_primary} identifies the edges incoming to these $\PDSCC$s, and inserts them in 
the set $\NewIncomingRoots$ (\cref{line:algo_makeprimary_newincoming} to \cref{line:populate_L}).  
$\NewIncomingRoots$ is later used to 
decide whether two $\PDSCC$s must be merged or not. 
It also populates $\Edges$ with (i) edges incoming to the $\PDSCC$s (lines  
\cref{line:algo_makeprimary_iterateoverpdsccinedges} to \cref{line:algo_makeprimary_insert_t_in_edges_x}) 
as well as with (ii) edges (\cref{line:algo_makeprimary_edgelist_of_pdsccs} to \cref{line:algo_makeprimary_insert_y_in_edges_r}) 
outgoing from the $\PDSCC$s, and entering an unaffected $\DSCC$ (those whose roots are not in $\NewRoots$).  Note that edges coming out of a $\PDSCC$ and entering an affected $\DSCC$ (those whose roots are in $\NewRoots$) 
are covered by \cref{line:algo_makeprimary_iterateoverpdsccinedges} to \cref{line:algo_makeprimary_insert_t_in_edges_x}.

To merge PDSCCs, the worklist $\Queue$ is populated 
using relevant entries  of $\NewIncomingRoots$  
to drive the new fixpoint computation (\cref{line:algo_makeprimary_newincoming} to \cref{line:algo_makeprimary_thirdstepend}).
In more detail, the algorithm iterates over all PDSCCs (\cref{line:algo_makeprimary_iterateoverpdsccs}), and for each node $t$ in each PDSCC (\cref{line:algo_makeprimary_iterateoverpdsccnodes}) and label $\CloseParenthesisBeta$ (\cref{line:algo_makeprimary_iterateoverpdsccnodeslabels}), it identifies the edges $s\DTo{\CloseParenthesisBeta} t$ by iterating over the nodes $s\in \InEdges[t][\CloseParenthesisBeta]$ (\cref{line:algo_makeprimary_iterateoverpdsccinedges}).
It then identifies the root node $x$ of $s$'s component (\cref{line:algo_makeprimary_rootoft}), and inserts $t$ in $\Edges[x][\CloseParenthesisBeta]$ as well as the pair $(x,\CloseParenthesisBeta)$ in $\NewIncomingRoots$, to record that we have an incoming edge $\cdot \DTo{\CloseParenthesisBeta} t$ to a PDSCC that the fixpoint computation needs to process.
Then, \cref{line:algo_makeprimary_insert_y_in_edges_r} inserts in $\Edges[r][\CloseParenthesisBeta]$ the edges $t\DTo{\CloseParenthesisBeta}y$ that outgo $t$'s component (rooted at $r$) and for which $y$ is not in a breaking component.
Finally, the algorithm identifies the pairs $(x,\CloseParenthesisBeta)$ in $\NewIncomingRoots$ for which it has recorded at least two $\CloseParenthesisBeta$ neighbors out of $x$'s component, and inserts them in $\Queue$ to be processed by the fixpoint algorithm later (\cref{line:algo_makeprimary_populateQstart}).

By structuring $\MakePrimary$ this way we have the following benefits:
(i)~we spend no time in the non-affected DSCCs, and
(ii)~because we maintain the PDSCCs, the time cost is only linear in the number of \emph{nodes} of these affected DSCCs (and hence $O(n)$ in the worst case), as opposed to linear in the number of the \emph{edges} incoming to these affected DSCCs (which would be $O(n^2)$). This ensures that 
$\ell=\sum_{x,\CloseParenthesis} |\Edges[x][\CloseParenthesis]|=O(n)$ at the end of $\MakePrimary$ before calling $\Fixpoint$; 
as shown in \cite[Lemma~3.4, Lemma~3.5]{Chatterjee18}, a call to $\Fixpoint()$ takes time $O(\ell \cdot \alpha(n))$.

\subsection{Example}\label{subsec:example}

Here we give a step-by-step illustration of $\DynamicAlgo$ on the example of \cref{fig:illustrative_example}.
We start with a description of the contents of the various data structures before handling the edge insertion.

Let us fix some notations. Let $(v_1,v_2,\ov{R})\colon k$ denote that 
there are $k$ edges from $v_1$ to $v_2$ labeled $\ov{R}$. The elements 
of the set $\Count$ will contain entries of this kind.  Likewise, 
let $(v,\ov{L})\colon [v_1,v_2, \dots, v_k]$ represent that 
there are out edges labeled $\ov{L}$ from $v$ to $v_1, \dots, v_k$. 
 $\OutEdges$  contains entries of this kind. We let $\InEdges$ to contain 
entries of the kind $(v, \ov{L}) : [v_1, \dots, v_k]$ to denote that there are edges 
labeled $\ov{L}$ from each of $v_1, \dots, v_k$ to $v$. The components in the $\DisjointSets$ and $\PrimCompDS$ data structures are listed as sequences of nodes $\{(a),b,c,\dots\}$, where the node in the parenthesis denotes the component representative.

\begin{compactenum}
\item $\Count : \{(c,g,\ov{R})\colon 1 ,\, (f,c,\ov{L})\colon 1 ,\, (f,e,\ov{L})\colon 1 ,\, (f,d,\ov{L})\colon 1 ,\, (g,e,\ov{L})\colon1 ,\, (h,f,\ov{L})\colon 1 \}$
\item $\OutEdges: \{(f,\ov{L})\colon [e,c,d] ,\, (g,\ov{L})\colon [e] ,\, (h,\ov{L})\colon [f],\, (c,\ov{R})\colon [g]\}$
\item $\InEdges: \{(d,\ov{L})\colon [f],\, (e,\ov{L})\colon [g],\, (f,\ov{L})\colon [h],\, (g,\ov{R})\colon [c] \}$
\item $\Edges: \{(f,\ov{L}): [c] ,\, (c,\ov{R}): [g],\, (g,\ov{L}): [e],\, (h,\ov{L}): [f]\}$
\item $\DisjointSets: \{(c),d,e\},\, \{(g)\},\, \{(f)\},\, \{(h)\}$
\item $\PrimCompDS-\PDSCC: \{(c),d,e\}, \{(g)\}, \{(f)\},\, \{(h)\}$ 
\item  $\Queue : \emptyset$
\end{compactenum}
We now describe the development of the data structures during the edge insertion $(d,h,\ov{R})$.
\begin{compactenum}
\item $\Count[d][h][\ov{R}] = 1$.
\item Since $|\OutEdges[d][\ov{R}]|==0$, the algorithm adds $d$ to the 
$\InEdges[h][\ov{R}]$.
Then, the  algorithm inserts $h$ as the head of $\OutEdges[d][\ov{R}]$. 
\item 
Get the root of $\DSCC$ containing $d$ , i.e. $c = \DisjointSets.\Find(d)$. 
Then add $h$ to $\Edges[c][\ov{R}]$. 
So, updated $\Edges[c][\ov{R}] = (g,h)$. 
\item Since, $|\Edges[c][\ov{R}]| \ge 2$. 
The algorithm inserts $(c,\ov{R})$ in $\Queue$ and invokes $\Fixpoint()$.
\item Inside $\Fixpoint()$: First, we extract $(c,\ov{R})$ from $\Queue$. 
Then we get the set $S = \{g,h\}$. 
Since $|S| \ge 2$, the algorithm merges  $\DSCC$s $[(g)], [(h)]$ into single $\DSCC [(h),g]$. 
Then, moves $\Edges[g][\ov{L}]$ to $\Edges[h][\ov{L}]$.
Now, since $|\Edges[h][\ov{L}]| \ge 2$, insert $(h,\ov{L})$ to $\Queue$.
This will lead to merging of $\DSCC$s $[(c),d,e]$ and $[(f)]$ into $\DSCC$ $[(f),c,d,e]$.
This finishes the fixpoint.
\end{compactenum}
After the insertion, the data structures are in the following state.
\begin{compactenum}
\item $\Count$: new entry $(d,h,\ov{R})\colon 1$;
$\OutEdges$: new entry $(d,\ov{R})\colon [h]$; 
$\InEdges$: new entry $(h,\ov{R})\colon [d]$;
$\Queue : \emptyset$
\item $\Edges: \{(f,\ov{L}): [c] ,\, (f,\ov{R}): [h], \, (h,\ov{L}): [f]$
\item $\DisjointSets: \{(f),c,d,e\},\, \{(h) , g\}$
\item $\PrimCompDS-\PDSCC: \{(c),d,e\},\, \{(g)\},\, \{(f)\},\, \{(h)\}$
\end{compactenum}
We now describe the development of the data structures during the edge deletion $(f,d,\ov{L})$.
\begin{compactenum}
\item $\Count[f][d][\ov{L}] = 0$.
Since $d$ is at tail of $\OutEdges[f][\ov{L}]$, remove $f$ from 
$\InEdges[d][\ov{L}]$.
\item In $\OutEdges[f][\ov{L}]$, $c$ is penultimate element. 
Hence, the algorithm inserts $f$ in $\InEdges[c][\ov{L}]$.
\item Remove edge $(c,d)$ from $\PrimCompDS$. 
\item Remove $d$ from $\OutEdges[f][\ov{L}]$.
\item Invoke $\MakePrimary()$:
\begin{compactenum}
    \item 
    Insert $f=\DSCCRep(d)$ to $\Stack$ and $\AffectedDSCCs$.
\item In the loop from \cref{line:algo_makeprimary_fwdsearchloop} to \cref{line:algo_makeprimary_firststepend} 
the  algorithm inserts $h=\DSCCRep(g)$ to $\Stack$ and $\AffectedDSCCs$. 
\item $\AffectedDSCCs$ contains $f, h$ and by following $\InEdges$, algorithm re-initializes $\Edges[f][\ov{L}]$, $\Edges[f][\ov{R}]$ and $\Edges[h][\ov{L}]$ to an empty list.
\item Since $\AffectedDSCCs$ contains $f$ and $h$, we split $\DSCC(f)$ and $\DSCC(h)$ to their $\PDSCC$s. 
Then the algorithm inserts roots of new components in $\NewRoots$. 
So, we get $\NewRoots = \{ c , d , f, g, h\}$.
\item In the next step the algorithm gathers edges for fixpoint computations. 
For each root $r$ in $\NewRoots$, the algorithm iterate over all nodes in $\DSCC(r)$.
For each node $t$ in $\DSCC(r)$ and each edge type $\CloseParenthesisBeta \in \{\ov{L} , \ov{R}\}$ the algorithm checks if there exists node $s \in \InEdges[t][\CloseParenthesisBeta]$.
For all such nodes $s$, the algorithm inserts $t$ in $\Edges[\DSCCRep(s)][\CloseParenthesisBeta]$ 
and $(\DSCCRep(s),\CloseParenthesisBeta)$ in $\NewIncomingRoots$ (lines \cref{line:algo_makeprimary_iterateoverpdsccinedges} to \cref{line:populate_L}) .
In lines \cref{line:algo_makeprimary_edgelist_of_pdsccs} - \cref{line:algo_makeprimary_insert_y_in_edges_r} the algorithm checks if $\OutEdges[t][\CloseParenthesisBeta] \ge 1$. 
If yes, then inserts $x=$ first node of $\OutEdges[t][\CloseParenthesisBeta]$ in $\Edges[r][\CloseParenthesisBeta]$.

So after \cref{line:algo_makeprimary_insert_y_in_edges_r} the updated 
$\Edges$ and $\NewIncomingRoots$ will be:
\begin{compactenum}
    \item $\Edges: \{(c,\ov{R})\colon[g], \, (f,\ov{L})\colon[c],\,(h,\ov{L})\colon[f],\,(g,\ov{L})\colon[e],
    (d,\ov{R})\colon[h]$. 
    \item $\NewIncomingRoots: (d,\ov{R}), (f,\ov{L}), (c,\ov{R}) , (g,\ov{L}), 
    (h,\ov{L})$.
\end{compactenum}
Then the algorithm iterates over all $(x,\CloseParenthesisBeta)$ $ \in \NewIncomingRoots$,
and checks if $\Edges[x][\CloseParenthesisBeta] \ge 2$. 
If yes, then inserts $(x,\CloseParenthesisBeta)$ in $\Queue$.  
So, for the $\NewIncomingRoots$, computed in above step,
no element will be added to the $\Queue$.
\end{compactenum}

\item Since $\Queue$ is empty, the call to $\Fixpoint()$ will not merge any components.
\end{compactenum}

After the deletion, the data structures are in the following state.
\begin{compactenum}
\item $\Count : \{(c,g,\ov{R})\colon 1 ,\, (f,c,\ov{L})\colon 1 ,\, (f,e,\ov{L})\colon 1 ,\, (f,d,\ov{L})\colon 0 ,\, (g,e,\ov{L})\colon 1 ,\, (h,f,\ov{L})\colon 1 ,\, (d,h,\ov{R})\colon 1\}$, 
\item $\OutEdges: \{(f,\ov{L})\colon [e,c] ,\, (g,\ov{L})\colon [e] ,\, (h,\ov{L})\colon [f],\, (c,\ov{R})\colon [g] ,\, (d,\ov{R})\colon [h]\}$,  
\item $\InEdges: \{(c,\ov{L})\colon [f],\, (e,\ov{L})\colon [g],\, (f,\ov{L})\colon [h], \, (g,\ov{R})\colon [c] ,\, (h,\ov{R})\colon [d]\}$
\item $\Edges: \{ (c,\ov{R}):[g] ,\, (f,\ov{L}):[c],\, (h,\ov{L}):[f], \, (g,\ov{L}):[e],\, (d,\ov{R}):[h] \}$
\item $\DisjointSets: \{(c),e\} , \{(d)\}, \{(g)\} , \{(f)\} , \{(h)\}$
\item $\PrimCompDS-\PDSCC: \{(c),e\} , \{(d)\}, \{(g)\} , \{(f)\} , \{(h)\}$
\item  $\Queue : \emptyset$
\end{compactenum}

\subsection{Correctness and Complexity}\label{subsec:analysis}

Finally, we state the correctness and complexity of $\DynamicAlgo$.

\thmmaintheorem*

Next, we state the main invariants that $\DynamicAlgo$ maintains which support its correctness and complexity, as well as some intuition behind them.
For proofs, we refer to \cref{sec:app_proofs}.

\Paragraph{Correctness.}
The basis of the correctness of $\DynamicAlgo$ is a number of invariants that are maintained along edge insertions and deletions.
Observe that $\OutEdges[u][\CloseParenthesis]$ is, at all times, a linked list representation of the edge set $u\DTo{\CloseParenthesis}\cdot$.

Our first invariant concerns the correct maintenance of the PDSCCs of $G$ in the $\PrimCompDS$ data structure maintaining undirected connectivity.
To prove the invariant, we argue that $\DynamicAlgo$ inserts and removes sufficient and necessary undirected edges in $\PrimCompDS$ to (sparsely) represent the connected components of the corresponding primal graph.
This follows directly from the sparsification approach outlined in \cref{subsec:intuition} and \cref{fig:pdsccs}.
Indeed, the algorithm maintains in $\PrimCompDS$ a set of edges that connect neighboring nodes in each linked list $\OutEdges[u][\CloseParenthesis]$.
Thus, the omitted edges (i.e., between non-neighboring nodes in $\OutEdges[u][\CloseParenthesis]$) are anyway transitively connected in $\PrimCompDS$.
Formally, we have the following lemma.

\begin{restatable}{lemma}{lemcorrectnesspdsccs}\label{lem:correctness_pdsccs}
After every insert and delete operation,
the connected components in $\PrimCompDS$ are precisely the PDSCCs of $G$.
\end{restatable}

The next invariant concerns the correct maintenance of the $\InEdges$ data structure, and is established in \cref{lem:soundness_inedges} and \cref{lem:completeness_inedges}.
In words, the invariant states that $x\in \InEdges[z][\CloseParenthesisBeta]$ iff we have an edge $x\DTo{\CloseParenthesisBeta}z$ and $z$ is the last node in $\OutEdges[x][\CloseParenthesisBeta]$.
Thus, when $\MakePrimary()$ constructs the PDSCCs and discovers their incoming edges, the edge $x\DTo{\CloseParenthesisBeta}\PDSCC(z)$ is correctly discovered by finding that $x\in \InEdges[z][\CloseParenthesisBeta]$ (\cref{line:algo_makeprimary_iterateoverpdsccinedges}).

\cref{lem:soundness_inedges} is concerned with the soundness, and is straightforward.
Any time the algorithm inserts $x\in\InEdges[z][\CloseParenthesisBeta]$, this is followed by inserting $z$ in $\OutEdges[x][\CloseParenthesisBeta]$ (in the case of edge insertions,  \cref{line:algo_insert_inedges} and \cref{line:algo_insert_outedges} in \cref{algo:insert_edge}). 
In the case of edge deletions, the insertion $x\in\InEdges[z][\CloseParenthesisBeta]$ is preceded by inserting $z$ to $\OutEdges[x][\CloseParenthesisBeta]$ in a previous update 
(\cref{line:penultimate} and \cref{line:algo_delete_inedges} in \cref{algo:delete_edge}). In this case,  $x\in\InEdges[z][\CloseParenthesisBeta]$ takes place when $z$ becomes the last node  in $\OutEdges[x][\CloseParenthesisBeta]$ after deletion of edges $x\DTo{\CloseParenthesisBeta} v$, where 
$v$ appeared later than $z$ in $\OutEdges[x][\CloseParenthesisBeta]$. 

\begin{restatable}{lemma}{lemsoundnessinedges}\label{lem:soundness_inedges}
After every insert and delete operation,
for every two nodes $x,z\in V$ and label $\CloseParenthesisBeta\in \SetClosedParenthesis$, if  $x\in\InEdges[z][\CloseParenthesisBeta]$  then $x\DTo{\CloseParenthesisBeta}z$ and $z$ is the last node in $\OutEdges[x][\CloseParenthesisBeta]$.
\end{restatable}

Similarly, for \cref{lem:completeness_inedges}, if $x\DTo{\CloseParenthesisBeta}\PDSCC(y)$,
then the last node $z$ in $\OutEdges[x][\CloseParenthesisBeta]$ is also in $\PDSCC(y)$, while the algorithm maintains that $x\in\InEdges[z][\CloseParenthesisBeta]$.
Hence, the fact that $x\DTo{\CloseParenthesisBeta}\PDSCC(y)$ is recoverable via discovering that $x\in\InEdges[z][\CloseParenthesisBeta]$.

\begin{restatable}{lemma}{lemcompletenessinedges}\label{lem:completeness_inedges}
After every insert and delete operation,
for every pair of nodes $x,y\in V$ and label $\CloseParenthesisBeta\in\SetClosedParenthesis$, if $x\DTo{\CloseParenthesisBeta}\PDSCC(y)$ then there exists a node $z\in \PDSCC(y)$ such that $x\in\InEdges[z][\CloseParenthesisBeta]$.
\end{restatable}

The next lemma captures the correctness of $\MakePrimary()$ and the invariants concerning the state that it passes on to function $\Fixpoint()$ for the final fixpoint computation after processing a $\DeleteEdge(u,v,\CloseParenthesis)$ update.
Recall that $\MakePrimary()$ identifies an overapproximation of the DSCCs that have to be split and rebuilt as a result of this edge deletion.
The lemma has two parts.
\cref{item:make_primary_lemma_dsccs} states that the components that $\MakePrimary()$ passes on to $\Fixpoint()$ (i.e., those found in $\DisjointSets$) are a refinement of the DSCC decomposition of $G$ after the deletion, i.e., it suffices to merge some of them in order to arrive at the correct DSCC-decomposition of the graph after the edge deletion.
The $\Fixpoint()$ function will perform this merging by processing the edges found in the $\Edges$ data structure.
\cref{item:make_primary_lemma_edges} states that $\MakePrimary()$ populates the $\Edges$ data structure with sufficiently many edges for $\Fixpoint()$ to process and arrive at the correct DSCC decomposition of $G$.

\begin{restatable}{lemma}{lemmakeprimarycorrectness}\label{lem:make_primary_correctness}
At the end of $\MakePrimary()$, the following assertions hold.
\begin{compactenum}
\item \label{item:make_primary_lemma_dsccs} Every component in $\DisjointSets$ is a (not necessarily maximal) DSCC of $G$.
\item\label{item:make_primary_lemma_edges} For every component in $\DisjointSets$ rooted in some node $x$, the following hold.
\begin{compactenum}[label=(\alph*)]
\item\label{item:make_primary_lemma_edges_sound} For every node $y$ and label $\CloseParenthesisBeta$, if $y\in\Edges[x][\CloseParenthesisBeta]$ then there is an edge $z\DTo{\CloseParenthesisBeta}w$, where $z$ is a node in the component of $x$, and $w$ is a node in the component of $y$ in $\DisjointSets$.
\item\label{item:make_primary_lemma_edges_complete} 
For every node $z$ and label $\CloseParenthesisBeta$ such that 
(i) ~$z$ is in the component rooted at node $x$ in $\DisjointSets$, and
~(ii) there is an edge $z\DTo{\CloseParenthesisBeta}w$,  
there exists a node $y$ in the component of $w$ in $\DisjointSets$ such that $y\in \Edges[x][\ov{\beta}]$.
\end{compactenum}
\end{compactenum}
\end{restatable}

Given the above invariants, $\MakePrimary$ creates a correct state of the worklist $\Queue$ and the $\Edges$ linked lists for the $\Fixpoint()$ call of \cref{algo:delete_edge} (\cref{line:algo_delete_fixpoint}) to compute the correct DSCC decomposition.
This leads ot the correctness of $\DynamicAlgo$.

\begin{restatable}{lemma}{lemcorrectness}\label{lem:correctness}
After every $\InsertEdge(u,v,\CloseParenthesis)$ and $\DeleteEdge(u,v,\CloseParenthesis)$,
$\DisjointSets$ contains the DSCCs of $G$.
\end{restatable}

\Paragraph{Complexity.}
We now turn our attention to the complexity bound of \cref{thm:main_theorem}.
We always start with a graph of $n$ nodes but without any edges, for which the initialization of all data structures takes $O(n)$ time.
In practice, the initial graph might already have some edges, which can be thought of being inserted one-by-one.

In high level, the time $\DynamicAlgo$ spends in each edge insertion and deletion is the sum of two parts:~(i) the time taken for maintaining the PDSCCs in the $\PrimCompDS$ data structure, and
(ii)~the time taken for all other computations.
Regarding (i), so far we have not specified the precise data structure for implementing $\PrimCompDS$, as $\DynamicAlgo$ treats $\PrimCompDS$ as a black box.
The theoretical guarantees of \cref{thm:main_theorem} can be obtained by using the data structure of~\cite{Eppstein1997} for undirected connectivity (see \cref{tab:connectivity_data_structures}).
Although this guarantees $O(\sqrt{n})$ and $O(1)$ time for edge insertions/deletions and queries, respectively, here we only use the fact that both bounds are less than $n$.
Regarding (ii), the algorithm spends $O(n\cdot \alpha(n))$ for each operation, which stems from the fact that the algorithm encounters each node of $G$ across all data structures a constant number of times (recall that we have $k=O(1)$ labels in $G$).
In particular, we have the following lemma.
\begin{restatable}{lemma}{lemcomplexity}\label{lem:complexity}
Every $\InsertEdge(u,v,\CloseParenthesis)$ and $\DeleteEdge(u,v,\CloseParenthesis)$ operation is processed in $O(n\cdot \alpha(n))$ time.
\end{restatable}
\begin{proof}
After every insert and delete operation,  $\Fixpoint()$ is invoked for merging the DSCCs (in $\DisjointSets$).
As shown in \cite[Lemma~3.4, Lemma~3.5]{Chatterjee18}, a call to $\Fixpoint()$
takes time $O(\ell \cdot \alpha(n))$, where $\ell=\sum_{u,\CloseParenthesis} |\Edges[x][\CloseParenthesis]|$ is the total number of nodes stored in the $\Edges$ linked lists, and $\alpha(n)$ is the inverse Ackermann function.
After $\Fixpoint()$ has completed, we have $|\Edges[x][\CloseParenthesis]|\leq 1$ for each $x$ and $\CloseParenthesis$.
Thus to prove \cref{lem:complexity}, it suffices to argue that every operation takes $O(n)$ time before calling $\Fixpoint()$.
This implies that $\ell=O(n)$ before the call, and thus $\Fixpoint()$ takes $O(n\cdot \alpha(n))$ time, yielding $O(n)+O(n\cdot \alpha(n))=O(n\cdot \alpha(n))$ total time.

\SubParagraph{Edge insertions.}
Consider the processing of an operation $\InsertEdge(u,v,\CloseParenthesis)$ by \cref{algo:insert_edge}.
Observe that the algorithm has no loops, thus the running time is dominated by the time taken to access the various data structures that the algorithm maintains.
In particular, $\OutEdges$ is a simple linked list, and checking its length (\cref{line:algo_insert_pdscc}), as well as  accessing and inserting in the head (\cref{line:algo_insert_gethead} and \cref{line:algo_insert_outedges}) takes constant time.
Similarly, the data structures $\InEdges$ and $\Edges$ are simple sets and linked lists with $O(1)$ accesses.
As we have already argued in \cref{subsec:dynamic_reachability}, performing a $\DisjointSets.\Find(v)$ operation (\cref{line:algo_insert_root_of_u}) takes $O(1)$ time.
Finally, inserting the edge $(v,y)$ in $\PrimCompDS$ takes $O(\sqrt{n})$ time~\cite{Eppstein1997}.

\SubParagraph{Edge deletions.}
Consider the processing of an operation $\DeleteEdge(u,v,\CloseParenthesis)$ by \cref{algo:delete_edge}.
The time spent in the body of \cref{algo:delete_edge} is $O(n)$, by an analysis very similar to \cref{algo:insert_edge} for edge deletions, and we will not repeat it here.
Instead, we focus on the time spent in the call to $\MakePrimary()$, which is the more complex part of processing the edge deletion (\cref{algo:make_primary}).
In the first step (\cref{line:algo_makeprimary_emptyqueue} to \cref{line:algo_makeprimary_firststepend}), for each iteration of the while loop of \cref{line:algo_makeprimary_fwdsearchloop}, the condition in \cref{line:algo_makeprimary_iftwoedgestot} can be checked in $O(|\DSCC(x)|)$ time, by iterating over all nodes of $\DSCC(x)$.
Observe that $\DSCC(x)$ is examined only once throughout this loop of \cref{line:algo_makeprimary_fwdsearchloop}, and since the DSCCs partition the node set, total time for running this loop is $O(n)$.

The time spent in the second step (\cref{line:algo_makeprimary_newroots} to \cref{line:algo_makeprimary_secondstepend}) is the sum of two parts.
The first part corresponds to 
the time spent in the nested loops, which  is proportional to the number of times the inner-most 
loop in \cref{line:algo_makeprimaryiterateoverpdsccinedgesofunaffected} is taken. That is, 
the number of nodes $x$ that exist in the lists $\InEdges[t][\CloseParenthesisBeta]$ of nodes $t$ in the components represented by their roots in $\AffectedDSCCs$.
It suffices to argue that there do not exist nodes $x,s_1,s_2$ and label $\CloseParenthesisBeta$ such that $x\in \InEdges[s_1][\CloseParenthesisBeta]$ and $x\in \InEdges[s_2][\CloseParenthesisBeta]$.
Indeed, by \cref{lem:soundness_inedges}, if $x\in \InEdges[s_1][\CloseParenthesisBeta]$ then $s_1$ is the last node in $\OutEdges[x][\CloseParenthesisBeta]$.
Clearly, as a linked list, $\OutEdges[x][\CloseParenthesisBeta]$ can have at most one last node, hence $x\not\in \InEdges[s_2][\CloseParenthesisBeta]$.
Thus, for every label $\CloseParenthesisBeta$, the loop in \cref{line:algo_makeprimaryiterateoverpdsccinedgesofunaffected} is executed at most once per node $x$ leading to $O(n)$ total iterations.
The second part corresponds to the total time spent in \cref{line:algo_makeprimary_split_to_pdsccs}. 
In \cref{line:algo_makeprimary_split_to_pdsccs}, the 
algorithm splits the potentially affected DSCCs (in $\DisjointSets$) into its PDSCCs.
This takes $O(n)$ time as the algorithm simply iterates over the nodes of the DSCCs stored in $\AffectedDSCCs$.
For each such node, the  algorithm performs a single membership query to $\PrimCompDS$, which takes $O(1)$ time~\cite{Eppstein1997}.

Finally, we consider the time spent by $\MakePrimary()$ in the third step (\cref{line:algo_makeprimary_newincoming} to \cref{line:algo_makeprimary_thirdstepend}).
Note that the time spent in these nested loops is proportional to the number of times the inner-most loop (\cref{line:algo_makeprimary_iterateoverpdsccinedges}) is executed.
Again, it suffices to argue that there do not exist nodes $x,s_1,s_2$ and label $\CloseParenthesisBeta$ such that $x\in \InEdges[s_1][\CloseParenthesisBeta]$ and $x\in \InEdges[s_2][\CloseParenthesisBeta]$.
Indeed, by \cref{lem:soundness_inedges}, if $x\in \InEdges[s_1][\CloseParenthesisBeta]$ then $s_1$ is the last node in $\OutEdges[x][\CloseParenthesisBeta]$,
hence $x\not\in \InEdges[s_2][\CloseParenthesisBeta]$.
Thus, for every label $\CloseParenthesisBeta$, the loop in \cref{line:algo_makeprimary_iterateoverpdsccinedges} is executed at most once per node $x$, leading to $O(n)$ total iterations.
In the end, the loop in \cref{line:algo_makeprimary_populateQstart} runs for $|\NewIncomingRoots|$ iterations, which is bounded by the iterations of the previous loop.

Since for every label $\CloseParenthesisBeta$ and node $x$ the loop in \cref{line:algo_makeprimary_iterateoverpdsccinedges} is executed at most once, it follows that  \cref{line:algo_makeprimary_insert_t_in_edges_x} will be executed $O(n)$ times. 
This leads to total $O(n)$ new entries to $\Edges$ at \cref{line:algo_makeprimary_insert_t_in_edges_x}.
Similarly, \cref{line:algo_makeprimary_insert_y_in_edges_r} is executed $O(n)$ times leading to $O(n)$ new entries to $\Edges$. 
Therefore, at the end of $\MakePrimary()$ we have $\sum_{x,\CloseParenthesis} |\Edges[x][\CloseParenthesis]| = O(n)$.
The loop in \cref{line:algo_makeprimary_populateQstart} runs for $|\NewIncomingRoots| = O(n)$ iterations. This leads to $O(n)$ entries in $\Queue$. Thus, at the end of $\MakePrimary()$, $|\Queue| = O(n)$.
Thus before calling $\Fixpoint()$ we have $\sum_{u,\CloseParenthesis} |\Edges[u][\CloseParenthesis]| = O(n)$ and $|\Queue| = O(n)$.
\end{proof}

\section{Experiments}\label{sec:experiments}
In this section we report on an implementation of $\DynamicAlgo$ algorithm behind \cref{thm:main_theorem}, and an evaluation of its performance on various datasets on real-world static analyses.
To some extent, our experimental setting follows~\cite{Li2022}.

\Paragraph{Compared algorithms.}
We compare three standard approaches to bidirected Dyck reachability\footnote{Since the algorithm of \cite{Li2022} has complexity and correctness issues, we have not included it in our evaluation.}.
\begin{compactenum}
\item $\OfflineAlgo$, as developed and implemented in~\cite{Chatterjee18}.
For each graph update (edge insertion/deletion), the algorithm is invoked from scratch to handle the updated graph.
\item Our $\DynamicAlgo$, which is implemented in C/C++, and closely follows the pseudocode presented in \cref{sec:data_structure}.
$\DynamicAlgo$ uses as a black box a data structure $\PrimCompDS$ for dynamic undirected connectivity (see \cref{tab:connectivity_data_structures}).
Although the one developed in~\cite{Eppstein1997} works best towards the complexity guarantees of \cref{thm:main_theorem}, in our implementation we use the one developed in~\cite{Holm2001} (and implemented in~\cite{DynamicUndirectedConnectivity_Tomtseng2020}), as it is conceptually simpler and well-performing in practice. 
\item A declarative approach in which the production rules of Dyck reachability are encoded as Datalog constraints and dispatched to a Datalog solver~\cite{Reps1995}.
Datalog-based static analyses have been popularized in the Flix programming language~\cite{Madsen2020} and the Doop framework~\cite{Bravenboer2009}.
Our bidirected setting allows us to optimize the Datalog program by explicitly focusing only on closing-parentheses edges, in similar style to $\OfflineAlgo$ and $\DynamicAlgo$.
To be fair in our comparison, we follow this approach here.
In particular, we use the following Datalog program.

\begin{lstlisting}[escapechar=\%]
Reaches(u,u)
Close(x,u,%$\CloseParenthesis$%):- Edge(x,u,%$\CloseParenthesis$%)
Close(x,u,%$\CloseParenthesis$%):- Edge(y,u,%$\CloseParenthesis$%), Reaches(x,y)
Reaches(u,v):- Close(x,u,%$\CloseParenthesis$%), Close(x,v,%$\CloseParenthesis$%)
Reaches(u,v):- Reaches(u,x), Reaches(x,v)
\end{lstlisting}
To handle our dynamic setting, we rely on an efficient, fully dynamic Datalog solver~\cite{DDLog},
as opposed to solving the Datalog program from scratch every time.
\end{compactenum}
The two algorithms ($\OfflineAlgo$ and dynamic Datalog) that support our experimental comparison serve as very fitting baselines.
$\OfflineAlgo$ is an algorithm dedicated to the  \emph{problem} we are solving (i.e., bidrected Dyck reachability) but not dedicated to the \emph{setting} we are solving it in (i.e., under dynamic updates).
On the other hand, the dynamic-Datalog approach is dedicated to the dynamic setting, but not dedicated to the bidirected Dyck-reachability problem.
As such, both algorithms are theoretically of worse complexity that our $\DynamicAlgo$, yet still the closest  that exist in the literature for this problem and setting.
Our experiments aim to highlight to what extent the theoretical superiority of $\DynamicAlgo$ is realized in practice.

\Paragraph{Benchmarks.}
We evaluate the above algorithms on two popular static analyses.
\begin{compactenum}
\item Context-sensitive data dependence analysis as formulated in~\cite{Tang15}, and evaluated on benchmark programs from~\cite{SPECjvm2008}.
In this case the parenthesis labels represent calling contexts, and a properly-balanced-parenthesis path represents an interprocedurally-valid dataflow via parameter-passing in function invocation and return.
\item Field-sensitive alias analysis for Java as formulated in~\cite{Yan11,Zhang13}, and evaluated on DaCapo benchmarks~\cite{DaCapo:paper}.
In this case the parenthesis labels represent field accesses on composite objects, as illustrated in \cref{fig:illustrative_example} of \cref{sec:intro}.
\end{compactenum}
Our graph models for the above analyses and benchmarks are obtained from~\cite{Li2022}.

\Paragraph{On-the-fly formulation and update sequences.}
To simulate the on-the-fly setting where the source code undergoes a sequence of changes,
for each benchmark graph $G$, we generate three sequences $\Sequence_{G}$ of updates (edge insertions/deletions), as follows.
\begin{compactenum}
\item \emph{Incremental setting:}~We randomly select a set $E^+$ of 90\% of the edges of $G$ and remove them from the graph.
We create a sequence of edge insertions $\Sequence^{\mathsf{inc}}_{G}$ as a random permutation of $E^+$.
This is a fully incremental setting, where code lines are only added but never removed in the program.
\item \emph{Decremental setting:}~We randomly select a set $E^-$ of 90\% of the edges of $G$.
We create a sequence of edge deletions $\Sequence^{\mathsf{dec}}_{G}$ as a random permutation of $E^-$.
This is a fully decremental setting, where code lines are only removed but never added in the program.
\item \emph{Mixed setting:}~We randomly split the edges of $G$ into two sets $E^+$ and $E^-$, with proportion 10\% to 90\%, and start with an initial graph containing the edges of $E^-$.
We create a sequence of mixed operations (both insertions and deletions) $\Sequence^{\mathsf{dec}}_{G}$ by repeated stochastic sampling:~in each step, we randomly choose the next operation as an edge insertion/edge deletion.
In the former case, we randomly select an edge from $E^+$, move it to $E^-$, and insert it in $G$.
In the latter case, we randomly select an edge from $E^-$, move it to $E^+$, and delete it from $G$.
The length of the sequence is equal to 90\% of the number of edges in $G$.
\end{compactenum}
In each case, for each benchmark we report the amortized time 
that each algorithm took to handle the whole sequence:~that is, the total running time over the whole sequence divided by the length of the sequence.
To gain more confidence in our results, we repeat the above process three times and report the average numbers.
We run our experiments on a conventional laptop with a 2.6GHz CPU and 16GBs of memory, which was always sufficient for the analysis.
As a sanity check, we have verified that all three algorithms give the same results on each benchmark and update sequence.

\begin{figure}
\captionsetup[subfigure]{aboveskip=-1pt,belowskip=-1pt}

\pgfplotsset{very tick label/.append style={font=\small}}
\def\RatioDataDependence{.58}
\def\PlotWidthDataDependence{9cm}
\def\PlotWidthAlias{6.5cm}
\def\RatioAlias{.38}

\def\subfigureTextWidthDataDependence{\NumBenchmarksDataDependence/\NumBenchmarksTotal}
\def\subfigureTextWidthAlias{\NumBenchmarksAlias/\NumBenchmarksTotal}

\def\plotheight{4cm}
\def\barwidth{2.5pt}
\def\bardistance{1pt}
\def\scaleboxvalue{0.9}

\def\offlinecolor{black!10}
\def\ddLogcolor{black!50}
\def\dynamiccolor{black!80}

\begin{subfigure}[T]{\RatioDataDependence\textwidth}
\scalebox{\scaleboxvalue}{%
\begin{tikzpicture}[tight background, inner sep=2pt,]
\begin{axis}[inner sep=2pt,
		title={\large \underline{Data Dependence Analysis}},
    ybar=\bardistance,
		enlarge y limits={abs=0.6cm,upper},
		enlarge x limits={abs=8pt},
		legend style={legend pos=north west, draw=black, legend columns=-1,/tikz/every even column/.append style={column sep=0.2cm}},
		ylabel near ticks,
    ylabel={microsecs},
		xtick pos=left,
    xtick={1,2,3,4,5,6,7,8,9,10,11,12,13,14,15},
    xticklabels={btree,check,compiler,compress,crypto,derby,helloworld,mpegaudio,mushroom,parser,sample,scimark,startup,sunflow,xml},
    ymajorgrids=true,
    grid style=dashed,
		width=\PlotWidthDataDependence,
		height=\plotheight,
    legend style={font=\footnotesize},
    xticklabel style={font=\footnotesize,rotate=45,anchor=east},
    yticklabel style={font=\footnotesize},
    ylabel style={font=\footnotesize},
    bar width=\barwidth,
    ymode=log,
    log origin=infty,
]

\addplot [fill=\offlinecolor] file {figures/suite2_90_10_runtime/expt_datadep_offline_incr.dat};
\addplot [fill=\ddLogcolor] file {figures/suite2_90_10_runtime/expt_datadep_dyn_incr_ddlog.dat};
\addplot[fill=\dynamiccolor] file {figures/suite2_90_10_runtime/expt_datadep_dyndyck_incr.dat};
\legend{Offline,DDlog,Dynamic}

\end{axis}
\end{tikzpicture}
}%
\caption{\label{subfig:datadependence_all_incr_suite2_90_10}
Incremental updates.
}
\end{subfigure}
\hspace*{\fill}
\begin{subfigure}[T]{\RatioAlias\textwidth}
\scalebox{\scaleboxvalue}{%
\begin{tikzpicture}[tight background]
\begin{axis}[inner sep=2pt,,
		title={\large \underline{Alias Analysis}},
    ybar=\bardistance,
		enlarge y limits={abs=0.6cm,upper},
		enlarge x limits={abs=8pt},
		legend style={legend pos=north west, draw=black, legend columns=-1,/tikz/every even column/.append style={column sep=0.2cm}},
		ylabel near ticks,
    ylabel={microsecs},
		xtick pos=left,
    xtick={1,2,3,4,5,6,7,8,9,10,11},
    xticklabels={antlr, bloat, chart, eclipse, fop, hsqldb, jython, luindex, lusearch, pmd, xalan},
    ymajorgrids=true,
    grid style=dashed,
		width=\PlotWidthAlias,
		height=\plotheight,
    legend style={font=\footnotesize},
    xticklabel style={font=\footnotesize,rotate=45,anchor=east},
    yticklabel style={font=\footnotesize},
    ylabel style={font=\footnotesize},
    bar width=\barwidth,
    ymode=log,
    log origin=infty,
]

\addplot [fill=\offlinecolor] file {figures/suite2_90_10_runtime/expt_alias_offline_incr.dat};
\addplot [fill=\ddLogcolor] file {figures/suite2_90_10_runtime/expt_alias_dyn_incr_ddlog.dat};
\addplot[fill=\dynamiccolor] file {figures/suite2_90_10_runtime/expt_alias_dyndyck_incr.dat};
\legend{Offline,DDlog,Dynamic}

\end{axis}
\end{tikzpicture}
}%
\vspace{7pt}
\caption{\label{subfig:alias_all_incr_suite2_90_10}
Incremental updates.
}
\end{subfigure}
\\[1em]
\begin{subfigure}[T]{\RatioDataDependence\textwidth}
\scalebox{\scaleboxvalue}{%
\begin{tikzpicture}[tight background, inner sep=2pt,]
\begin{axis}[inner sep=2pt,
    ybar=\bardistance,
		enlarge y limits={abs=0.6cm,upper},
		enlarge x limits={abs=8pt},
		legend style={legend pos=north west, draw=black, legend columns=-1,/tikz/every even column/.append style={column sep=0.2cm}},
		ylabel near ticks,
    ylabel={microsecs},
		xtick pos=left,
    xtick={1,2,3,4,5,6,7,8,9,10,11,12,13,14,15},
    xticklabels={btree,check,compiler,compress,crypto,derby,helloworld,mpegaudio,mushroom,parser,sample,scimark,startup,sunflow,xml},
    ymajorgrids=true,
    grid style=dashed,
		width=\PlotWidthDataDependence,
		height=\plotheight,
    legend style={font=\footnotesize},
    xticklabel style={font=\footnotesize,rotate=45,anchor=east},
    yticklabel style={font=\footnotesize},
    ylabel style={font=\footnotesize},
    bar width=\barwidth,
    ymode=log,
    log origin=infty,
]

\addplot [fill=\offlinecolor] file {figures/suite2_90_10_runtime/expt_datadep_offline_decr.dat};
\addplot [fill=\ddLogcolor] file {figures/suite2_90_10_runtime/expt_datadep_dyn_decr_ddlog.dat};
\addplot[fill=\dynamiccolor] file {figures/suite2_90_10_runtime/expt_datadep_dyndyck_decr.dat};
\legend{Offline,DDlog,Dynamic}

\end{axis}
\end{tikzpicture}
}%
\caption{\label{subfig:datadependence_all_decr_suite2_90_10}
Decremental updates.
}

\end{subfigure}
\hspace*{\fill}
\begin{subfigure}[T]{\RatioAlias\textwidth}
\scalebox{\scaleboxvalue}{%
\begin{tikzpicture}[tight background]
\begin{axis}[inner sep=2pt,,
    ybar=\bardistance,
		enlarge y limits={abs=0.6cm,upper},
		enlarge x limits={abs=8pt},
		legend style={legend pos=north west, draw=black, legend columns=-1,/tikz/every even column/.append style={column sep=0.2cm}},
		ylabel near ticks,
    ylabel={microsecs},
		xtick pos=left,
    xtick={1,2,3,4,5,6,7,8,9,10,11},
    xticklabels={antlr, bloat, chart, eclipse, fop, hsqldb, jython, luindex, lusearch, pmd, xalan},
    ymajorgrids=true,
    grid style=dashed,
		width=\PlotWidthAlias,
		height=\plotheight,
    legend style={font=\footnotesize},
    xticklabel style={font=\footnotesize,rotate=45,anchor=east},
    yticklabel style={font=\footnotesize},
    ylabel style={font=\footnotesize},
    bar width=\barwidth,
    ymode=log,
    log origin=infty,
]

\addplot [fill=\offlinecolor] file {figures/suite2_90_10_runtime/expt_alias_offline_decr.dat};
\addplot [fill=\ddLogcolor] file {figures/suite2_90_10_runtime/expt_alias_dyn_decr_ddlog.dat};
\addplot[fill=\dynamiccolor] file {figures/suite2_90_10_runtime/expt_alias_dyndyck_decr.dat};
\legend{Offline,DDlog,Dynamic}

\end{axis}
\end{tikzpicture}
}%
\vspace{7pt}
\caption{\label{subfig:alias_all_decr_suite2_90_10}
Decremental updates.
}
\end{subfigure}
\\[1em]
\begin{subfigure}[T]{\RatioDataDependence\textwidth}
\scalebox{\scaleboxvalue}{%
\begin{tikzpicture}[tight background, inner sep=2pt,]
\begin{axis}[inner sep=2pt,
    ybar=\bardistance,
		enlarge y limits={abs=0.6cm,upper},
		enlarge x limits={abs=8pt},
		legend style={legend pos=north west, draw=black, legend columns=-1,/tikz/every even column/.append style={column sep=0.2cm}},
		ylabel near ticks,
    ylabel={microsecs},
		xtick pos=left,
    xtick={1,2,3,4,5,6,7,8,9,10,11,12,13,14,15},
    xticklabels={btree,check,compiler,compress,crypto,derby,helloworld,mpegaudio,mushroom,parser,sample,scimark,startup,sunflow,xml},
    ymajorgrids=true,
    grid style=dashed,
		width=\PlotWidthDataDependence,
		height=\plotheight,
    legend style={font=\footnotesize},
    xticklabel style={font=\footnotesize,rotate=45,anchor=east},
    yticklabel style={font=\footnotesize},
    ylabel style={font=\footnotesize},
    bar width=\barwidth,
    ymode=log,
    log origin=infty,
]

\addplot [fill=\offlinecolor] file {figures/suite2_90_10_runtime/expt_datadep_offline_mixed.dat};
\addplot [fill=\ddLogcolor] file {figures/suite2_90_10_runtime/expt_datadep_dyn_mixed_ddlog.dat};
\addplot[fill=\dynamiccolor] file {figures/suite2_90_10_runtime/expt_datadep_dyndyck_mixed.dat};
\legend{Offline,DDlog,Dynamic}

\end{axis}
\end{tikzpicture}
}%
\caption{\label{subfig:datadependence_all_mixed_suite2_90_10}
Mixed updates.
}

\end{subfigure}
\hspace*{\fill}
\begin{subfigure}[T]{\RatioAlias\textwidth}
\scalebox{\scaleboxvalue}{%
\begin{tikzpicture}[tight background]
\begin{axis}[inner sep=2pt,,
    ybar=\bardistance,
		enlarge y limits={abs=0.6cm,upper},
		enlarge x limits={abs=8pt},
		legend style={legend pos=north west, draw=black, legend columns=-1,/tikz/every even column/.append style={column sep=0.2cm}},
		ylabel near ticks,
    ylabel={microsecs},
		xtick pos=left,
    xtick={1,2,3,4,5,6,7,8,9,10,11},
    xticklabels={antlr, bloat, chart, eclipse, fop, hsqldb, jython, luindex, lusearch, pmd, xalan},
    ymajorgrids=true,
    grid style=dashed,
		width=\PlotWidthAlias,
		height=\plotheight,
    legend style={font=\footnotesize},
    xticklabel style={font=\footnotesize,rotate=45,anchor=east},
    yticklabel style={font=\footnotesize},
    ylabel style={font=\footnotesize},
    bar width=\barwidth,
    ymode=log,
    log origin=infty,
]

\addplot [fill=\offlinecolor] file {figures/suite2_90_10_runtime/expt_alias_offline_mixed.dat};
\addplot [fill=\ddLogcolor] file {figures/suite2_90_10_runtime/expt_alias_dyn_mixed_ddlog.dat};
\addplot[fill=\dynamiccolor] file {figures/suite2_90_10_runtime/expt_alias_dyndyck_mixed.dat};
\legend{Offline,DDlog,Dynamic}

\end{axis}
\end{tikzpicture}
}%
\vspace{7pt}
\caption{\label{subfig:alias_all_mixed_suite2_90_10}
Mixed updates.
}
\end{subfigure}
\caption{\label{fig:experiments_suite2_90_10}
The average time to handle a single update on on-the-fly data-dependence analysis (\subref{subfig:datadependence_all_incr_suite2_90_10}, \subref{subfig:datadependence_all_decr_suite2_90_10}, \subref{subfig:datadependence_all_mixed_suite2_90_10}) and on-the-fly alias analysis (\subref{subfig:alias_all_incr_suite2_90_10}, \subref{subfig:alias_all_decr_suite2_90_10}, \subref{subfig:alias_all_mixed_suite2_90_10}) for sequence files generated with 90-10 split of original graph.
Note that all results are in log-scale.
}
\end{figure}

\Paragraph{Results on data-dependence analysis.}
Our experimental results on on-the-fly the data-dependence analysis are shown in \cref{fig:experiments_suite2_90_10} (left column).
We see that in all three settings (incremental/decremental/mixed), the dynamic-Datalog approach is measurably faster than $\OfflineAlgo$.
Although, in the worst case, the Datalog solver has worse complexity than $\OfflineAlgo$, 
the nature of the updates on the analyses graphs allows an efficient Datalog solver dedicated to dynamic updates to perform faster, as it never exhibits its worst-case performance.
Still, the time spent by the Datalog solver is typically quite large, given that we are looking at updates by a \emph{single edge} (i.e., a single source-code line).

On the other hand, our $\DynamicAlgo$ spends much less time than both $\OfflineAlgo$ and the dynamic Datalog solver, leading to typical speedups between two and three orders of magnitude.
The only difference is in \texttt{xml}, where $\DynamicAlgo$ appears to spend more time than in the rest of the benchmarks.
We have identified that this benchmark has a disproportionately large number of parenthesis symbols, which is the likely cause of this behavior.
Still, $\DynamicAlgo$ is by far the fastest approach also on this benchmark.
Naturally, decremental updates yield the least speedup on average.
This is expected given how much more complex our procedure for handling deletions is compared to that of handling insertions.
Indeed, when processing a $\InsertEdge(u,v,\CloseParenthesis)$ update, our algorithm only merges DSCCs, while
processing a  $\DeleteEdge(u,v,\CloseParenthesis)$ update both splits and merges DSCCs.
Nevertheless, the speedups are still in the range of two orders of magnitude, enough to render $\DynamicAlgo$ the clear best approach overall.

\Paragraph{Results on alias analysis.}
Our experimental results on on-the-fly the alias analysis are shown in \cref{fig:experiments_suite2_90_10} (right column).
In contrast to the data-dependence analysis, here the dynamic-Datalog approach is always a bit worse that $\OfflineAlgo$,
indicating that reachability patterns in these graphs are more challenging; enough so to push the dynamic-Datalog solver to worse performance than the from-scratch $\OfflineAlgo$.
Note, also, that the analysis times here are considerably larger than in the case of data-dependence analysis.

$\DynamicAlgo$ is again the best-performing approach by far, consistently by three orders of magnitude.
In several benchmarks, its running time is not visible despite the log-scale of the plots.
Finally, we again observe that decremental updates are overall more challenging that incremental updates.

\Paragraph{In summary.}
Our experiments clearly show that $\DynamicAlgo$ is the right approach to on-the-fly static analyses formulated as bidirected Dyck reachability, giving several orders of magnitude speedups over both
(i)~the offline algorithm that is dedicated (and optimal) for bidirected Dyck reachability, and
(ii)~the dynamic Datalog approach that is dedicated to dynamic updates (but agnostic to the setting of bidirected Dyck reachability).
Although the worst-case running time of $\DynamicAlgo$ is linear, we rarely observed this in our experiments.
Instead, the time cost of each update is barely (if at all) noticeable to the human eye, and thus the algorithm is suitable for continuous analysis, e.g., integrated inside an IDE.

\section{Related Work}\label{sec:related_work}

The  importance of Dyck reachability in static analyses has lead to a systematic study of its complexity in various settings.
The problem has a simple $O(n^3)$ upper bound~\cite{Yannakakis90}, which has resisted improvements beyond logarithmic~\cite{Chaudhuri2008}.
Due to this reason, the complexity of Dyck reachability has also been studied in terms of lower bounds.
Even for a single-pair query, the problem has been known to be 2NPDA-hard~\cite{Heintze97}, while its combinatorial cubic hardness persists even on constant-treewidth graphs~\cite{Chatterjee18}.
All-pairs Dyck reachability was recently shown to have a conditional $n^{2.5}$ lower bound based on popular complexity-theoretic hypotheses~\cite{Koutris2023}.
Despite the cubic hardness of the general problem, it is known to have sub-cubic certificates for both positive and negative instances~\cite{Chistikov2022}.
All-pairs reachability with $k=1$ parenthesis (aka one-counter systems) was recently shown to admit  an $O(n^{\omega}\cdot \log^2n)$ bound~\cite{Mathiasen2021}, where $\omega$ is the matrix multiplication exponent, which is also tight even for single-pair queries~\cite{Hansen2021}.
We refer to~\cite{Pavlogiannis2023} for a recent survey on this rich problem.

The technique developed in this work for on-the-fly bidirected Dyck reachability is motivated by the same setting on undirected connectivity, which has been studied extensively.
We leverage the data structure developed in~\cite{Eppstein1997} for maintaining the PDSCCs of a graph as connected components of the underlying primary graph (which is indeed undirected), as well as a sparsification technique that is specific to our setting (even though the concept of sparsification was developed in~\cite{Eppstein1997} to handle undirected connectivity).
It would be interesting to investigate whether techniques from undirected connectivity can be used further in our setting so as to reduce the complexity to sublinear (as is the status quo in undirected connectivity).
However, as a single update can create or merge $\Theta(n)$ DSCCs, sublinear complexity can only arise in the amortized sense, or by not requiring the explicit maintenance of DSCCs throughout updates.
Though we can easily obtain an $O(\alpha(n))$ amortized insertion cost for our technique (by amortizing over the linear cost of deletes), tighter bounds appear non-trivial and are open to interesting future work.

Bidirected Dyck reachability is very similar to unification closure, though the later problem is typically phrased with labels on the nodes as opposed to the edges of the input graph~\cite{Kanellakis1989}.
Unification closure has found widespread applications in programming languages, such as in 
efficient binding-time analysis~\cite{Henglein1991},
simple first-order type-inference~\cite{spa},
efficient dynamic type inference for LISP~\cite{Henglein1992}, as well as Steensgaard's famous pointer analysis~\cite{Steensgaard96}.
The insights developed here for the on-the-fly setting are likely extendable to these applications, though this merits further investigation.

In this work we have focused on online analyses where the analyzed source code changes frequently.
Another ``on-the-fly'' style of analysis is that of \emph{on-demand} analysis.
Here the analyzed program remains unchanged, but the task is to answer a sequence of analysis queries that is not known in advance.
This setting has been studied a lot in the context of pointer analysis~\cite{Heintze01b,Sridharan2005,Zheng2008,Yan11} and data flow analysis~\cite{Horwitz95,Naeem2010,Lerch2015}; the latter have also been efficiently parameterized by treewidth~\cite{Chatterjee2015b,Chatterjee2016,Chatterjee2020} and treedepth~\cite{Goharshady2023}.
\section{Conclusion}\label{sec:conclusion}
On-the-fly analysis is a very appealing feature to static analyzers, in order to run in real-time during code development and incorporate the constant changes in the source code.
However, on-the-fly analysis algorithms with provable complexity benefits have been missing, due to the intricacies of typical static analyses.
In this work we have considered a wide class of static analyses phrased as bidirected Dyck reachability.
We have developed a dynamic algorithm for handling the addition and removal of source code lines, with a provable guarantee that each such modification  takes only (nearly) linear time in the worst case. 
Our experiments show that our dynamic algorithm is extremely performant in practice, with a clear advantage over the (optimal) offline static analysis algorithm, as well as dynamic Datalog approaches, with typical speedups of three orders of magnitude.
From a practical standpoint, our results indicate that our algorithm can directly support lightweight analyses inside IDEs at a time cost that is barely (if at all) noticeable to the developers.

\clearpage

\begin{acks}
Andreas Pavlogiannis was partially supported by a research grant (VIL42117) from VILLUM FONDEN.
S. Krishna was partially supported  by the SERB MATRICS grant MTR/2019/000095. 
\end{acks}


\bibliography{bibliography}


\begin{thebibliography}{68}


\ifx \showCODEN    \undefined \def \showCODEN     #1{\unskip}     \fi
\ifx \showDOI      \undefined \def \showDOI       #1{#1}\fi
\ifx \showISBNx    \undefined \def \showISBNx     #1{\unskip}     \fi
\ifx \showISBNxiii \undefined \def \showISBNxiii  #1{\unskip}     \fi
\ifx \showISSN     \undefined \def \showISSN      #1{\unskip}     \fi
\ifx \showLCCN     \undefined \def \showLCCN      #1{\unskip}     \fi
\ifx \shownote     \undefined \def \shownote      #1{#1}          \fi
\ifx \showarticletitle \undefined \def \showarticletitle #1{#1}   \fi
\ifx \showURL      \undefined \def \showURL       {\relax}        \fi
\providecommand\bibfield[2]{#2}
\providecommand\bibinfo[2]{#2}
\providecommand\natexlab[1]{#1}
\providecommand\showeprint[2][]{arXiv:#2}

\bibitem[\protect\citeauthoryear{??}{Wal}{2003}]%
        {Wala}
 \bibinfo{year}{2003}\natexlab{}.
\newblock \bibinfo{title}{T. J. Watson Libraries for Analysis (WALA)}.
\newblock \bibinfo{howpublished}{https://github.com}.
\newblock


\bibitem[\protect\citeauthoryear{??}{SPE}{2008}]%
        {SPECjvm2008}
 \bibinfo{year}{2008}\natexlab{}.
\newblock \bibinfo{title}{SPECjvm2008 Benchmark Suit}.
\newblock \bibinfo{howpublished}{http://www.spec.org/jvm2008/}.
\newblock


\bibitem[\protect\citeauthoryear{Arnold}{Arnold}{1996}]%
        {Arnold96}
\bibfield{author}{\bibinfo{person}{Robert~S. Arnold}.}
  \bibinfo{year}{1996}\natexlab{}.
\newblock \bibinfo{booktitle}{\emph{Software Change Impact Analysis}}.
\newblock \bibinfo{publisher}{IEEE Computer Society Press},
  \bibinfo{address}{Los Alamitos, CA, USA}.
\newblock
\showISBNx{0818673842}


\bibitem[\protect\citeauthoryear{Arzt and Bodden}{Arzt and Bodden}{2014}]%
        {Arzt2014}
\bibfield{author}{\bibinfo{person}{Steven Arzt} {and} \bibinfo{person}{Eric
  Bodden}.} \bibinfo{year}{2014}\natexlab{}.
\newblock \showarticletitle{Reviser: Efficiently Updating IDE-/IFDS-Based
  Data-Flow Analyses in Response to Incremental Program Changes}. In
  \bibinfo{booktitle}{\emph{Proceedings of the 36th International Conference on
  Software Engineering}} (Hyderabad, India) \emph{(\bibinfo{series}{ICSE
  2014})}. \bibinfo{publisher}{Association for Computing Machinery},
  \bibinfo{address}{New York, NY, USA}, \bibinfo{pages}{288–298}.
\newblock
\showISBNx{9781450327565}
\urldef\tempurl%
\url{https://doi.org/10.1145/2568225.2568243}
\showDOI{\tempurl}


\bibitem[\protect\citeauthoryear{Blackburn, Garner, Hoffman, Khan, McKinley,
  Bentzur, Diwan, Feinberg, Frampton, Guyer, Hirzel, Hosking, Jump, Lee, Moss,
  Phansalkar, Stefanovi\'{c}, {VanDrunen}, von Dincklage, and
  Wiedermann}{Blackburn et~al\mbox{.}}{2006}]%
        {DaCapo:paper}
\bibfield{author}{\bibinfo{person}{S.~M. Blackburn}, \bibinfo{person}{R.
  Garner}, \bibinfo{person}{C. Hoffman}, \bibinfo{person}{A.~M. Khan},
  \bibinfo{person}{K.~S. McKinley}, \bibinfo{person}{R. Bentzur},
  \bibinfo{person}{A. Diwan}, \bibinfo{person}{D. Feinberg},
  \bibinfo{person}{D. Frampton}, \bibinfo{person}{S.~Z. Guyer},
  \bibinfo{person}{M. Hirzel}, \bibinfo{person}{A. Hosking},
  \bibinfo{person}{M. Jump}, \bibinfo{person}{H. Lee},
  \bibinfo{person}{J.~E.~B. Moss}, \bibinfo{person}{A. Phansalkar},
  \bibinfo{person}{D. Stefanovi\'{c}}, \bibinfo{person}{T. {VanDrunen}},
  \bibinfo{person}{D. von Dincklage}, {and} \bibinfo{person}{B. Wiedermann}.}
  \bibinfo{year}{2006}\natexlab{}.
\newblock \showarticletitle{The {DaCapo} Benchmarks: {J}ava Benchmarking
  Development and Analysis}. In \bibinfo{booktitle}{\emph{OOPSLA '06:
  Proceedings of the 21st annual ACM SIGPLAN conference on Object-Oriented
  Programing, Systems, Languages, and Applications}} (Portland, OR, USA).
  \bibinfo{publisher}{ACM Press}, \bibinfo{address}{New York, NY, USA},
  \bibinfo{pages}{169--190}.
\newblock
\urldef\tempurl%
\url{https://doi.org/10.1145/1167473.1167488}
\showDOI{\tempurl}


\bibitem[\protect\citeauthoryear{Bodden}{Bodden}{2012}]%
        {Bodden12}
\bibfield{author}{\bibinfo{person}{Eric Bodden}.}
  \bibinfo{year}{2012}\natexlab{}.
\newblock \showarticletitle{Inter-procedural Data-flow Analysis with IFDS/IDE
  and Soot}. In \bibinfo{booktitle}{\emph{SOAP}}. \bibinfo{publisher}{ACM},
  \bibinfo{address}{New York, NY, USA}.
\newblock


\bibitem[\protect\citeauthoryear{Bravenboer and Smaragdakis}{Bravenboer and
  Smaragdakis}{2009}]%
        {Bravenboer2009}
\bibfield{author}{\bibinfo{person}{Martin Bravenboer} {and}
  \bibinfo{person}{Yannis Smaragdakis}.} \bibinfo{year}{2009}\natexlab{}.
\newblock \showarticletitle{Strictly Declarative Specification of Sophisticated
  Points-to Analyses}.
\newblock \bibinfo{journal}{\emph{SIGPLAN Not.}} \bibinfo{volume}{44},
  \bibinfo{number}{10} (\bibinfo{date}{oct} \bibinfo{year}{2009}),
  \bibinfo{pages}{243–262}.
\newblock
\showISSN{0362-1340}
\urldef\tempurl%
\url{https://doi.org/10.1145/1639949.1640108}
\showDOI{\tempurl}


\bibitem[\protect\citeauthoryear{Burke and Ryder}{Burke and Ryder}{1990}]%
        {Burke1990}
\bibfield{author}{\bibinfo{person}{M.G. Burke} {and} \bibinfo{person}{B.G.
  Ryder}.} \bibinfo{year}{1990}\natexlab{}.
\newblock \showarticletitle{A critical analysis of incremental iterative data
  flow analysis algorithms}.
\newblock \bibinfo{journal}{\emph{IEEE Transactions on Software Engineering}}
  \bibinfo{volume}{16}, \bibinfo{number}{7} (\bibinfo{year}{1990}),
  \bibinfo{pages}{723--728}.
\newblock
\urldef\tempurl%
\url{https://doi.org/10.1109/32.56098}
\showDOI{\tempurl}


\bibitem[\protect\citeauthoryear{{Cetti Hansen}, {Husted Kjelstrøm}, and
  Pavlogiannis}{{Cetti Hansen} et~al\mbox{.}}{2021}]%
        {Hansen2021}
\bibfield{author}{\bibinfo{person}{Jakob {Cetti Hansen}}, \bibinfo{person}{Adam
  {Husted Kjelstrøm}}, {and} \bibinfo{person}{Andreas Pavlogiannis}.}
  \bibinfo{year}{2021}\natexlab{}.
\newblock \showarticletitle{Tight bounds for reachability problems on
  one-counter and pushdown systems}.
\newblock \bibinfo{journal}{\emph{Inform. Process. Lett.}}
  \bibinfo{volume}{171} (\bibinfo{year}{2021}), \bibinfo{pages}{106135}.
\newblock
\showISSN{0020-0190}
\urldef\tempurl%
\url{https://doi.org/10.1016/j.ipl.2021.106135}
\showDOI{\tempurl}


\bibitem[\protect\citeauthoryear{Chatterjee, Choudhary, and
  Pavlogiannis}{Chatterjee et~al\mbox{.}}{2018}]%
        {Chatterjee18}
\bibfield{author}{\bibinfo{person}{Krishnendu Chatterjee},
  \bibinfo{person}{Bhavya Choudhary}, {and} \bibinfo{person}{Andreas
  Pavlogiannis}.} \bibinfo{year}{2018}\natexlab{}.
\newblock \showarticletitle{Optimal Dyck Reachability for Data-Dependence and
  Alias Analysis}.
\newblock \bibinfo{journal}{\emph{Proc. ACM Program. Lang.}}
  \bibinfo{volume}{2}, \bibinfo{number}{POPL}, Article \bibinfo{articleno}{30}
  (\bibinfo{date}{Dec.} \bibinfo{year}{2018}), \bibinfo{numpages}{30}~pages.
\newblock


\bibitem[\protect\citeauthoryear{Chatterjee, Goharshady, Ibsen-Jensen, and
  Pavlogiannis}{Chatterjee et~al\mbox{.}}{2016}]%
        {Chatterjee2016}
\bibfield{author}{\bibinfo{person}{Krishnendu Chatterjee},
  \bibinfo{person}{Amir~Kafshdar Goharshady}, \bibinfo{person}{Rasmus
  Ibsen-Jensen}, {and} \bibinfo{person}{Andreas Pavlogiannis}.}
  \bibinfo{year}{2016}\natexlab{}.
\newblock \showarticletitle{Algorithms for Algebraic Path Properties in
  Concurrent Systems of Constant Treewidth Components}.
\newblock \bibinfo{journal}{\emph{SIGPLAN Not.}} \bibinfo{volume}{51},
  \bibinfo{number}{1} (\bibinfo{date}{jan} \bibinfo{year}{2016}),
  \bibinfo{pages}{733–747}.
\newblock
\showISSN{0362-1340}
\urldef\tempurl%
\url{https://doi.org/10.1145/2914770.2837624}
\showDOI{\tempurl}


\bibitem[\protect\citeauthoryear{Chatterjee, Goharshady, Ibsen-Jensen, and
  Pavlogiannis}{Chatterjee et~al\mbox{.}}{2020}]%
        {Chatterjee2020}
\bibfield{author}{\bibinfo{person}{Krishnendu Chatterjee},
  \bibinfo{person}{Amir~Kafshdar Goharshady}, \bibinfo{person}{Rasmus
  Ibsen-Jensen}, {and} \bibinfo{person}{Andreas Pavlogiannis}.}
  \bibinfo{year}{2020}\natexlab{}.
\newblock \showarticletitle{Optimal and Perfectly Parallel Algorithms for
  On-demand Data-Flow Analysis}. In \bibinfo{booktitle}{\emph{Programming
  Languages and Systems}}, \bibfield{editor}{\bibinfo{person}{Peter
  M{\"u}ller}} (Ed.). \bibinfo{publisher}{Springer International Publishing},
  \bibinfo{address}{Cham}, \bibinfo{pages}{112--140}.
\newblock
\showISBNx{978-3-030-44914-8}


\bibitem[\protect\citeauthoryear{Chatterjee, Ibsen-Jensen, Pavlogiannis, and
  Goyal}{Chatterjee et~al\mbox{.}}{2015}]%
        {Chatterjee2015b}
\bibfield{author}{\bibinfo{person}{Krishnendu Chatterjee},
  \bibinfo{person}{Rasmus Ibsen-Jensen}, \bibinfo{person}{Andreas
  Pavlogiannis}, {and} \bibinfo{person}{Prateesh Goyal}.}
  \bibinfo{year}{2015}\natexlab{}.
\newblock \showarticletitle{Faster Algorithms for Algebraic Path Properties in
  Recursive State Machines with Constant Treewidth}. In
  \bibinfo{booktitle}{\emph{Proceedings of the 42nd Annual ACM SIGPLAN-SIGACT
  Symposium on Principles of Programming Languages}} (Mumbai, India)
  \emph{(\bibinfo{series}{POPL '15})}. \bibinfo{publisher}{Association for
  Computing Machinery}, \bibinfo{address}{New York, NY, USA},
  \bibinfo{pages}{97–109}.
\newblock
\showISBNx{9781450333009}
\urldef\tempurl%
\url{https://doi.org/10.1145/2676726.2676979}
\showDOI{\tempurl}


\bibitem[\protect\citeauthoryear{Chaudhuri}{Chaudhuri}{2008}]%
        {Chaudhuri2008}
\bibfield{author}{\bibinfo{person}{Swarat Chaudhuri}.}
  \bibinfo{year}{2008}\natexlab{}.
\newblock \showarticletitle{Subcubic Algorithms for Recursive State Machines}.
\newblock \bibinfo{journal}{\emph{SIGPLAN Not.}} \bibinfo{volume}{43},
  \bibinfo{number}{1} (\bibinfo{date}{Jan.} \bibinfo{year}{2008}),
  \bibinfo{pages}{159–169}.
\newblock
\showISSN{0362-1340}
\urldef\tempurl%
\url{https://doi.org/10.1145/1328897.1328460}
\showDOI{\tempurl}


\bibitem[\protect\citeauthoryear{Chistikov, Majumdar, and Schepper}{Chistikov
  et~al\mbox{.}}{2022}]%
        {Chistikov2022}
\bibfield{author}{\bibinfo{person}{Dmitry Chistikov}, \bibinfo{person}{Rupak
  Majumdar}, {and} \bibinfo{person}{Philipp Schepper}.}
  \bibinfo{year}{2022}\natexlab{}.
\newblock \showarticletitle{Subcubic Certificates for CFL Reachability}.
\newblock \bibinfo{journal}{\emph{Proc. ACM Program. Lang.}}
  \bibinfo{volume}{6}, \bibinfo{number}{POPL}, Article \bibinfo{articleno}{41}
  (\bibinfo{date}{jan} \bibinfo{year}{2022}), \bibinfo{numpages}{29}~pages.
\newblock
\urldef\tempurl%
\url{https://doi.org/10.1145/3498702}
\showDOI{\tempurl}


\bibitem[\protect\citeauthoryear{Eppstein, Galil, Italiano, and
  Nissenzweig}{Eppstein et~al\mbox{.}}{1997}]%
        {Eppstein1997}
\bibfield{author}{\bibinfo{person}{David Eppstein}, \bibinfo{person}{Zvi
  Galil}, \bibinfo{person}{Giuseppe~F. Italiano}, {and} \bibinfo{person}{Amnon
  Nissenzweig}.} \bibinfo{year}{1997}\natexlab{}.
\newblock \showarticletitle{Sparsification—a Technique for Speeding up
  Dynamic Graph Algorithms}.
\newblock \bibinfo{journal}{\emph{J. ACM}} \bibinfo{volume}{44},
  \bibinfo{number}{5} (\bibinfo{date}{sep} \bibinfo{year}{1997}),
  \bibinfo{pages}{669–696}.
\newblock
\showISSN{0004-5411}
\urldef\tempurl%
\url{https://doi.org/10.1145/265910.265914}
\showDOI{\tempurl}


\bibitem[\protect\citeauthoryear{Ganardi, Majumdar, Pavlogiannis, Sch\"{u}tze,
  and Zetzsche}{Ganardi et~al\mbox{.}}{2022b}]%
        {Ganardi2022}
\bibfield{author}{\bibinfo{person}{Moses Ganardi}, \bibinfo{person}{Rupak
  Majumdar}, \bibinfo{person}{Andreas Pavlogiannis}, \bibinfo{person}{Lia
  Sch\"{u}tze}, {and} \bibinfo{person}{Georg Zetzsche}.}
  \bibinfo{year}{2022}\natexlab{b}.
\newblock \showarticletitle{{Reachability in Bidirected Pushdown VASS}}. In
  \bibinfo{booktitle}{\emph{49th International Colloquium on Automata,
  Languages, and Programming (ICALP 2022)}} \emph{(\bibinfo{series}{Leibniz
  International Proceedings in Informatics (LIPIcs)},
  Vol.~\bibinfo{volume}{229})}, \bibfield{editor}{\bibinfo{person}{Miko{\l}aj
  Boja\'{n}czyk}, \bibinfo{person}{Emanuela Merelli}, {and}
  \bibinfo{person}{David~P. Woodruff}} (Eds.). \bibinfo{publisher}{Schloss
  Dagstuhl -- Leibniz-Zentrum f{\"u}r Informatik}, \bibinfo{address}{Dagstuhl,
  Germany}, \bibinfo{pages}{124:1--124:20}.
\newblock
\showISBNx{978-3-95977-235-8}
\showISSN{1868-8969}
\urldef\tempurl%
\url{https://doi.org/10.4230/LIPIcs.ICALP.2022.124}
\showDOI{\tempurl}


\bibitem[\protect\citeauthoryear{Ganardi, Majumdar, and Zetzsche}{Ganardi
  et~al\mbox{.}}{2022a}]%
        {Ganardi2022a}
\bibfield{author}{\bibinfo{person}{Moses Ganardi}, \bibinfo{person}{Rupak
  Majumdar}, {and} \bibinfo{person}{Georg Zetzsche}.}
  \bibinfo{year}{2022}\natexlab{a}.
\newblock \showarticletitle{The {{Complexity}} of {{Bidirected Reachability}}
  in {{Valence Systems}}}. In \bibinfo{booktitle}{\emph{Proceedings of the 37th
  {{Annual ACM}}/{{IEEE Symposium}} on {{Logic}} in {{Computer Science}}}}
  \emph{(\bibinfo{series}{{{LICS}} '22})}. \bibinfo{publisher}{{Association for
  Computing Machinery}}, \bibinfo{address}{{New York, NY, USA}},
  \bibinfo{pages}{1--15}.
\newblock
\showISBNx{978-1-4503-9351-5}
\urldef\tempurl%
\url{https://doi.org/10.1145/3531130.3533345}
\showDOI{\tempurl}


\bibitem[\protect\citeauthoryear{Goharshady and Zaher}{Goharshady and
  Zaher}{2023}]%
        {Goharshady2023}
\bibfield{author}{\bibinfo{person}{Amir~Kafshdar Goharshady} {and}
  \bibinfo{person}{Ahmed~Khaled Zaher}.} \bibinfo{year}{2023}\natexlab{}.
\newblock \showarticletitle{Efficient Interprocedural Data-Flow Analysis Using
  Treedepth and Treewidth}. In \bibinfo{booktitle}{\emph{Verification, Model
  Checking, and Abstract Interpretation}},
  \bibfield{editor}{\bibinfo{person}{Cezara Dragoi}, \bibinfo{person}{Michael
  Emmi}, {and} \bibinfo{person}{Jingbo Wang}} (Eds.).
  \bibinfo{publisher}{Springer Nature Switzerland}, \bibinfo{address}{Cham},
  \bibinfo{pages}{177--202}.
\newblock
\showISBNx{978-3-031-24950-1}


\bibitem[\protect\citeauthoryear{Heintze and McAllester}{Heintze and
  McAllester}{1997}]%
        {Heintze97}
\bibfield{author}{\bibinfo{person}{Nevin Heintze} {and} \bibinfo{person}{David
  McAllester}.} \bibinfo{year}{1997}\natexlab{}.
\newblock \showarticletitle{On the Cubic Bottleneck in Subtyping and Flow
  Analysis}. In \bibinfo{booktitle}{\emph{Proceedings of the 12th Annual IEEE
  Symposium on Logic in Computer Science}} \emph{(\bibinfo{series}{LICS '97})}.
  \bibinfo{publisher}{IEEE Computer Society}, \bibinfo{address}{Washington, DC,
  USA}, \bibinfo{pages}{342--}.
\newblock
\showISBNx{0-8186-7925-5}
\urldef\tempurl%
\url{http://dl.acm.org/citation.cfm?id=788019.788876}
\showURL{%
\tempurl}


\bibitem[\protect\citeauthoryear{Heintze and Tardieu}{Heintze and
  Tardieu}{2001}]%
        {Heintze01b}
\bibfield{author}{\bibinfo{person}{Nevin Heintze} {and}
  \bibinfo{person}{Olivier Tardieu}.} \bibinfo{year}{2001}\natexlab{}.
\newblock \showarticletitle{Demand-Driven Pointer Analysis}. In
  \bibinfo{booktitle}{\emph{Proceedings of the ACM SIGPLAN 2001 Conference on
  Programming Language Design and Implementation}} (Snowbird, Utah, USA)
  \emph{(\bibinfo{series}{PLDI ’01})}. \bibinfo{publisher}{Association for
  Computing Machinery}, \bibinfo{address}{New York, NY, USA},
  \bibinfo{pages}{24–34}.
\newblock
\showISBNx{1581134142}
\urldef\tempurl%
\url{https://doi.org/10.1145/378795.378802}
\showDOI{\tempurl}


\bibitem[\protect\citeauthoryear{Henglein}{Henglein}{1991}]%
        {Henglein1991}
\bibfield{author}{\bibinfo{person}{Fritz Henglein}.}
  \bibinfo{year}{1991}\natexlab{}.
\newblock \showarticletitle{Efficient Type Inference for Higher-Order
  Binding-Time Analysis}. In \bibinfo{booktitle}{\emph{Functional {{Programming
  Languages}} and {{Computer Architecture}}}} \emph{(\bibinfo{series}{Lecture
  {{Notes}} in {{Computer Science}}})}, \bibfield{editor}{\bibinfo{person}{John
  Hughes}} (Ed.). \bibinfo{publisher}{{Springer}}, \bibinfo{address}{{Berlin,
  Heidelberg}}, \bibinfo{pages}{448--472}.
\newblock
\showISBNx{978-3-540-47599-6}
\urldef\tempurl%
\url{https://doi.org/10.1007/3540543961_22}
\showDOI{\tempurl}


\bibitem[\protect\citeauthoryear{Henglein}{Henglein}{1992}]%
        {Henglein1992}
\bibfield{author}{\bibinfo{person}{Fritz Henglein}.}
  \bibinfo{year}{1992}\natexlab{}.
\newblock \showarticletitle{Global Tagging Optimization by Type Inference}. In
  \bibinfo{booktitle}{\emph{Proceedings of the 1992 {{ACM}} Conference on
  {{LISP}} and Functional Programming}} \emph{(\bibinfo{series}{{{LFP}} '92})}.
  \bibinfo{publisher}{{Association for Computing Machinery}},
  \bibinfo{address}{{New York, NY, USA}}, \bibinfo{pages}{205--215}.
\newblock
\showISBNx{978-0-89791-481-9}
\urldef\tempurl%
\url{https://doi.org/10.1145/141471.141542}
\showDOI{\tempurl}


\bibitem[\protect\citeauthoryear{Holm, de~Lichtenberg, and Thorup}{Holm
  et~al\mbox{.}}{2001}]%
        {Holm2001}
\bibfield{author}{\bibinfo{person}{Jacob Holm}, \bibinfo{person}{Kristian de
  Lichtenberg}, {and} \bibinfo{person}{Mikkel Thorup}.}
  \bibinfo{year}{2001}\natexlab{}.
\newblock \showarticletitle{Poly-Logarithmic Deterministic Fully-Dynamic
  Algorithms for Connectivity, Minimum Spanning Tree, 2-Edge, and
  Biconnectivity}.
\newblock \bibinfo{journal}{\emph{J. ACM}} \bibinfo{volume}{48},
  \bibinfo{number}{4} (\bibinfo{date}{jul} \bibinfo{year}{2001}),
  \bibinfo{pages}{723–760}.
\newblock
\showISSN{0004-5411}
\urldef\tempurl%
\url{https://doi.org/10.1145/502090.502095}
\showDOI{\tempurl}


\bibitem[\protect\citeauthoryear{Horwitz, Reps, and Sagiv}{Horwitz
  et~al\mbox{.}}{1995}]%
        {Horwitz95}
\bibfield{author}{\bibinfo{person}{Susan Horwitz}, \bibinfo{person}{Thomas
  Reps}, {and} \bibinfo{person}{Mooly Sagiv}.} \bibinfo{year}{1995}\natexlab{}.
\newblock \showarticletitle{Demand Interprocedural Dataflow Analysis}.
\newblock \bibinfo{journal}{\emph{SIGSOFT Softw. Eng. Notes}}
  (\bibinfo{year}{1995}).
\newblock


\bibitem[\protect\citeauthoryear{Huang, Dong, Milanova, and Dolby}{Huang
  et~al\mbox{.}}{2015}]%
        {Huang2015}
\bibfield{author}{\bibinfo{person}{Wei Huang}, \bibinfo{person}{Yao Dong},
  \bibinfo{person}{Ana Milanova}, {and} \bibinfo{person}{Julian Dolby}.}
  \bibinfo{year}{2015}\natexlab{}.
\newblock \showarticletitle{Scalable and Precise Taint Analysis for Android}.
  In \bibinfo{booktitle}{\emph{Proceedings of the 2015 International Symposium
  on Software Testing and Analysis}} (Baltimore, MD, USA)
  \emph{(\bibinfo{series}{ISSTA 2015})}. \bibinfo{publisher}{Association for
  Computing Machinery}, \bibinfo{address}{New York, NY, USA},
  \bibinfo{pages}{106–117}.
\newblock
\showISBNx{9781450336208}
\urldef\tempurl%
\url{https://doi.org/10.1145/2771783.2771803}
\showDOI{\tempurl}


\bibitem[\protect\citeauthoryear{Kanellakis and Revesz}{Kanellakis and
  Revesz}{1989}]%
        {Kanellakis1989}
\bibfield{author}{\bibinfo{person}{Paris~C. Kanellakis} {and}
  \bibinfo{person}{Peter~Z. Revesz}.} \bibinfo{year}{1989}\natexlab{}.
\newblock \showarticletitle{On the Relationship of Congruence Closureand
  Unification}.
\newblock \bibinfo{journal}{\emph{Journal of Symbolic Computation}}
  \bibinfo{volume}{7}, \bibinfo{number}{3-4} (\bibinfo{date}{March}
  \bibinfo{year}{1989}), \bibinfo{pages}{427--444}.
\newblock
\showISSN{07477171}
\urldef\tempurl%
\url{https://doi.org/10.1016/S0747-7171(89)80018-5}
\showDOI{\tempurl}


\bibitem[\protect\citeauthoryear{Kjelstr\o{}m and Pavlogiannis}{Kjelstr\o{}m
  and Pavlogiannis}{2022}]%
        {Kjelstrom2022}
\bibfield{author}{\bibinfo{person}{Adam~Husted Kjelstr\o{}m} {and}
  \bibinfo{person}{Andreas Pavlogiannis}.} \bibinfo{year}{2022}\natexlab{}.
\newblock \showarticletitle{The Decidability and Complexity of Interleaved
  Bidirected Dyck Reachability}.
\newblock \bibinfo{journal}{\emph{Proc. ACM Program. Lang.}}
  \bibinfo{volume}{6}, \bibinfo{number}{POPL}, Article \bibinfo{articleno}{12}
  (\bibinfo{date}{jan} \bibinfo{year}{2022}), \bibinfo{numpages}{26}~pages.
\newblock
\urldef\tempurl%
\url{https://doi.org/10.1145/3498673}
\showDOI{\tempurl}


\bibitem[\protect\citeauthoryear{Koutris and Deep}{Koutris and Deep}{2023}]%
        {Koutris2023}
\bibfield{author}{\bibinfo{person}{Paraschos Koutris} {and}
  \bibinfo{person}{Shaleen Deep}.} \bibinfo{year}{2023}\natexlab{}.
\newblock \showarticletitle{The Fine-Grained Complexity of CFL Reachability}.
\newblock \bibinfo{journal}{\emph{Proc. ACM Program. Lang.}}
  \bibinfo{volume}{7}, \bibinfo{number}{POPL}, Article \bibinfo{articleno}{59}
  (\bibinfo{date}{jan} \bibinfo{year}{2023}), \bibinfo{numpages}{27}~pages.
\newblock
\urldef\tempurl%
\url{https://doi.org/10.1145/3571252}
\showDOI{\tempurl}


\bibitem[\protect\citeauthoryear{Lerch, Sp\"{a}th, Bodden, and Mezini}{Lerch
  et~al\mbox{.}}{2015}]%
        {Lerch2015}
\bibfield{author}{\bibinfo{person}{Johannes Lerch}, \bibinfo{person}{Johannes
  Sp\"{a}th}, \bibinfo{person}{Eric Bodden}, {and} \bibinfo{person}{Mira
  Mezini}.} \bibinfo{year}{2015}\natexlab{}.
\newblock \showarticletitle{Access-Path Abstraction: Scaling Field-Sensitive
  Data-Flow Analysis with Unbounded Access Paths}. In
  \bibinfo{booktitle}{\emph{Proceedings of the 30th IEEE/ACM International
  Conference on Automated Software Engineering}} (Lincoln, Nebraska)
  \emph{(\bibinfo{series}{ASE '15})}. \bibinfo{publisher}{IEEE Press},
  \bibinfo{pages}{619–629}.
\newblock
\showISBNx{9781509000241}
\urldef\tempurl%
\url{https://doi.org/10.1109/ASE.2015.9}
\showDOI{\tempurl}


\bibitem[\protect\citeauthoryear{Lhot\'{a}k and Hendren}{Lhot\'{a}k and
  Hendren}{2006}]%
        {Lhotak06}
\bibfield{author}{\bibinfo{person}{Ond\v{r}ej Lhot\'{a}k} {and}
  \bibinfo{person}{Laurie Hendren}.} \bibinfo{year}{2006}\natexlab{}.
\newblock \showarticletitle{Context-Sensitive Points-to Analysis: Is It Worth
  It?}. In \bibinfo{booktitle}{\emph{Proceedings of the 15th International
  Conference on Compiler Construction}} \emph{(\bibinfo{series}{CC})}.
  \bibinfo{pages}{47--64}.
\newblock


\bibitem[\protect\citeauthoryear{Li, Satya, and Zhang}{Li
  et~al\mbox{.}}{2022}]%
        {Li2022}
\bibfield{author}{\bibinfo{person}{Yuanbo Li}, \bibinfo{person}{Kris Satya},
  {and} \bibinfo{person}{Qirun Zhang}.} \bibinfo{year}{2022}\natexlab{}.
\newblock \showarticletitle{Efficient Algorithms for Dynamic Bidirected
  Dyck-Reachability}.
\newblock \bibinfo{journal}{\emph{Proc. ACM Program. Lang.}}
  \bibinfo{volume}{6}, \bibinfo{number}{POPL}, Article \bibinfo{articleno}{62}
  (\bibinfo{date}{jan} \bibinfo{year}{2022}), \bibinfo{numpages}{29}~pages.
\newblock
\urldef\tempurl%
\url{https://doi.org/10.1145/3498724}
\showDOI{\tempurl}


\bibitem[\protect\citeauthoryear{Li, Zhang, and Reps}{Li et~al\mbox{.}}{2020}]%
        {Li2020}
\bibfield{author}{\bibinfo{person}{Yuanbo Li}, \bibinfo{person}{Qirun Zhang},
  {and} \bibinfo{person}{Thomas Reps}.} \bibinfo{year}{2020}\natexlab{}.
\newblock \showarticletitle{Fast Graph Simplification for Interleaved
  Dyck-Reachability}. In \bibinfo{booktitle}{\emph{Proceedings of the 41st ACM
  SIGPLAN Conference on Programming Language Design and Implementation}}
  (London, UK) \emph{(\bibinfo{series}{PLDI 2020})}.
  \bibinfo{publisher}{Association for Computing Machinery},
  \bibinfo{address}{New York, NY, USA}, \bibinfo{pages}{780–793}.
\newblock
\showISBNx{9781450376136}
\urldef\tempurl%
\url{https://doi.org/10.1145/3385412.3386021}
\showDOI{\tempurl}


\bibitem[\protect\citeauthoryear{Liu, Huang, and Rauchwerger}{Liu
  et~al\mbox{.}}{2019}]%
        {Liu2019}
\bibfield{author}{\bibinfo{person}{Bozhen Liu}, \bibinfo{person}{Jeff Huang},
  {and} \bibinfo{person}{Lawrence Rauchwerger}.}
  \bibinfo{year}{2019}\natexlab{}.
\newblock \showarticletitle{Rethinking Incremental and Parallel Pointer
  Analysis}.
\newblock \bibinfo{journal}{\emph{ACM Trans. Program. Lang. Syst.}}
  \bibinfo{volume}{41}, \bibinfo{number}{1}, Article \bibinfo{articleno}{6}
  (\bibinfo{date}{mar} \bibinfo{year}{2019}), \bibinfo{numpages}{31}~pages.
\newblock
\showISSN{0164-0925}
\urldef\tempurl%
\url{https://doi.org/10.1145/3293606}
\showDOI{\tempurl}


\bibitem[\protect\citeauthoryear{Lu and Xue}{Lu and Xue}{2019}]%
        {Lu2019}
\bibfield{author}{\bibinfo{person}{Jingbo Lu} {and} \bibinfo{person}{Jingling
  Xue}.} \bibinfo{year}{2019}\natexlab{}.
\newblock \showarticletitle{Precision-Preserving yet Fast Object-Sensitive
  Pointer Analysis with Partial Context Sensitivity}.
\newblock \bibinfo{journal}{\emph{Proc. ACM Program. Lang.}}
  \bibinfo{volume}{3}, \bibinfo{number}{OOPSLA}, Article
  \bibinfo{articleno}{148} (\bibinfo{date}{Oct.} \bibinfo{year}{2019}),
  \bibinfo{numpages}{29}~pages.
\newblock
\urldef\tempurl%
\url{https://doi.org/10.1145/3360574}
\showDOI{\tempurl}


\bibitem[\protect\citeauthoryear{Madsen and Lhot\'{a}k}{Madsen and
  Lhot\'{a}k}{2020}]%
        {Madsen2020}
\bibfield{author}{\bibinfo{person}{Magnus Madsen} {and}
  \bibinfo{person}{Ond\v{r}ej Lhot\'{a}k}.} \bibinfo{year}{2020}\natexlab{}.
\newblock \showarticletitle{Fixpoints for the Masses: Programming with
  First-Class Datalog Constraints}.
\newblock \bibinfo{journal}{\emph{Proc. ACM Program. Lang.}}
  \bibinfo{volume}{4}, \bibinfo{number}{OOPSLA}, Article
  \bibinfo{articleno}{125} (\bibinfo{date}{nov} \bibinfo{year}{2020}),
  \bibinfo{numpages}{28}~pages.
\newblock
\urldef\tempurl%
\url{https://doi.org/10.1145/3428193}
\showDOI{\tempurl}


\bibitem[\protect\citeauthoryear{Mathiasen and Pavlogiannis}{Mathiasen and
  Pavlogiannis}{2021}]%
        {Mathiasen2021}
\bibfield{author}{\bibinfo{person}{Anders~Alnor Mathiasen} {and}
  \bibinfo{person}{Andreas Pavlogiannis}.} \bibinfo{year}{2021}\natexlab{}.
\newblock \showarticletitle{The Fine-Grained and Parallel Complexity of
  Andersen’s Pointer Analysis}.
\newblock \bibinfo{journal}{\emph{Proc. ACM Program. Lang.}}
  \bibinfo{volume}{5}, \bibinfo{number}{POPL}, Article \bibinfo{articleno}{34}
  (\bibinfo{date}{Jan.} \bibinfo{year}{2021}), \bibinfo{numpages}{29}~pages.
\newblock
\urldef\tempurl%
\url{https://doi.org/10.1145/3434315}
\showDOI{\tempurl}


\bibitem[\protect\citeauthoryear{Milanova}{Milanova}{2020}]%
        {Milanova2020}
\bibfield{author}{\bibinfo{person}{Ana Milanova}.}
  \bibinfo{year}{2020}\natexlab{}.
\newblock \showarticletitle{FlowCFL: Generalized Type-Based Reachability
  Analysis: Graph Reduction and Equivalence of CFL-Based and Type-Based
  Reachability}.
\newblock \bibinfo{journal}{\emph{Proc. ACM Program. Lang.}}
  \bibinfo{volume}{4}, \bibinfo{number}{OOPSLA}, Article
  \bibinfo{articleno}{178} (\bibinfo{date}{Nov.} \bibinfo{year}{2020}),
  \bibinfo{numpages}{29}~pages.
\newblock
\urldef\tempurl%
\url{https://doi.org/10.1145/3428246}
\showDOI{\tempurl}


\bibitem[\protect\citeauthoryear{M\o{}ller and Schwartzbach}{M\o{}ller and
  Schwartzbach}{2018}]%
        {spa}
\bibfield{author}{\bibinfo{person}{Anders M\o{}ller} {and}
  \bibinfo{person}{Michael~I. Schwartzbach}.} \bibinfo{year}{2018}\natexlab{}.
\newblock \bibinfo{booktitle}{\emph{Static Program Analysis}}.
\newblock \bibinfo{type}{{T}echnical {R}eport}.
  \bibinfo{institution}{Department of Computer Science, Aarhus University}.
\newblock
\urldef\tempurl%
\url{http://cs.au.dk/\~amoeller/spa/}
\showURL{%
\tempurl}


\bibitem[\protect\citeauthoryear{Naeem, Lhot{\'a}k, and Rodriguez}{Naeem
  et~al\mbox{.}}{2010}]%
        {Naeem2010}
\bibfield{author}{\bibinfo{person}{Nomair~A. Naeem},
  \bibinfo{person}{Ond{\v{r}}ej Lhot{\'a}k}, {and} \bibinfo{person}{Jonathan
  Rodriguez}.} \bibinfo{year}{2010}\natexlab{}.
\newblock \showarticletitle{Practical Extensions to the IFDS Algorithm}. In
  \bibinfo{booktitle}{\emph{Compiler Construction}},
  \bibfield{editor}{\bibinfo{person}{Rajiv Gupta}} (Ed.).
  \bibinfo{publisher}{Springer Berlin Heidelberg}, \bibinfo{address}{Berlin,
  Heidelberg}, \bibinfo{pages}{124--144}.
\newblock
\showISBNx{978-3-642-11970-5}


\bibitem[\protect\citeauthoryear{Pacak, Erdweg, and Szab\'{o}}{Pacak
  et~al\mbox{.}}{2020}]%
        {Pacak2020}
\bibfield{author}{\bibinfo{person}{Andr\'{e} Pacak}, \bibinfo{person}{Sebastian
  Erdweg}, {and} \bibinfo{person}{Tam\'{a}s Szab\'{o}}.}
  \bibinfo{year}{2020}\natexlab{}.
\newblock \showarticletitle{A Systematic Approach to Deriving Incremental Type
  Checkers}.
\newblock \bibinfo{journal}{\emph{Proc. ACM Program. Lang.}}
  \bibinfo{volume}{4}, \bibinfo{number}{OOPSLA}, Article
  \bibinfo{articleno}{127} (\bibinfo{date}{nov} \bibinfo{year}{2020}),
  \bibinfo{numpages}{28}~pages.
\newblock
\urldef\tempurl%
\url{https://doi.org/10.1145/3428195}
\showDOI{\tempurl}


\bibitem[\protect\citeauthoryear{Pavlogiannis}{Pavlogiannis}{2023}]%
        {Pavlogiannis2023}
\bibfield{author}{\bibinfo{person}{Andreas Pavlogiannis}.}
  \bibinfo{year}{2023}\natexlab{}.
\newblock \showarticletitle{CFL/Dyck Reachability: An Algorithmic Perspective}.
\newblock \bibinfo{journal}{\emph{ACM SIGLOG News}} \bibinfo{volume}{9},
  \bibinfo{number}{4} (\bibinfo{date}{feb} \bibinfo{year}{2023}),
  \bibinfo{pages}{5–25}.
\newblock
\urldef\tempurl%
\url{https://doi.org/10.1145/3583660.3583664}
\showDOI{\tempurl}


\bibitem[\protect\citeauthoryear{Rehof and F\"{a}hndrich}{Rehof and
  F\"{a}hndrich}{2001}]%
        {Rehof01}
\bibfield{author}{\bibinfo{person}{Jakob Rehof} {and} \bibinfo{person}{Manuel
  F\"{a}hndrich}.} \bibinfo{year}{2001}\natexlab{}.
\newblock \showarticletitle{Type-base Flow Analysis: From Polymorphic Subtyping
  to CFL-reachability}. In \bibinfo{booktitle}{\emph{Proceedings of the 28th
  ACM SIGPLAN-SIGACT Symposium on Principles of Programming Languages}}
  \emph{(\bibinfo{series}{POPL})}. \bibinfo{pages}{54--66}.
\newblock


\bibitem[\protect\citeauthoryear{Reps}{Reps}{1995a}]%
        {Reps1995b}
\bibfield{author}{\bibinfo{person}{Thomas Reps}.}
  \bibinfo{year}{1995}\natexlab{a}.
\newblock \showarticletitle{Shape Analysis As a Generalized Path Problem}. In
  \bibinfo{booktitle}{\emph{Proceedings of the 1995 ACM SIGPLAN Symposium on
  Partial Evaluation and Semantics-based Program Manipulation}}
  \emph{(\bibinfo{series}{PEPM '95})}. \bibinfo{publisher}{ACM},
  \bibinfo{pages}{1--11}.
\newblock


\bibitem[\protect\citeauthoryear{Reps}{Reps}{1997}]%
        {Reps97}
\bibfield{author}{\bibinfo{person}{Thomas Reps}.}
  \bibinfo{year}{1997}\natexlab{}.
\newblock \showarticletitle{Program Analysis via Graph Reachability}. In
  \bibinfo{booktitle}{\emph{Proceedings of the 1997 International Symposium on
  Logic Programming}} \emph{(\bibinfo{series}{ILPS})}. \bibinfo{pages}{5--19}.
\newblock


\bibitem[\protect\citeauthoryear{Reps}{Reps}{2000}]%
        {Reps00}
\bibfield{author}{\bibinfo{person}{Thomas Reps}.}
  \bibinfo{year}{2000}\natexlab{}.
\newblock \showarticletitle{Undecidability of Context-sensitive Data-dependence
  Analysis}.
\newblock \bibinfo{journal}{\emph{ACM Trans. Program. Lang. Syst.}}
  \bibinfo{volume}{22}, \bibinfo{number}{1} (\bibinfo{year}{2000}),
  \bibinfo{pages}{162--186}.
\newblock


\bibitem[\protect\citeauthoryear{Reps, Horwitz, and Sagiv}{Reps
  et~al\mbox{.}}{1995}]%
        {Reps95}
\bibfield{author}{\bibinfo{person}{Thomas Reps}, \bibinfo{person}{Susan
  Horwitz}, {and} \bibinfo{person}{Mooly Sagiv}.}
  \bibinfo{year}{1995}\natexlab{}.
\newblock \showarticletitle{Precise Interprocedural Dataflow Analysis via Graph
  Reachability}. In \bibinfo{booktitle}{\emph{POPL}}. \bibinfo{publisher}{ACM},
  \bibinfo{address}{New York, NY, USA}.
\newblock


\bibitem[\protect\citeauthoryear{Reps, Horwitz, Sagiv, and Rosay}{Reps
  et~al\mbox{.}}{1994}]%
        {Reps94}
\bibfield{author}{\bibinfo{person}{Thomas Reps}, \bibinfo{person}{Susan
  Horwitz}, \bibinfo{person}{Mooly Sagiv}, {and} \bibinfo{person}{Genevieve
  Rosay}.} \bibinfo{year}{1994}\natexlab{}.
\newblock \showarticletitle{Speeding Up Slicing}.
\newblock \bibinfo{journal}{\emph{SIGSOFT Softw. Eng. Notes}}
  \bibinfo{volume}{19}, \bibinfo{number}{5} (\bibinfo{year}{1994}),
  \bibinfo{pages}{11--20}.
\newblock


\bibitem[\protect\citeauthoryear{Reps}{Reps}{1995b}]%
        {Reps1995}
\bibfield{author}{\bibinfo{person}{Thomas~W. Reps}.}
  \bibinfo{year}{1995}\natexlab{b}.
\newblock \bibinfo{booktitle}{\emph{Demand Interprocedural Program Analysis
  Using Logic Databases}}.
\newblock \bibinfo{publisher}{Springer US}, \bibinfo{address}{Boston, MA},
  \bibinfo{pages}{163--196}.
\newblock
\showISBNx{978-1-4615-2207-2}
\urldef\tempurl%
\url{https://doi.org/10.1007/978-1-4615-2207-2_8}
\showDOI{\tempurl}


\bibitem[\protect\citeauthoryear{Ryzhyk and Budiu}{Ryzhyk and Budiu}{2019}]%
        {DDLog}
\bibfield{author}{\bibinfo{person}{Leonid Ryzhyk} {and} \bibinfo{person}{Mihai
  Budiu}.} \bibinfo{year}{2019}\natexlab{}.
\newblock \showarticletitle{Differential Datalog}. In
  \bibinfo{booktitle}{\emph{Datalog 2.0 2019 - 3rd International Workshop on
  the Resurgence of Datalog in Academia and Industry}}
  \emph{(\bibinfo{series}{CEUR Workshop Proceedings},
  Vol.~\bibinfo{volume}{2368})}. \bibinfo{pages}{56--67}.
\newblock
\urldef\tempurl%
\url{http://ceur-ws.org/Vol-2368/paper6.pdf}
\showURL{%
\tempurl}


\bibitem[\protect\citeauthoryear{Shang, Xie, and Xue}{Shang
  et~al\mbox{.}}{2012}]%
        {Shang2012}
\bibfield{author}{\bibinfo{person}{Lei Shang}, \bibinfo{person}{Xinwei Xie},
  {and} \bibinfo{person}{Jingling Xue}.} \bibinfo{year}{2012}\natexlab{}.
\newblock \showarticletitle{On-demand Dynamic Summary-based Points-to
  Analysis}. In \bibinfo{booktitle}{\emph{Proceedings of the Tenth
  International Symposium on Code Generation and Optimization}}
  \emph{(\bibinfo{series}{CGO '12})}. \bibinfo{publisher}{ACM},
  \bibinfo{pages}{264--274}.
\newblock


\bibitem[\protect\citeauthoryear{Sp\"{a}th, Ali, and Bodden}{Sp\"{a}th
  et~al\mbox{.}}{2019}]%
        {Spath2019}
\bibfield{author}{\bibinfo{person}{Johannes Sp\"{a}th}, \bibinfo{person}{Karim
  Ali}, {and} \bibinfo{person}{Eric Bodden}.} \bibinfo{year}{2019}\natexlab{}.
\newblock \showarticletitle{Context-, Flow-, and Field-Sensitive Data-Flow
  Analysis Using Synchronized Pushdown Systems}.
\newblock \bibinfo{journal}{\emph{Proc. ACM Program. Lang.}}
  \bibinfo{volume}{3}, \bibinfo{number}{POPL}, Article \bibinfo{articleno}{48}
  (\bibinfo{date}{Jan.} \bibinfo{year}{2019}), \bibinfo{numpages}{29}~pages.
\newblock
\urldef\tempurl%
\url{https://doi.org/10.1145/3290361}
\showDOI{\tempurl}


\bibitem[\protect\citeauthoryear{Sridharan and Bod\'{\i}k}{Sridharan and
  Bod\'{\i}k}{2006}]%
        {Sridharan2006}
\bibfield{author}{\bibinfo{person}{Manu Sridharan} {and}
  \bibinfo{person}{Rastislav Bod\'{\i}k}.} \bibinfo{year}{2006}\natexlab{}.
\newblock \showarticletitle{Refinement-based Context-sensitive Points-to
  Analysis for Java}.
\newblock \bibinfo{journal}{\emph{SIGPLAN Not.}} \bibinfo{volume}{41},
  \bibinfo{number}{6} (\bibinfo{year}{2006}), \bibinfo{pages}{387--400}.
\newblock


\bibitem[\protect\citeauthoryear{Sridharan, Gopan, Shan, and
  Bod\'{\i}k}{Sridharan et~al\mbox{.}}{2005}]%
        {Sridharan2005}
\bibfield{author}{\bibinfo{person}{Manu Sridharan}, \bibinfo{person}{Denis
  Gopan}, \bibinfo{person}{Lexin Shan}, {and} \bibinfo{person}{Rastislav
  Bod\'{\i}k}.} \bibinfo{year}{2005}\natexlab{}.
\newblock \showarticletitle{Demand-driven Points-to Analysis for Java}. In
  \bibinfo{booktitle}{\emph{OOPSLA}}.
\newblock


\bibitem[\protect\citeauthoryear{Steensgaard}{Steensgaard}{1996}]%
        {Steensgaard96}
\bibfield{author}{\bibinfo{person}{Bjarne Steensgaard}.}
  \bibinfo{year}{1996}\natexlab{}.
\newblock \showarticletitle{Points-to Analysis in Almost Linear Time}. In
  \bibinfo{booktitle}{\emph{Proceedings of the 23rd ACM SIGPLAN-SIGACT
  Symposium on Principles of Programming Languages}} (St. Petersburg Beach,
  Florida, USA) \emph{(\bibinfo{series}{POPL ’96})}.
  \bibinfo{publisher}{Association for Computing Machinery},
  \bibinfo{address}{New York, NY, USA}, \bibinfo{pages}{32–41}.
\newblock
\showISBNx{0897917693}
\urldef\tempurl%
\url{https://doi.org/10.1145/237721.237727}
\showDOI{\tempurl}


\bibitem[\protect\citeauthoryear{Szab\'{o}, Erdweg, and Voelter}{Szab\'{o}
  et~al\mbox{.}}{2016}]%
        {Szabo2016}
\bibfield{author}{\bibinfo{person}{Tam\'{a}s Szab\'{o}},
  \bibinfo{person}{Sebastian Erdweg}, {and} \bibinfo{person}{Markus Voelter}.}
  \bibinfo{year}{2016}\natexlab{}.
\newblock \showarticletitle{IncA: A DSL for the Definition of Incremental
  Program Analyses}. In \bibinfo{booktitle}{\emph{Proceedings of the 31st
  IEEE/ACM International Conference on Automated Software Engineering}}
  (Singapore, Singapore) \emph{(\bibinfo{series}{ASE '16})}.
  \bibinfo{publisher}{Association for Computing Machinery},
  \bibinfo{address}{New York, NY, USA}, \bibinfo{pages}{320–331}.
\newblock
\showISBNx{9781450338455}
\urldef\tempurl%
\url{https://doi.org/10.1145/2970276.2970298}
\showDOI{\tempurl}


\bibitem[\protect\citeauthoryear{Tang, Wang, Xiong, Zhang, Wang, and
  Zhang}{Tang et~al\mbox{.}}{2017}]%
        {Tang2017}
\bibfield{author}{\bibinfo{person}{Hao Tang}, \bibinfo{person}{Di Wang},
  \bibinfo{person}{Yingfei Xiong}, \bibinfo{person}{Lingming Zhang},
  \bibinfo{person}{Xiaoyin Wang}, {and} \bibinfo{person}{Lu Zhang}.}
  \bibinfo{year}{2017}\natexlab{}.
\newblock \showarticletitle{Conditional Dyck-CFL Reachability Analysis for
  Complete and Efficient Library Summarization}. In
  \bibinfo{booktitle}{\emph{Programming Languages and Systems}},
  \bibfield{editor}{\bibinfo{person}{Hongseok Yang}} (Ed.).
  \bibinfo{publisher}{Springer Berlin Heidelberg}, \bibinfo{address}{Berlin,
  Heidelberg}, \bibinfo{pages}{880--908}.
\newblock
\showISBNx{978-3-662-54434-1}


\bibitem[\protect\citeauthoryear{Tang, Wang, Zhang, Xie, Zhang, and Mei}{Tang
  et~al\mbox{.}}{2015}]%
        {Tang15}
\bibfield{author}{\bibinfo{person}{Hao Tang}, \bibinfo{person}{Xiaoyin Wang},
  \bibinfo{person}{Lingming Zhang}, \bibinfo{person}{Bing Xie},
  \bibinfo{person}{Lu Zhang}, {and} \bibinfo{person}{Hong Mei}.}
  \bibinfo{year}{2015}\natexlab{}.
\newblock \showarticletitle{Summary-Based Context-Sensitive Data-Dependence
  Analysis in Presence of Callbacks}. In \bibinfo{booktitle}{\emph{Proceedings
  of the 42Nd Annual ACM SIGPLAN-SIGACT Symposium on Principles of Programming
  Languages}} \emph{(\bibinfo{series}{POPL})}. \bibinfo{pages}{83--95}.
\newblock
\showISBNx{978-1-4503-3300-9}


\bibitem[\protect\citeauthoryear{Tseng}{Tseng}{2020}]%
        {DynamicUndirectedConnectivity_Tomtseng2020}
\bibfield{author}{\bibinfo{person}{Tom Tseng}.}
  \bibinfo{year}{2020}\natexlab{}.
\newblock \bibinfo{title}{Dynamic connectivity data structure by Holm, de
  Lichtenberg, and Thorup}.
\newblock
  \bibinfo{howpublished}{\url{https://github.com/tomtseng/dynamic-connectivity-hdt}}.
\newblock


\bibitem[\protect\citeauthoryear{Vedurada and Nandivada}{Vedurada and
  Nandivada}{2019}]%
        {Vedurada2019}
\bibfield{author}{\bibinfo{person}{Jyothi Vedurada} {and}
  \bibinfo{person}{V.~Krishna Nandivada}.} \bibinfo{year}{2019}\natexlab{}.
\newblock \showarticletitle{Batch Alias Analysis}. In
  \bibinfo{booktitle}{\emph{Proceedings of the 34th IEEE/ACM International
  Conference on Automated Software Engineering}} (San Diego, California)
  \emph{(\bibinfo{series}{ASE '19})}. \bibinfo{publisher}{IEEE Press},
  \bibinfo{pages}{936–948}.
\newblock
\showISBNx{9781728125084}
\urldef\tempurl%
\url{https://doi.org/10.1109/ASE.2019.00091}
\showDOI{\tempurl}


\bibitem[\protect\citeauthoryear{Xu, Rountev, and Sridharan}{Xu
  et~al\mbox{.}}{2009}]%
        {Xu09}
\bibfield{author}{\bibinfo{person}{Guoqing Xu}, \bibinfo{person}{Atanas
  Rountev}, {and} \bibinfo{person}{Manu Sridharan}.}
  \bibinfo{year}{2009}\natexlab{}.
\newblock \showarticletitle{Scaling CFL-Reachability-Based Points-To Analysis
  Using Context-Sensitive Must-Not-Alias Analysis}. In
  \bibinfo{booktitle}{\emph{Proceedings of the 23rd European Conference on
  ECOOP 2009 --- Object-Oriented Programming}}
  \emph{(\bibinfo{series}{Genoa})}. \bibinfo{pages}{98--122}.
\newblock


\bibitem[\protect\citeauthoryear{Yan, Xu, and Rountev}{Yan
  et~al\mbox{.}}{2011}]%
        {Yan11}
\bibfield{author}{\bibinfo{person}{Dacong Yan}, \bibinfo{person}{Guoqing Xu},
  {and} \bibinfo{person}{Atanas Rountev}.} \bibinfo{year}{2011}\natexlab{}.
\newblock \showarticletitle{Demand-driven Context-sensitive Alias Analysis for
  Java}. In \bibinfo{booktitle}{\emph{Proceedings of the 2011 International
  Symposium on Software Testing and Analysis}}
  \emph{(\bibinfo{series}{ISSTA})}. \bibinfo{pages}{155--165}.
\newblock


\bibitem[\protect\citeauthoryear{Yannakakis}{Yannakakis}{1990}]%
        {Yannakakis90}
\bibfield{author}{\bibinfo{person}{Mihalis Yannakakis}.}
  \bibinfo{year}{1990}\natexlab{}.
\newblock \showarticletitle{Graph-theoretic Methods in Database Theory}. In
  \bibinfo{booktitle}{\emph{Proceedings of the Ninth ACM SIGACT-SIGMOD-SIGART
  Symposium on Principles of Database Systems}}
  \emph{(\bibinfo{series}{PODS})}. \bibinfo{pages}{230--242}.
\newblock


\bibitem[\protect\citeauthoryear{Yuan and Eugster}{Yuan and Eugster}{2009}]%
        {Yuan09}
\bibfield{author}{\bibinfo{person}{Hao Yuan} {and} \bibinfo{person}{Patrick
  Eugster}.} \bibinfo{year}{2009}\natexlab{}.
\newblock \showarticletitle{An Efficient Algorithm for Solving the Dyck-CFL
  Reachability Problem on Trees}. In \bibinfo{booktitle}{\emph{Proceedings of
  the 18th European Symposium on Programming Languages and Systems: Held As
  Part of the Joint European Conferences on Theory and Practice of Software,
  ETAPS 2009}} \emph{(\bibinfo{series}{ESOP})}. \bibinfo{pages}{175--189}.
\newblock


\bibitem[\protect\citeauthoryear{Zadeck}{Zadeck}{1984}]%
        {Zadeck1984}
\bibfield{author}{\bibinfo{person}{Frank~Kenneth Zadeck}.}
  \bibinfo{year}{1984}\natexlab{}.
\newblock \showarticletitle{Incremental Data Flow Analysis in a Structured
  Program Editor}. In \bibinfo{booktitle}{\emph{Proceedings of the 1984 SIGPLAN
  Symposium on Compiler Construction}} (Montreal, Canada)
  \emph{(\bibinfo{series}{SIGPLAN '84})}. \bibinfo{publisher}{Association for
  Computing Machinery}, \bibinfo{address}{New York, NY, USA},
  \bibinfo{pages}{132–143}.
\newblock
\showISBNx{0897911393}
\urldef\tempurl%
\url{https://doi.org/10.1145/502874.502888}
\showDOI{\tempurl}


\bibitem[\protect\citeauthoryear{Zhang, Lyu, Yuan, and Su}{Zhang
  et~al\mbox{.}}{2013}]%
        {Zhang13}
\bibfield{author}{\bibinfo{person}{Qirun Zhang}, \bibinfo{person}{Michael~R.
  Lyu}, \bibinfo{person}{Hao Yuan}, {and} \bibinfo{person}{Zhendong Su}.}
  \bibinfo{year}{2013}\natexlab{}.
\newblock \showarticletitle{Fast Algorithms for Dyck-CFL-reachability with
  Applications to Alias Analysis} \emph{(\bibinfo{series}{PLDI})}.
  \bibinfo{publisher}{ACM}.
\newblock


\bibitem[\protect\citeauthoryear{Zhang and Su}{Zhang and Su}{2017}]%
        {Zhang2017}
\bibfield{author}{\bibinfo{person}{Qirun Zhang} {and} \bibinfo{person}{Zhendong
  Su}.} \bibinfo{year}{2017}\natexlab{}.
\newblock \showarticletitle{Context-Sensitive Data-Dependence Analysis via
  Linear Conjunctive Language Reachability}.
\newblock \bibinfo{journal}{\emph{SIGPLAN Not.}} \bibinfo{volume}{52},
  \bibinfo{number}{1} (\bibinfo{date}{Jan.} \bibinfo{year}{2017}),
  \bibinfo{pages}{344–358}.
\newblock
\showISSN{0362-1340}
\urldef\tempurl%
\url{https://doi.org/10.1145/3093333.3009848}
\showDOI{\tempurl}


\bibitem[\protect\citeauthoryear{Zheng and Rugina}{Zheng and Rugina}{2008}]%
        {Zheng2008}
\bibfield{author}{\bibinfo{person}{Xin Zheng} {and} \bibinfo{person}{Radu
  Rugina}.} \bibinfo{year}{2008}\natexlab{}.
\newblock \showarticletitle{Demand-driven Alias Analysis for C}. In
  \bibinfo{booktitle}{\emph{Proceedings of the 35th Annual ACM SIGPLAN-SIGACT
  Symposium on Principles of Programming Languages}}
  \emph{(\bibinfo{series}{POPL '08})}. \bibinfo{publisher}{ACM},
  \bibinfo{pages}{197--208}.
\newblock


\end{thebibliography}

\clearpage
\appendix
\section{Proofs}\label{sec:app_proofs}

In this section we provide proofs for the correctness invariants of \cref{subsec:analysis}.

\lemcorrectnesspdsccs*
\begin{proof}
We argue that the statement holds after each $\InsertEdge(u,v,\CloseParenthesis)$ and $\DeleteEdge(u,v, \CloseParenthesis)$ operations.
First, observe that at all times, we have an edge $(x,y)$ in the $\PrimCompDS$ data structure iff there is a node $z$ and a label $\CloseParenthesisBeta$ such that $x$ is a neighbor of $y$ in the linked list $\OutEdges[z][\CloseParenthesisBeta]$.
As we have stated above, $\OutEdges[u][\CloseParenthesisBeta]$ is, at all times, a linked list representation of the edge set $u\DTo{\CloseParenthesis}\cdot$.

\SubParagraph{Edge insertions.}
Consider the processing of $\InsertEdge(u,v,\CloseParenthesis)$.
This operation can only merge $\PDSCC(v)$ with $\PDSCC(y)$, where $y$ is another node for which there already exists an edge $u\DTo{\CloseParenthesis}y$.
Hence $|\OutEdges[u][\CloseParenthesis]|\geq 1$, and the if-else-block in \cref{line:algo_insert_pdscc}-\cref{line:algo_insert_insert_primal_edge} of \cref{algo:insert_edge} inserts an edge $(v,y)$ in $\PrimCompDS$ (the else part). 
Once the insertion of $(v,y)$ has been completed, the components in $\PrimCompDS$ are again precisely the PDSCCs of $G$.

\SubParagraph{Edge deletions.}
Consider the processing of an operation $\DeleteEdge(u,v,\CloseParenthesis)$.
\cref{line:algo_delete_delete_primal_edge} to \cref{line:algo_delete_remove_from_outedges} of \cref{algo:delete_edge} deletes from $\PrimCompDS$ the edges $(x,v)$ and $(y,v)$, where $x$ and $y$ are the neighbors of $v$ in the linked list $\OutEdges[u][\CloseParenthesis]$.
Thus, once the deletion of $(x,v)$ and $(y,v)$ has been completed, and $x,y$ are connected back (\cref{line:algo_delete_reinsert_pdscc}),
the components in $\PrimCompDS$ are again precisely the PDSCCs of $G$.
\end{proof}

\lemsoundnessinedges*
\begin{proof}
We argue that the statement holds after each $\InsertEdge(u,v,\CloseParenthesis)$ and $\DeleteEdge(u,v, \CloseParenthesis)$ operation.

\SubParagraph{Edge insertions.}
Consider the processing of an operation $\InsertEdge(u,v,\CloseParenthesis)$.
\cref{line:algo_insert_outedges} of \cref{algo:insert_edge} inserts $u$ in $\InEdges[v][\CloseParenthesis]$ iff $|\OutEdges[u][\CloseParenthesis]|=0$, i.e., $v$ is the only (first and last) node $y$ in $\OutEdges[u][\CloseParenthesis]$ such that  $u\DTo{\CloseParenthesis}y$.
Hence, the statement holds after each edge insertion.

\SubParagraph{Edge deletions.}
Consider the processing of an operation $\DeleteEdge(u,v,\CloseParenthesis)$.
\cref{line:algo_delete_inedges} of \cref{algo:delete_edge} inserts $u$ in $\InEdges[w][\CloseParenthesis]$ iff $w$ is the penultimate node in $\OutEdges[u][\CloseParenthesis]$ and $v$ is the last node in $\OutEdges[u][\CloseParenthesis]$.
Since $w$ is in $\OutEdges[u][\CloseParenthesis]$, we have an edge $u\DTo{\CloseParenthesis}w$, while, after removing $v$ from $\OutEdges[u][\CloseParenthesis]$ (\cref{line:algo_delete_remove_from_outedges}), $w$ is the last node in $\OutEdges[u][\CloseParenthesis]$.
\end{proof}

\lemcompletenessinedges*
\begin{proof}
We argue that the statement holds after each $\InsertEdge(u,v,\CloseParenthesisBeta)$ and $\DeleteEdge(u,v, \CloseParenthesisBeta)$ operation.
In particular, $z$ will always be the last node in $\OutEdges[u][\CloseParenthesis]$.

\SubParagraph{Edge insertions.}
Consider the processing of an operation $\InsertEdge(u,v,\CloseParenthesisBeta)$. It suffices to prove the statement for $x=u$ and $y=v$.
If there exists another node $w\neq v$ and an edge $u\DTo{\CloseParenthesisBeta}w$, observe that $\PDSCC(w)=\PDSCC(v)$, as the primal graph now has an edge $(v,w)$.
By the induction hypothesis, there exists a node $z\in\PDSCC(w)$ such that $u\in\InEdges[z][\CloseParenthesisBeta]$, and we are done.
Otherwise $v$ is the first $\CloseParenthesisBeta$-neighbor of $u$, and \cref{line:algo_insert_inedges} of \cref{algo:insert_edge} will set $u\in\InEdges[v][\CloseParenthesisBeta]$, as desired.

\SubParagraph{Edge deletions.}
Consider the processing of an operation $\DeleteEdge(u,v,\CloseParenthesisBeta)$, and it suffices to prove the statement for $x=u$ and any node $y$ such that $u\DTo{\CloseParenthesisBeta}y$.
If $v$ is not the last node in $\OutEdges[u][\CloseParenthesisBeta]$, the statement holds by the induction hypothesis.
Otherwise $v$ is the last node in $\OutEdges[u][\CloseParenthesisBeta]$, and there exists a penultimate 
node $w$ in $\OutEdges[u][\CloseParenthesisBeta]$. 
Then \cref{line:algo_delete_inedges} of \cref{algo:delete_edge} will set $u\in\InEdges[w][\CloseParenthesisBeta]$, thereby restoring the invariant.
\end{proof}

\lemmakeprimarycorrectness*
\begin{proof}
First, observe that the first step of $\MakePrimary()$ (\cref{algo:make_primary}) correctly stores in $\AffectedDSCCs$ a sound overapproximation of the DSCCs that have to be split after removing the edge $u\DTo{\CloseParenthesisBeta}v$ (in particular, $\AffectedDSCCs$ stores the representative nodes of the corresponding $\DSCC$s).
Indeed, the algorithm first marks $\DSCC(v)$ as affected (\cref{line:algo_makeprimary_mark_dscc_v_affected}).
In turn, any other DSCC $S'$ that has to be split must contain two nodes $s \neq t$ with incoming edges $\DSCC(S)\DTo{\CloseParenthesisBeta}s$ and $\DSCC(S)\DTo{\CloseParenthesisBeta}t$, where $S$ is a DSCC that has to be split.
Thus, after the end of the first step (\cref{line:algo_makeprimary_emptyqueue} to \cref{line:algo_makeprimary_firststepend}), $\AffectedDSCCs$ soundly overapproximates the DSCCs that have to be split.
After executing step~2 (\cref{line:algo_makeprimary_newroots} to \cref{line:algo_makeprimary_secondstepend}), every component in $\DisjointSets$ is either a previously computed DSCC, or a PDSCC (computed from $\PrimCompDS$) of a previously computed DSCC. Since PDSCCs are also DSCCs, \cref{item:make_primary_lemma_dsccs} follows (\cref{lem:correctness_pdsccs} guarantees that the components of $\PrimCompDS$ are the PDSCCs). 

We now turn our attention to \cref{item:make_primary_lemma_edges}.
Note that the algorithm modifies the $\Edges$ lists of nodes of two types:~
(i)~roots of the potentially affected DSCCs (those stored in $\AffectedDSCCs$), in \cref{line:algo_makeprimary_insert_y_in_edges_r} and
(ii)~  roots of the unaffected DSCCs that have $\CloseParenthesisBeta$-edges to nodes of affected DSCCs, in \cref{line:algo_makeprimary_insert_t_in_edges_x,line:algo_makeprimary_insert_t_in_edges_x}.

First, consider \cref{item:make_primary_lemma_edges_sound}.
\begin{enumerate}

\item If $x$ is the root node of a potentially affected DSCC (case (i)), $y$ must have been added to $\Edges[x][\CloseParenthesisBeta]$ in \cref{line:algo_makeprimary_insert_y_in_edges_r} (where $x=r$ in the pseudocode), as $\Edges[x][\CloseParenthesisBeta]$ was reinitialized to an empty list earlier in \cref{line:makeprimary_reinitialize_edges_of_affected_nodes}.
But then $y$ appears in $\OutEdges[z][\CloseParenthesisBeta]$ for some node $z \in \DSCC(x)$ ($z=t$ in the pseudocode, \cref{line:first}) 
and thus we have an edge $z\DTo{\CloseParenthesisBeta}y$ as desired (here we obtain $y=w$). 

\item On the other hand, if $x$ is the root of an unaffected DSCC that has $\CloseParenthesisBeta$-edges to nodes of affected DSCCs (case (ii)), $y$ must have been added to $\Edges[x][\CloseParenthesisBeta]$ in \cref{line:algo_makeprimary_insert_t_in_edges_x} (where $y=t$ in the pseudocode), as $\Edges[x][\CloseParenthesisBeta]$ was reinitialized to an empty list earlier in \cref{line:algo_makeprimary_unaffected}.
But then there exists a node $z \in \DSCC(x)$ (\cref{line:algo_makeprimary_rootoft}, $z=s$ in the pseudocode) such that $z \in \InEdges[y][\CloseParenthesisBeta]$ (\cref{line:algo_makeprimary_iterateoverpdsccinedges}).  
Hence, by \cref{lem:soundness_inedges}, we have an edge $z\DTo{\CloseParenthesisBeta}y$ as desired (here we obtain $y=w$).

 \item Finally, if $x$ is neither in one of the two types (i) or (ii) above, then it is part of an unaffected $\DSCC$ which  does not have an edge 
entering an affected $\DSCC$. Since these $\DSCC$s are untouched by the algorithm,  
the statement holds by the induction hypothesis. In this case, we have the guarantees from the $\Fixpoint$ computation. 
In this case, we can have $y \neq w$, which comes from \cref{line:append2} in $\Fixpoint$.

\end{enumerate}
Second, consider \cref{item:make_primary_lemma_edges_complete}.
\begin{enumerate}
\item If $x$ is a node of a potentially affected DSCC (case (i)), when \cref{line:algo_makeprimary_edgelist_of_pdsccs} is executed for some node $z$ in the component of $x$ ($z=t,x=r$ in the pseudocode), since we have an edge $z\DTo{\CloseParenthesisBeta}w$, the if condition will succeed at least once, inserting a suitable node $y$ in $\Edges[x][\CloseParenthesisBeta]$.
\item On the other hand, if $x$ is the root of an unaffected DSCC that has $\CloseParenthesisBeta$-edges to nodes of affected DSCCs (case (ii)). 
By definition, all these nodes having the incoming 
$\CloseParenthesisBeta$ edge  will belong to the same PDSCC. 
By \cref{lem:completeness_inedges}, one such node $t$ from the PDSCC will have $s \in\InEdges[t][\CloseParenthesisBeta]$ where 
$s \in \DSCC(x)$. Then \cref{line:algo_makeprimary_insert_t_in_edges_x} will insert $y$ to $\Edges[x][\CloseParenthesisBeta$] for $y$ being one of these nodes $t$.
\item If $x$ is neither in one of the two types (i) or (ii) above, then as seen above, 
the statement holds by the induction hypothesis.
\end{enumerate}

\end{proof}

\lemcorrectness*
\begin{proof}
We argue that the statement holds after each $\InsertEdge(u,v,\CloseParenthesis)$ and $\DeleteEdge(u,v, \CloseParenthesis)$ operation.

\SubParagraph{Edge insertions.}
Consider the processing of an operation $\InsertEdge(u,v,\CloseParenthesis)$.
The $\Fixpoint()$ function (\cref{algo:offlinealgo}) guarantees that, upon termination, $\DisjointSets$ represents the DSCCs of the input graph,
while for every node $x$ that is the representative node of some set in $\DisjointSets$,
$\Edges[x][\CloseParenthesisBeta]$ either is empty, if there are no edges $\DSCC(x)\DTo{\CloseParenthesisBeta}\cdot$, or has a exactly one node $y$ such that $\DSCC(x)\DTo{\CloseParenthesisBeta}y$
(note that all other nodes $z$ for which we also have $\DSCC(x)\DTo{\CloseParenthesisBeta}z$ belong to $\DSCC(z)$, and thus in the same set of $\DisjointSets$).

Now, when processing $\InsertEdge(u,v,\CloseParenthesis)$, \cref{algo:insert_edge} inserts $v$ in $\Edges[x][\CloseParenthesis]$,
where $x=\DisjointSets.\Find(u)$ is the representative of $\DSCC(u)$ \cref{line:algo_insert_insert_edges}.
The new edge $u\DTo{\CloseParenthesis}v$ leads to the (potential) merge of $\DSCC(v)$ and $\DSCC(z)$, where $z$ is another node with $\DSCC(z)\neq \DSCC(v)$ and $\DSCC(u)\DTo{\CloseParenthesis}z$.
The condition $|\Edges[x][\CloseParenthesis]|\geq 2$ will further insert $(x,\CloseParenthesis)$ in $\Queue$ and trigger a new fixpoint computation,
and the lemma then follows from the correctness of $\Fixpoint()$ (as established in~\cite{Chatterjee18}).

\SubParagraph{Edge deletions.}
Consider the processing of an operation $\DeleteEdge(u,v,\CloseParenthesis)$.
\cref{lem:make_primary_correctness} guarantees at that at the end of $\MakePrimary()$, every component in $\DisjointSets$ is a DSCC, while edges between components represented in the $\Edges$ data structure capture in a sound and complete way edges between nodes in the corresponding components.
Finally, \cref{line:algo_makeprimary_inserttoqueue} of $\MakePrimary()$ inserts in $\Queue$  all the pairs $(x,\CloseParenthesisBeta)$ that can trigger new fixpoint steps, hence after the call to $\Fixpoint()$ in \cref{line:algo_delete_fixpoint} of \cref{algo:delete_edge} has been completed, $\DisjointSets$ correctly represents the DSCCs of the updated graph.
\end{proof}

\section{A Note on Earlier Approaches}\label{sec:mistakes}

The setting of dynamic bidirected Dyck reachability was studied recently in~\cite{Li2022}.
Unfortunately, that approach suffers correctness and complexity issues, which we illustrate here.

\Paragraph{Complexity counterexamples.}
The approach developed in~\cite{Li2022} claims a running time of $O(n\cdot \alpha(n))$ for a graph of $n$ nodes.
As we show here, that statement is wrong:~the proposed algorithm can take $\Omega(n^2)$ time for a single edge deletion.
At close inspection, there are two independent parts of the deletion algorithm that can exhibit this quadratic bound, and the complexity analysis fails to account for both.
Here we illustrate these counterexamples (\cref{fig:mistake_complexity_graphs}) and the runtime behavior that the tool accompanying~\cite{Li2022} has on them (\cref{fig:mistake_complexity_plots}).

\begin{figure}
\scalebox{0.9}{%
\begin{subfigure}[b]{0.35\textwidth}
\centering
\begin{tikzpicture}[thick, >=latex, node distance=0.3cm and 1cm,
pre/.style={<-,shorten >= 1pt, shorten <=1pt,},
post/.style={->,shorten >= 1pt, shorten <=1pt,},
und/.style={very thick, draw=gray},
node/.style={circle, minimum size=4.5mm, draw=black!100, fill=white!100, thick, inner sep=0},
virt/.style={circle,draw=black!50,fill=black!20, opacity=0}]

\newcommand{\xdisposition}{0}
\newcommand{\ydisposition}{0}
\newcommand{\xstep}{1.4}
\newcommand{\ystep}{1.2}
\def\bend{20}

\node	[node]		(u)	at(-1*\xstep,-1.5*\ystep) {$u$};

\node	[node]		(a1)	at(0*\xstep,0*\ystep) {$a_{1}$};
\node	[]		(adots)	at(1*\xstep,0*\ystep) {$\dots$};
\node	[node]		(an)	at(2*\xstep,0*\ystep) {$a_{n}$};

\node	[node]		(b1)	at(0*\xstep,-1*\ystep) {$b_{1}$};
\node	[]		(bdots)	at(1*\xstep,-1*\ystep) {$\dots$};
\node	[node]		(bn)	at(2*\xstep,-1*\ystep) {$b_{n}$};

\node	[node]		(c1)	at(0*\xstep,-2*\ystep) {$c_{1}$};
\node	[]		(cdots)	at(1*\xstep,-2*\ystep) {$\dots$};
\node	[node]		(cn)	at(2*\xstep,-2*\ystep) {$c_{n}$};

\node	[node]		(d1)	at(0*\xstep,-3*\ystep) {$d_{1}$};
\node	[]		(ddots)	at(1*\xstep,-3*\ystep) {$\dots$};
\node	[node]		(dn)	at(2*\xstep,-3*\ystep) {$d_{n}$};

\draw[post] (a1) to node[left] {$\CloseParenthesis$} (b1);
\draw[post] (a1) to node[left] {$\CloseParenthesis$} (bdots);
\draw[post] (a1) to node[left] {$\CloseParenthesis$} (bn);
\draw[post] (an) to node[left] {$\CloseParenthesis$} (b1);
\draw[post] (an) to node[left] {$\CloseParenthesis$} (bdots);
\draw[post] (an) to node[left] {$\CloseParenthesis$} (bn);

\draw[post] (d1) to node[left] {$\CloseParenthesis$} (c1);
\draw[post] (d1) to node[left] {$\CloseParenthesis$} (cdots);
\draw[post] (d1) to node[left] {$\CloseParenthesis$} (cn);
\draw[post] (dn) to node[left] {$\CloseParenthesis$} (c1);
\draw[post] (dn) to node[left] {$\CloseParenthesis$} (cdots);
\draw[post] (dn) to node[left] {$\CloseParenthesis$} (cn);

%

\draw[post] (u) to node[above] {$\CloseParenthesis$} (b1);
\draw[post, dashed] (u) to node[below] {$\CloseParenthesis$} (c1);

\begin{pgfonlayer}{bg}
\node[box, very thick, rounded corners, draw=mybluecolor, dotted, fill=mybluecolor!10, fit=(b1)(bdots)(bn)(c1)(cdots)(cn)] (C1) {};
\end{pgfonlayer}

\end{tikzpicture}
\caption{
$\DeleteEdge(u,c_1,\CloseParenthesis)$
\label{subfig:mistake_complexity_graphs_dense}
}
\end{subfigure}
}
\hfill
\scalebox{0.9}{%
\begin{subfigure}[b]{0.6\textwidth}
\centering
\begin{tikzpicture}[thick, >=latex, node distance=0.4cm and 2cm,tight background,
pre/.style={<-,shorten >= 1pt, shorten <=1pt, thick},
post/.style={->,shorten >= 1pt, shorten <=1pt,  thick},
und/.style={very thick, draw=gray},
node/.style={circle, minimum size=4.5mm, draw=black!100, fill=white!100, thick, inner sep=0},
virt/.style={circle,draw=black!50,fill=black!20, opacity=0}]

\newcommand{\xdisposition}{0}
\newcommand{\ydisposition}{0}
\newcommand{\xstep}{1.5}
\newcommand{\ystep}{1}
\def\bend{20}

\node[] at (-0.5*\xstep, 3.7*\ystep) {};
\node[] at (5.3*\xstep, -1.2*\ystep) {};

\begin{scope}[shift={(0*\xstep,0*\ystep])}]

\node	[node]		(u)	at	(0*\xstep,0*\ystep) {$u$};
\node	[node]		(a1) at	(1*\xstep,0.5*\ystep)	 {$a_1$};
\node	[node]		(b1) at	(1*\xstep,-0.5*\ystep)	 {$b_1$};
\node	[node]		(a2) at	(2*\xstep,0.5*\ystep)	 {$a_2$};
\node	[node]		(b2) at	(2*\xstep,-0.5*\ystep)	 {$b_2$};
\node	[]		(adots) at	(3*\xstep,0.5*\ystep)	 {$\dots$};
\node	[]		(bdots) at	(3*\xstep,-0.5*\ystep)	 {$\dots$};
\node	[node]		(an) at	(4*\xstep,0.5*\ystep)	 {$a_n$};
\node	[node]		(bn) at	(4*\xstep,-0.5*\ystep)	 {$b_n$};

\node[node] (v) at (2.5*\xstep, -3*\ystep) {$v$};
\node[node] (c1) at (1*\xstep, -2*\ystep) {$c_1$};
\node[node] (c2) at (2*\xstep, -2*\ystep) {$c_2$};
\node[] (cdots) at (3*\xstep, -2*\ystep) {$\dots$};
\node[node] (cn) at (4*\xstep, -2*\ystep) {$c_n$};

\draw[post, \ColorAlpha] (v) to node[left] {$\CloseParenthesis$} (c1);
\draw[post, \ColorAlpha] (v) to node[left] {$\CloseParenthesis$} (c2);
\draw[post, \ColorAlpha] (v) to node[left] {$\CloseParenthesis$} (cdots);
\draw[post, \ColorAlpha] (v) to node[left] {$\CloseParenthesis$} (cn);

\draw[post, \ColorAlpha] (u) to node[above] {$\CloseParenthesis$} (a1);
\draw[post, \ColorAlpha, dashed] (u) to node[below] {$\CloseParenthesis$} (b1);

\draw[post, \ColorAlpha] (a1) to node[above] {$\CloseParenthesis$} (a2);
\draw[post, \ColorAlpha] (b1) to node[above] {$\CloseParenthesis$} (b2);
\draw[post, \ColorAlpha] (a2) to node[above] {$\CloseParenthesis$} (adots);
\draw[post, \ColorAlpha] (b2) to node[above] {$\CloseParenthesis$} (bdots);
\draw[post, \ColorAlpha] (adots) to node[above] {$\CloseParenthesis$} (an);
\draw[post, \ColorAlpha] (bdots) to node[above] {$\CloseParenthesis$} (bn);

\draw[post, \ColorBeta, bend left=\bend] (a1) to node[right] {$\CloseParenthesisBeta$} (c1);
\draw[post, \ColorBeta, bend right=0] (b1) to node[left] {$\CloseParenthesisBeta$} (c1);

\draw[post, \ColorBeta, bend left=\bend] (a2) to node[right] {$\CloseParenthesisBeta$} (c2);
\draw[post, \ColorBeta, bend right=0] (b2) to node[left] {$\CloseParenthesisBeta$} (c2);

\draw[post, \ColorBeta, bend left=\bend] (an) to node[right] {$\CloseParenthesisBeta$} (cn);
\draw[post, \ColorBeta, bend right=0] (bn) to node[left] {$\CloseParenthesisBeta$} (cn);

\end{scope}

\begin{pgfonlayer}{bg}
\node[box, very thick, rounded corners, draw=mybluecolor, dotted, fill=mybluecolor!10, fit=(a1)(b1)] (S1) {};
\node[box, very thick, rounded corners, draw=mybluecolor, dotted, fill=mybluecolor!10, fit=(a2)(b2)] (S2) {};
\node[box, very thick, rounded corners, draw=mybluecolor, dotted, fill=mybluecolor!10, fit=(an)(bn)] (Sn) {};
\node[box, very thick, rounded corners, draw=mybluecolor, dotted, fill=mybluecolor!10, fit=(c1)(c2)(cdots)(cn)] (SC) {};
\end{pgfonlayer}

\end{tikzpicture}
\caption{
$\DeleteEdge(u,b_1,\CloseParenthesis)$
\label{subfig:mistake_complexity_graphs_sparse}
}
\end{subfigure}
}
\caption{
A family of dense graphs (\protect\subref{subfig:mistake_complexity_graphs_dense}) and a family of sparse graphs (\protect\subref{subfig:mistake_complexity_graphs_sparse}) on which the edge deletion method of~\cite{Li2022} takes quadratic time.
\label{fig:mistake_complexity_graphs}
}
\end{figure}
\begin{figure}
\pgfplotsset{every tick label/.append style={font=\small}}
\def\plotwidth{7cm}
\def\offlinecolor{black!20}
\def\dynamiccolor{black!60}
\scalebox{0.9}{%
\begin{subfigure}[b]{0.475\textwidth}
\begin{tikzpicture}
\begin{axis}[
    xlabel={$n$},
    ylabel={seconds},
    xmin=0, xmax=1700,
    ymin=0, ymax=01,
    ytick={0,0.2,0.4,0.6,0.8,1},
    legend pos=north west,
    ymajorgrids=true,
    grid style=dashed,
    width=\plotwidth,
    legend style={font=\small},
    ylabel near ticks,
    xlabel near ticks,
    legend cell align={left},
]

\addplot [color=blue,mark=*] file {figures/dense_popl22Offline.dat};
\addplot[color=red,mark=square*] file {figures/dense_popl22Dyn.dat};
\legend{\cite[Offline]{Li2022},\cite[Dynamic]{Li2022}}

\end{axis}
\end{tikzpicture}
\caption{
Dense inputs.
\label{subfig:mistake_complexity_plots_dense}
}
\end{subfigure}
}
\hfill
\scalebox{0.9}{%
\begin{subfigure}[b]{0.475\textwidth}
\begin{tikzpicture}
\begin{axis}[
    xlabel={$n$},
    ylabel={seconds},
    xmin=0, xmax=1300,
    ymin=0, ymax=0.09,
    ytick={0,0.02,0.04,0.06,0.08},
    yticklabel style={/pgf/number format/fixed},
    legend pos=north west,
    ymajorgrids=true,
    grid style=dashed,
    width=\plotwidth,
    legend style={font=\small},
    ylabel near ticks,
    xlabel near ticks,    
    legend cell align={left},
]

\addplot[color=blue,mark=*] file {figures/sparse_popl22Offline.dat};
\addplot[color=red,mark=square*] file {figures/sparse_popl22Dyn.dat};

\legend{\cite[Offline]{Li2022},\cite[Dynamic]{Li2022}}
\end{axis}
\end{tikzpicture}
\caption{
Sparse inputs.
\label{subfig:mistake_complexity_plots_sparse}
}
\end{subfigure}
}
\caption{
Running time of the offline and dynamic algorithms of~\cite{Li2022} on the dense graphs of \cref{subfig:mistake_complexity_graphs_dense} (\protect\subref{subfig:mistake_complexity_plots_dense}) and sparse graphs of \cref{subfig:mistake_complexity_graphs_sparse} (\protect\subref{subfig:mistake_complexity_plots_sparse}).
\label{fig:mistake_complexity_plots}
}
\end{figure}

\SubParagraph{Dense inputs.}
\cref{subfig:mistake_complexity_graphs_dense} showcases a family of dense graphs parameterized by $n$ and having $O(n)$ nodes and $\Theta(n^2)$ edges.
Before deleting the edge $u\DTo{\ov{\alpha}}c_1$, $G$ has a DSCC $S=\{b_1,\dots, b_n, c_1,\dots c_n\}$.
After deleting the edge $u\DTo{\ov{\alpha}}c_1$, we obtain two disjoint DSCCs $S_b=\{b_1,\dots, b_n\}$ and $S_c=\{c_1,\dots, c_n\}$.
Running the edge deletion procedure of~\cite{Li2022} takes $\Omega(n^2)$ time.
Intuitively, this quadratic behavior stems from the fact that the algorithm iterates over one of the two sets of edges $\{a_i\DTo{\ov{\alpha}}b_j\}_{i,j\in [n]}$ and $\{d_i\DTo{\ov{\alpha}}c_j\}_{i,j\in [n]}$.
As these are quadratically many, the edge-deletion procedure takes quadratic time, as shown in the plot of \cref{subfig:mistake_complexity_plots_dense}.

This behavior occurs in \cite[Procedure 4, Lines~12-18]{Li2022}, while~\cite[Lemma~4.16]{Li2022} does not make a thorough complexity analysis of these lines.

\SubParagraph{Sparse inputs.}
\cref{subfig:mistake_complexity_graphs_sparse} showcases a family of sparse graphs parameterized by $n$ and having $O(n)$ nodes and $\Theta(n)$ edges.
Before deleting the edge $u\DTo{\CloseParenthesis}b_1$, we have a sequence of DSCCs $S_i=\{a_i, b_i\}_{i\in[n]}$.
Intuitively, we form  $S_1$ due to the edges $u\DTo{\CloseParenthesis}a_1$ and $u\DTo{\CloseParenthesis}b_1$, and each $S_i$, for $i>1$, connects $a_i$ and $b_i$ via a path that goes through $S_{i-1}$.
Moreover, we have an independent DSCC $S=\{c_1,\dots, c_n\}$, established via multiple two-edged paths going through $v$.
Upon deleting $u\DTo{\CloseParenthesis}b_1$, the algorithm of~\cite{Li2022} splits $S_1$, which in turn splits $S_2$ and so on, up to $S_n$.
Due to the edges $a_i\DTo{\ov{\beta}}c_i$ and $b_i\DTo{\ov{\beta}}c_i$, after splitting each $S_i$, the algorithm also processes the $\DSCC$ $S=\{c_1,\dots, c_n\}$, and checks whether it has to be split.
Naturally, it discovers that $S$ should not be split, due to the presence of the edges $v\DTo{\CloseParenthesis}c_i$.
However, each attempt to split $S$ requires time that is proportional to its size, i.e., $\Omega(n)$ time.
Since this process repeats after splitting each of the $n$ DSCCS $S_i$, the algorithm takes quadratic time on this graph as well, as shown in the plot of \cref{subfig:mistake_complexity_plots_sparse}.

Note that here the vanilla offline algorithm takes $O(n\cdot \alpha(n))$ time for handling this edge deletion.
Hence, for sparse graphs, the dynamic algorithm of~\cite{Li2022} can even become $n$ times slower than the vanilla offline algorithm, which is evident in \cref{subfig:mistake_complexity_plots_sparse}.

This behavior occurs in \cite[Procedure 4, Lines~1-6 and Line~27]{Li2022}, while~\cite[Lemma~4.16]{Li2022} does not make a thorough complexity analysis of these lines.

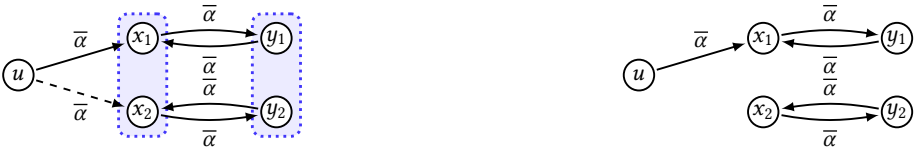
\begin{figure}
\scalebox{0.9}{%
\begin{subfigure}[b]{0.45\textwidth}
\centering
\begin{tikzpicture}[thick, >=latex, node distance=0.2cm and 1.5cm,
pre/.style={<-,shorten >= 1pt, shorten <=1pt, thick},
post/.style={->,shorten >= 1pt, shorten <=1pt,  thick},
und/.style={very thick, draw=gray},
node/.style={circle, minimum size=4.5mm, draw=black!100, fill=white!100, thick, inner sep=0},
virt/.style={circle,draw=black!50,fill=black!20, opacity=0}]

\newcommand{\xdisposition}{0}
\newcommand{\ydisposition}{0}
\newcommand{\xstep}{1.5}
\newcommand{\ystep}{1}
\def\bend{20}

\node	[node]		(u)	at(0,0) {$u$};
\node	[node, above right=of u]		(x1)	 {$x_1$};
\node	[node, below right=of u]		(x2)	 {$x_2$};

\node	[node, right=of x1]		(y1)	 {$y_1$};
\node	[node, right=of x2]		(y2)	 {$y_2$};

\draw[post] (u) to node[above] {$\CloseParenthesis$} (x1);§
\draw[post, dashed] (u) to node[below] {$\CloseParenthesis$} (x2);

\draw[post, bend left=10] (x1) to node[above] {$\CloseParenthesis$} (y1);
\draw[post, bend right=10] (x2) to node[below] {$\CloseParenthesis$} (y2);

\draw[post, bend left=10] (y1) to node[below] {$\CloseParenthesis$} (x1);
\draw[post, bend right=10] (y2) to node[above] {$\CloseParenthesis$} (x2);

\begin{pgfonlayer}{bg}
\node[box, very thick, rounded corners, draw=mybluecolor, dotted, fill=mybluecolor!10, fit=(x1)(x2)] {};
\node[box, very thick, rounded corners, draw=mybluecolor, dotted, fill=mybluecolor!10, fit=(y1)(y2)] {};
\end{pgfonlayer}

\end{tikzpicture}
\end{subfigure}
}
\hfill
\scalebox{0.9}{%
\begin{subfigure}[b]{0.45\textwidth}
\centering
\begin{tikzpicture}[thick, >=latex, node distance=0.2cm and 1.5cm,
pre/.style={<-,shorten >= 1pt, shorten <=1pt, },
post/.style={->,shorten >= 1pt, shorten <=1pt,  },
und/.style={very thick, draw=gray},
node/.style={circle, minimum size=4.5mm, draw=black!100, fill=white!100, thick, inner sep=0},
virt/.style={circle,draw=black!50,fill=black!20, opacity=0}]

\newcommand{\xdisposition}{0}
\newcommand{\ydisposition}{0}
\newcommand{\xstep}{1.5}
\newcommand{\ystep}{1}
\def\bend{20}

\node	[node]		(u)	at(0,0) {$u$};
\node	[node, above right=of u]		(x1)	 {$x_1$};
\node	[node, below right=of u]		(x2)	 {$x_2$};

\node	[node, right=of x1]		(y1)	 {$y_1$};
\node	[node, right=of x2]		(y2)	 {$y_2$};

\draw[post] (u) to node[above] {$\CloseParenthesis$} (x1);

\draw[post, bend left=10] (x1) to node[above] {$\CloseParenthesis$} (y1);
\draw[post, bend right=10] (x2) to node[below] {$\CloseParenthesis$} (y2);

\draw[post, draw=black, bend left=10] (y1) to node[below] {$\CloseParenthesis$} (x1);
\draw[post, draw=black, bend right=10] (y2) to node[above] {$\CloseParenthesis$} (x2);


\end{tikzpicture}
\end{subfigure}
}
\caption{
A graph $G$ where $\{x_1, x_2\}$ and $\{y_1, y_2\}$ are inter-reachable (left).
After deleting the edge $u\DTo{\CloseParenthesis}x_2$, every pair of nodes are not inter-reachable (right).
\label{fig:mistake_correctness}
}
\end{figure}

\Paragraph{Correctness counterexample.}
\cref{fig:mistake_correctness} showcases a simple example in which~\cite{Li2022} gives an incorrect answer.
Initially (left), $G$ contains the two DSCCs $S_x=\{x_1, x_2\}$ and $S_y=\{y_1, y_2\}$, where $S_y$ is formed after $S_x$ has been formed.
Deleting the edge $u\DTo{\CloseParenthesis}x_2$ from $G$ (right) breaks $S_x$ into singleton DSCCs $\{x_1\}$ and $\{x_2\}$, which, in turn, breaks $S_y$ into singleton DSCCs $\{y_1\}$ and $\{y_2\}$.
Instead,~\cite{Li2022} incorrectly returns that $S_x$ and $S_y$ are present even after deleting $u\DTo{\CloseParenthesis}x_2$.
Intuitively, this occurs because the algorithm fails to break $S_x$, as it incorrectly assumes the existence of a path $x_1\DTo{\alpha}y_1\DPath{}y_2\DTo{\ov{\alpha}} x_2$.
This, however, is not true, as the inner sub-path $y_1\DPath{}y_2$ is dependent on $x_1$ and $x_2$ being already connected.
This behavior occurs in~\cite[Lines~1-8]{Li2022}, while \cite[Lemma~4.8]{Li2022} incorrectly argues that the early termination (happening in Line 8) implies the existence of an alternative path (keeping $x_1$ and $x_2$ connected).
The behavior is also reproducible by the accompanying tool.
\section{Experimental Details}\label{sec:app_experiments}
\cref{tab:data_dependence_analysis_statistics} and \cref{tab:alias_analysis_statistics} shows the statistics of the graphs extracted from the benchmarks, which includes the number of nodes, number of edges and number of parenthesis types in the full graph. 
Parentheses correspond to calling contexts in \cref{tab:data_dependence_analysis_statistics} and to fields in \cref{tab:alias_analysis_statistics}.
We also tabulate the sequence length of 90\%-10\% split presented in \cref{sec:experiments} in the mentioned tables.
\begin{table}[t]
\def\RatioDataDependence{.475}
\def\RatioAlias{.475}
\setlength\tabcolsep{4.5pt}
\footnotesize
\centering
\begin{minipage}{\RatioDataDependence\textwidth}
\caption{   \label{tab:data_dependence_analysis_statistics}
 Data Dependence Analysis}
\centering
\begin{tabular}{@{}c c c c c@{}}
\hline\hline
 Name   &   $\#$Nodes   &  $\#$Edges   &  $\#$Labels  &    Seq. Length\\
\hline
btree & 1811 & 1757 & 801 & 1581 \\
check & 5240 & 5267 & 2167 & 4740 \\
compiler & 4189 & 4101 & 1646 & 3690 \\
compress & 4375 & 4238 & 1721 & 3814 \\
crypto & 6202 & 6300 & 2540 & 5670 \\
derby & 6116 & 5948 & 2358 & 5353 \\
helloworld & 4074 & 3969 & 1596 & 3572 \\
mpegaudio & 9650 & 9391 & 3564 & 8451 \\
mushroom & 899 & 809 & 376 & 728 \\
parser & 1686 & 1561 & 690 & 1404 \\
sample & 931 & 834 & 389 & 750 \\
scimark & 4583 & 4429 & 1782 & 3986 \\
startup & 5493 & 5367 & 2165 & 4830 \\
sunflow & 3891 & 3792 & 1520 & 3412 \\
xml & 23922 & 24391 & 9128 & 21951 \\
\hline
\end{tabular}
\end{minipage}
\hfill
\begin{minipage}{\RatioAlias\textwidth}
  \caption{  \label{tab:alias_analysis_statistics}
 Alias Analysis}
\centering
\begin{tabular}{@{}c c c c c@{}}
\hline\hline
Name   &   $\#$Nodes   &  $\#$Edges   &  $\#$Labels  &    Seq. Length\\
\hline
antlr & 23031 & 21353 & 1246 & 19217 \\
bloat & 26656 & 23598 & 1360 & 21238 \\
chart & 51356 & 44501 & 3132 & 40050 \\
eclipse & 24004 & 21943 & 1346 & 19748 \\
fop & 46253 & 39125 & 2857 & 35212 \\
hsqldb & 21646 & 20271 & 1160 & 18243 \\
jython & 28033 & 24889 & 1398 & 22400 \\
luindex & 22631 & 20915 & 1228 & 18823 \\
lusearch & 23344 & 21569 & 1275 & 19412 \\
pmd & 24586 & 22522 & 1322 & 20269 \\
xalan & 21574 & 20186 & 1152 & 18167 \\
\hline
\end{tabular}

\end{minipage}

\end{table}

\cref{fig:experiments_suite1_80_20} and \cref{fig:experiments_suite3_100_0} show experimental results on update sequences constructed similarly to those in \cref{sec:experiments}, but with a 80\%-20\% split (\cref{fig:experiments_suite1_80_20}) and 100\%-0\% split (\cref{fig:experiments_suite3_100_0}), as opposed to the 90\%-10\% split presented in \cref{sec:experiments}.
We observe that all the conclusions made in \cref{sec:experiments} also hold in these alternative settings.

\begin{figure}
\captionsetup[subfigure]{aboveskip=-1pt,belowskip=-1pt}

\pgfplotsset{very tick label/.append style={font=\small}}
\def\RatioDataDependence{.58}
\def\PlotWidthDataDependence{9cm}
\def\PlotWidthAlias{6.5cm}
\def\RatioAlias{.38}

\def\subfigureTextWidthDataDependence{\NumBenchmarksDataDependence/\NumBenchmarksTotal}
\def\subfigureTextWidthAlias{\NumBenchmarksAlias/\NumBenchmarksTotal}

\def\plotheight{4cm}
\def\barwidth{2.5pt}
\def\bardistance{1pt}
\def\scaleboxvalue{0.9}

\def\offlinecolor{black!10}
\def\ddLogcolor{black!50}
\def\dynamiccolor{black!80}

\begin{subfigure}[T]{\RatioDataDependence\textwidth}
\scalebox{\scaleboxvalue}{%
\begin{tikzpicture}[tight background, inner sep=2pt,]
\begin{axis}[inner sep=2pt,
		title={\large \underline{Data Dependence Analysis}},
    ybar=\bardistance,
		enlarge y limits={abs=0.6cm,upper},
		enlarge x limits={abs=8pt},
		legend style={legend pos=north west, draw=black, legend columns=-1,/tikz/every even column/.append style={column sep=0.2cm}},
		ylabel near ticks,
    ylabel={microsecs},
		xtick pos=left,
    xtick={1,2,3,4,5,6,7,8,9,10,11,12,13,14,15},
    xticklabels={btree,check,compiler,compress,crypto,derby,helloworld,mpegaudio,mushroom,parser,sample,scimark,startup,sunflow,xml},
    ymajorgrids=true,
    grid style=dashed,
		width=\PlotWidthDataDependence,
		height=\plotheight,
    legend style={font=\footnotesize},
    xticklabel style={font=\footnotesize,rotate=45,anchor=east},
    yticklabel style={font=\footnotesize},
    ylabel style={font=\footnotesize},
    bar width=\barwidth,
    ymode=log,
    log origin=infty,
]

\addplot [fill=\offlinecolor] file {figures/suite1_80_20_runtime/expt_datadep_offline_incr.dat};
\addplot [fill=\ddLogcolor] file {figures/suite1_80_20_runtime/expt_datadep_dyn_incr_ddlog.dat};
\addplot[fill=\dynamiccolor] file {figures/suite1_80_20_runtime/expt_datadep_dyndyck_incr.dat};
\legend{Offline,DDlog,Dynamic}

\end{axis}
\end{tikzpicture}
}%
\caption{\label{subfig:datadependence_all_incr_suite1_80_20}
Incremental updates.
}
\end{subfigure}
\hspace*{\fill}
\begin{subfigure}[T]{\RatioAlias\textwidth}
\scalebox{\scaleboxvalue}{%
\begin{tikzpicture}[tight background]
\begin{axis}[inner sep=2pt,,
		title={\large \underline{Alias Analysis}},
    ybar=\bardistance,
		enlarge y limits={abs=0.6cm,upper},
		enlarge x limits={abs=8pt},
		legend style={legend pos=north west, draw=black, legend columns=-1,/tikz/every even column/.append style={column sep=0.2cm}},
		ylabel near ticks,
    ylabel={microsecs},
		xtick pos=left,
    xtick={1,2,3,4,5,6,7,8,9,10,11},
    xticklabels={antlr, bloat, chart, eclipse, fop, hsqldb, jython, luindex, lusearch, pmd, xalan},
    ymajorgrids=true,
    grid style=dashed,
		width=\PlotWidthAlias,
		height=\plotheight,
    legend style={font=\footnotesize},
    xticklabel style={font=\footnotesize,rotate=45,anchor=east},
    yticklabel style={font=\footnotesize},
    ylabel style={font=\footnotesize},
    bar width=\barwidth,
    ymode=log,
    log origin=infty,
]

\addplot [fill=\offlinecolor] file {figures/suite1_80_20_runtime/expt_alias_offline_incr.dat};
\addplot [fill=\ddLogcolor] file {figures/suite1_80_20_runtime/expt_alias_dyn_incr_ddlog.dat};
\addplot[fill=\dynamiccolor] file {figures/suite1_80_20_runtime/expt_alias_dyndyck_incr.dat};
\legend{Offline,DDlog,Dynamic}

\end{axis}
\end{tikzpicture}
}%
\vspace{7pt}
\caption{\label{subfig:alias_all_incr_suite1_80_20}
Incremental updates.
}
\end{subfigure}
\\[1em]
\begin{subfigure}[T]{\RatioDataDependence\textwidth}
\scalebox{\scaleboxvalue}{%
\begin{tikzpicture}[tight background, inner sep=2pt,]
\begin{axis}[inner sep=2pt,
    ybar=\bardistance,
		enlarge y limits={abs=0.6cm,upper},
		enlarge x limits={abs=8pt},
		legend style={legend pos=north west, draw=black, legend columns=-1,/tikz/every even column/.append style={column sep=0.2cm}},
		ylabel near ticks,
    ylabel={microsecs},
		xtick pos=left,
    xtick={1,2,3,4,5,6,7,8,9,10,11,12,13,14,15},
    xticklabels={btree,check,compiler,compress,crypto,derby,helloworld,mpegaudio,mushroom,parser,sample,scimark,startup,sunflow,xml},
    ymajorgrids=true,
    grid style=dashed,
		width=\PlotWidthDataDependence,
		height=\plotheight,
    legend style={font=\footnotesize},
    xticklabel style={font=\footnotesize,rotate=45,anchor=east},
    yticklabel style={font=\footnotesize},
    ylabel style={font=\footnotesize},
    bar width=\barwidth,
    ymode=log,
    log origin=infty,
]

\addplot [fill=\offlinecolor] file {figures/suite1_80_20_runtime/expt_datadep_offline_decr.dat};
\addplot [fill=\ddLogcolor] file {figures/suite1_80_20_runtime/expt_datadep_dyn_decr_ddlog.dat};
\addplot[fill=\dynamiccolor] file {figures/suite1_80_20_runtime/expt_datadep_dyndyck_decr.dat};
\legend{Offline,DDlog,Dynamic}

\end{axis}
\end{tikzpicture}
}%
\caption{\label{subfig:datadependence_all_decr_suite1_80_20}
Decremental updates.
}

\end{subfigure}
\hspace*{\fill}
\begin{subfigure}[T]{\RatioAlias\textwidth}
\scalebox{\scaleboxvalue}{%
\begin{tikzpicture}[tight background]
\begin{axis}[inner sep=2pt,,
    ybar=\bardistance,
		enlarge y limits={abs=0.6cm,upper},
		enlarge x limits={abs=8pt},
		legend style={legend pos=north west, draw=black, legend columns=-1,/tikz/every even column/.append style={column sep=0.2cm}},
		ylabel near ticks,
    ylabel={microsecs},
		xtick pos=left,
    xtick={1,2,3,4,5,6,7,8,9,10,11},
    xticklabels={antlr, bloat, chart, eclipse, fop, hsqldb, jython, luindex, lusearch, pmd, xalan},
    ymajorgrids=true,
    grid style=dashed,
		width=\PlotWidthAlias,
		height=\plotheight,
    legend style={font=\footnotesize},
    xticklabel style={font=\footnotesize,rotate=45,anchor=east},
    yticklabel style={font=\footnotesize},
    ylabel style={font=\footnotesize},
    bar width=\barwidth,
    ymode=log,
    log origin=infty,
]

\addplot [fill=\offlinecolor] file {figures/suite1_80_20_runtime/expt_alias_offline_decr.dat};
\addplot [fill=\ddLogcolor] file {figures/suite1_80_20_runtime/expt_alias_dyn_decr_ddlog.dat};
\addplot[fill=\dynamiccolor] file {figures/suite1_80_20_runtime/expt_alias_dyndyck_decr.dat};
\legend{Offline,DDlog,Dynamic}

\end{axis}
\end{tikzpicture}
}%
\vspace{7pt}
\caption{\label{subfig:alias_all_decr_suite1_80_20}
Decremental updates.
}
\end{subfigure}
\\[1em]
\begin{subfigure}[T]{\RatioDataDependence\textwidth}
\scalebox{\scaleboxvalue}{%
\begin{tikzpicture}[tight background, inner sep=2pt,]
\begin{axis}[inner sep=2pt,
    ybar=\bardistance,
		enlarge y limits={abs=0.6cm,upper},
		enlarge x limits={abs=8pt},
		legend style={legend pos=north west, draw=black, legend columns=-1,/tikz/every even column/.append style={column sep=0.2cm}},
		ylabel near ticks,
    ylabel={microsecs},
		xtick pos=left,
    xtick={1,2,3,4,5,6,7,8,9,10,11,12,13,14,15},
    xticklabels={btree,check,compiler,compress,crypto,derby,helloworld,mpegaudio,mushroom,parser,sample,scimark,startup,sunflow,xml},
    ymajorgrids=true,
    grid style=dashed,
		width=\PlotWidthDataDependence,
		height=\plotheight,
    legend style={font=\footnotesize},
    xticklabel style={font=\footnotesize,rotate=45,anchor=east},
    yticklabel style={font=\footnotesize},
    ylabel style={font=\footnotesize},
    bar width=\barwidth,
    ymode=log,
    log origin=infty,
]

\addplot [fill=\offlinecolor] file {figures/suite1_80_20_runtime/expt_datadep_offline_mixed.dat};
\addplot [fill=\ddLogcolor] file {figures/suite1_80_20_runtime/expt_datadep_dyn_mixed_ddlog.dat};
\addplot[fill=\dynamiccolor] file {figures/suite1_80_20_runtime/expt_datadep_dyndyck_mixed.dat};
\legend{Offline,DDlog,Dynamic}

\end{axis}
\end{tikzpicture}
}%
\caption{\label{subfig:datadependence_all_mixed_suite1_80_20}
Mixed updates.
}

\end{subfigure}
\hspace*{\fill}
\begin{subfigure}[T]{\RatioAlias\textwidth}
\scalebox{\scaleboxvalue}{%
\begin{tikzpicture}[tight background]
\begin{axis}[inner sep=2pt,,
    ybar=\bardistance,
		enlarge y limits={abs=0.6cm,upper},
		enlarge x limits={abs=8pt},
		legend style={legend pos=north west, draw=black, legend columns=-1,/tikz/every even column/.append style={column sep=0.2cm}},
		ylabel near ticks,
    ylabel={microsecs},
		xtick pos=left,
    xtick={1,2,3,4,5,6,7,8,9,10,11},
    xticklabels={antlr, bloat, chart, eclipse, fop, hsqldb, jython, luindex, lusearch, pmd, xalan},
    ymajorgrids=true,
    grid style=dashed,
		width=\PlotWidthAlias,
		height=\plotheight,
    legend style={font=\footnotesize},
    xticklabel style={font=\footnotesize,rotate=45,anchor=east},
    yticklabel style={font=\footnotesize},
    ylabel style={font=\footnotesize},
    bar width=\barwidth,
    ymode=log,
    log origin=infty,
]

\addplot [fill=\offlinecolor] file {figures/suite1_80_20_runtime/expt_alias_offline_mixed.dat};
\addplot [fill=\ddLogcolor] file {figures/suite1_80_20_runtime/expt_alias_dyn_mixed_ddlog.dat};
\addplot[fill=\dynamiccolor] file {figures/suite1_80_20_runtime/expt_alias_dyndyck_mixed.dat};
\legend{Offline,DDlog,Dynamic}

\end{axis}
\end{tikzpicture}
}%
\vspace{7pt}
\caption{\label{subfig:alias_all_mixed_suite1_80_20}
Mixed updates.
}
\end{subfigure}
\caption{\label{fig:experiments_suite1_80_20}
The average time to handle a single update on on-the-fly data-dependence analysis (\subref{subfig:datadependence_all_incr_suite1_80_20}, \subref{subfig:datadependence_all_decr_suite1_80_20}, \subref{subfig:datadependence_all_mixed_suite1_80_20}) and on-the-fly alias analysis (\subref{subfig:alias_all_incr_suite1_80_20}, \subref{subfig:alias_all_decr_suite1_80_20}, \subref{subfig:alias_all_mixed_suite1_80_20}) for sequence files generated with 80-20 split of original graph.
Note that all results are in log-scale.
}
\end{figure}
\begin{figure}
\captionsetup[subfigure]{aboveskip=-1pt,belowskip=-1pt}

\pgfplotsset{very tick label/.append style={font=\small}}
\def\RatioDataDependence{.58}
\def\PlotWidthDataDependence{9cm}
\def\PlotWidthAlias{6.5cm}
\def\RatioAlias{.38}

\def\subfigureTextWidthDataDependence{\NumBenchmarksDataDependence/\NumBenchmarksTotal}
\def\subfigureTextWidthAlias{\NumBenchmarksAlias/\NumBenchmarksTotal}

\def\plotheight{4cm}
\def\barwidth{2.5pt}
\def\bardistance{1pt}
\def\scaleboxvalue{0.9}

\def\offlinecolor{black!10}
\def\ddLogcolor{black!50}
\def\dynamiccolor{black!80}

\begin{subfigure}[T]{\RatioDataDependence\textwidth}
\scalebox{\scaleboxvalue}{%
\begin{tikzpicture}[tight background, inner sep=2pt,]
\begin{axis}[inner sep=2pt,
		title={\large \underline{Data Dependence Analysis}},
    ybar=\bardistance,
		enlarge y limits={abs=0.6cm,upper},
		enlarge x limits={abs=8pt},
		legend style={legend pos=north west, draw=black, legend columns=-1,/tikz/every even column/.append style={column sep=0.2cm}},
		ylabel near ticks,
    ylabel={microsecs},
		xtick pos=left,
    xtick={1,2,3,4,5,6,7,8,9,10,11,12,13,14,15},
    xticklabels={btree,check,compiler,compress,crypto,derby,helloworld,mpegaudio,mushroom,parser,sample,scimark,startup,sunflow,xml},
    ymajorgrids=true,
    grid style=dashed,
		width=\PlotWidthDataDependence,
		height=\plotheight,
    legend style={font=\footnotesize},
    xticklabel style={font=\footnotesize,rotate=45,anchor=east},
    yticklabel style={font=\footnotesize},
    ylabel style={font=\footnotesize},
    bar width=\barwidth,
    ymode=log,
    log origin=infty,
]

\addplot [fill=\offlinecolor] file {figures/suite3_100_0_runtime/expt_datadep_offline_incr.dat};
\addplot [fill=\ddLogcolor] file {figures/suite3_100_0_runtime/expt_datadep_dyn_incr_ddlog.dat};
\addplot[fill=\dynamiccolor] file {figures/suite3_100_0_runtime/expt_datadep_dyndyck_incr.dat};
\legend{Offline,DDlog,Dynamic}

\end{axis}
\end{tikzpicture}
}%
\caption{\label{subfig:datadependence_all_incr_suite3_100_0}
Incremental updates.
}
\end{subfigure}
\hspace*{\fill}
\begin{subfigure}[T]{\RatioAlias\textwidth}
\scalebox{\scaleboxvalue}{%
\begin{tikzpicture}[tight background]
\begin{axis}[inner sep=2pt,,
		title={\large \underline{Alias Analysis}},
    ybar=\bardistance,
		enlarge y limits={abs=0.6cm,upper},
		enlarge x limits={abs=8pt},
		legend style={legend pos=north west, draw=black, legend columns=-1,/tikz/every even column/.append style={column sep=0.2cm}},
		ylabel near ticks,
    ylabel={microsecs},
		xtick pos=left,
    xtick={1,2,3,4,5,6,7,8,9,10,11},
    xticklabels={antlr, bloat, chart, eclipse, fop, hsqldb, jython, luindex, lusearch, pmd, xalan},
    ymajorgrids=true,
    grid style=dashed,
		width=\PlotWidthAlias,
		height=\plotheight,
    legend style={font=\footnotesize},
    xticklabel style={font=\footnotesize,rotate=45,anchor=east},
    yticklabel style={font=\footnotesize},
    ylabel style={font=\footnotesize},
    bar width=\barwidth,
    ymode=log,
    log origin=infty,
]

\addplot [fill=\offlinecolor] file {figures/suite3_100_0_runtime/expt_alias_offline_incr.dat};
\addplot [fill=\ddLogcolor] file {figures/suite3_100_0_runtime/expt_alias_dyn_incr_ddlog.dat};
\addplot[fill=\dynamiccolor] file {figures/suite3_100_0_runtime/expt_alias_dyndyck_incr.dat};
\legend{Offline,DDlog,Dynamic}

\end{axis}
\end{tikzpicture}
}%
\vspace{7pt}
\caption{\label{subfig:alias_all_incr_suite3_100_0}
Incremental updates.
}
\end{subfigure}
\\[1em]
\begin{subfigure}[T]{\RatioDataDependence\textwidth}
\scalebox{\scaleboxvalue}{%
\begin{tikzpicture}[tight background, inner sep=2pt,]
\begin{axis}[inner sep=2pt,
    ybar=\bardistance,
		enlarge y limits={abs=0.6cm,upper},
		enlarge x limits={abs=8pt},
		legend style={legend pos=north west, draw=black, legend columns=-1,/tikz/every even column/.append style={column sep=0.2cm}},
		ylabel near ticks,
    ylabel={microsecs},
		xtick pos=left,
    xtick={1,2,3,4,5,6,7,8,9,10,11,12,13,14,15},
    xticklabels={btree,check,compiler,compress,crypto,derby,helloworld,mpegaudio,mushroom,parser,sample,scimark,startup,sunflow,xml},
    ymajorgrids=true,
    grid style=dashed,
		width=\PlotWidthDataDependence,
		height=\plotheight,
    legend style={font=\footnotesize},
    xticklabel style={font=\footnotesize,rotate=45,anchor=east},
    yticklabel style={font=\footnotesize},
    ylabel style={font=\footnotesize},
    bar width=\barwidth,
    ymode=log,
    log origin=infty,
]

\addplot [fill=\offlinecolor] file {figures/suite3_100_0_runtime/expt_datadep_offline_decr.dat};
\addplot [fill=\ddLogcolor] file {figures/suite3_100_0_runtime/expt_datadep_dyn_decr_ddlog.dat};
\addplot[fill=\dynamiccolor] file {figures/suite3_100_0_runtime/expt_datadep_dyndyck_decr.dat};
\legend{Offline,DDlog,Dynamic}

\end{axis}
\end{tikzpicture}
}%
\caption{\label{subfig:datadependence_all_decr_suite3_100_0}
Decremental updates.
}

\end{subfigure}
\hspace*{\fill}
\begin{subfigure}[T]{\RatioAlias\textwidth}
\scalebox{\scaleboxvalue}{%
\begin{tikzpicture}[tight background]
\begin{axis}[inner sep=2pt,,
    ybar=\bardistance,
		enlarge y limits={abs=0.6cm,upper},
		enlarge x limits={abs=8pt},
		legend style={legend pos=north west, draw=black, legend columns=-1,/tikz/every even column/.append style={column sep=0.2cm}},
		ylabel near ticks,
    ylabel={microsecs},
		xtick pos=left,
    xtick={1,2,3,4,5,6,7,8,9,10,11},
    xticklabels={antlr, bloat, chart, eclipse, fop, hsqldb, jython, luindex, lusearch, pmd, xalan},
    ymajorgrids=true,
    grid style=dashed,
		width=\PlotWidthAlias,
		height=\plotheight,
    legend style={font=\footnotesize},
    xticklabel style={font=\footnotesize,rotate=45,anchor=east},
    yticklabel style={font=\footnotesize},
    ylabel style={font=\footnotesize},
    bar width=\barwidth,
    ymode=log,
    log origin=infty,
]

\addplot [fill=\offlinecolor] file {figures/suite3_100_0_runtime/expt_alias_offline_decr.dat};
\addplot [fill=\ddLogcolor] file {figures/suite3_100_0_runtime/expt_alias_dyn_decr_ddlog.dat};
\addplot[fill=\dynamiccolor] file {figures/suite3_100_0_runtime/expt_alias_dyndyck_decr.dat};
\legend{Offline,DDlog,Dynamic}

\end{axis}
\end{tikzpicture}
}%
\vspace{7pt}
\caption{\label{subfig:alias_all_decr_suite3_100_0}
Decremental updates.
}
\end{subfigure}
\\[1em]
\begin{subfigure}[T]{\RatioDataDependence\textwidth}
\scalebox{\scaleboxvalue}{%
\begin{tikzpicture}[tight background, inner sep=2pt,]
\begin{axis}[inner sep=2pt,
    ybar=\bardistance,
		enlarge y limits={abs=0.6cm,upper},
		enlarge x limits={abs=8pt},
		legend style={legend pos=north west, draw=black, legend columns=-1,/tikz/every even column/.append style={column sep=0.2cm}},
		ylabel near ticks,
    ylabel={microsecs},
		xtick pos=left,
    xtick={1,2,3,4,5,6,7,8,9,10,11,12,13,14,15},
    xticklabels={btree,check,compiler,compress,crypto,derby,helloworld,mpegaudio,mushroom,parser,sample,scimark,startup,sunflow,xml},
    ymajorgrids=true,
    grid style=dashed,
		width=\PlotWidthDataDependence,
		height=\plotheight,
    legend style={font=\footnotesize},
    xticklabel style={font=\footnotesize,rotate=45,anchor=east},
    yticklabel style={font=\footnotesize},
    ylabel style={font=\footnotesize},
    bar width=\barwidth,
    ymode=log,
    log origin=infty,
]

\addplot [fill=\offlinecolor] file {figures/suite3_100_0_runtime/expt_datadep_offline_mixed.dat};
\addplot [fill=\ddLogcolor] file {figures/suite3_100_0_runtime/expt_datadep_dyn_mixed_ddlog.dat};
\addplot[fill=\dynamiccolor] file {figures/suite3_100_0_runtime/expt_datadep_dyndyck_mixed.dat};
\legend{Offline,DDlog,Dynamic}

\end{axis}
\end{tikzpicture}
}%
\caption{\label{subfig:datadependence_all_mixed_suite3_100_0}
Mixed updates.
}

\end{subfigure}
\hspace*{\fill}
\begin{subfigure}[T]{\RatioAlias\textwidth}
\scalebox{\scaleboxvalue}{%
\begin{tikzpicture}[tight background]
\begin{axis}[inner sep=2pt,,
    ybar=\bardistance,
		enlarge y limits={abs=0.6cm,upper},
		enlarge x limits={abs=8pt},
		legend style={legend pos=north west, draw=black, legend columns=-1,/tikz/every even column/.append style={column sep=0.2cm}},
		ylabel near ticks,
    ylabel={microsecs},
		xtick pos=left,
    xtick={1,2,3,4,5,6,7,8,9,10,11},
    xticklabels={antlr, bloat, chart, eclipse, fop, hsqldb, jython, luindex, lusearch, pmd, xalan},
    ymajorgrids=true,
    grid style=dashed,
		width=\PlotWidthAlias,
		height=\plotheight,
    legend style={font=\footnotesize},
    xticklabel style={font=\footnotesize,rotate=45,anchor=east},
    yticklabel style={font=\footnotesize},
    ylabel style={font=\footnotesize},
    bar width=\barwidth,
    ymode=log,
    log origin=infty,
]

\addplot [fill=\offlinecolor] file {figures/suite3_100_0_runtime/expt_alias_offline_mixed.dat};
\addplot [fill=\ddLogcolor] file {figures/suite3_100_0_runtime/expt_alias_dyn_mixed_ddlog.dat};
\addplot[fill=\dynamiccolor] file {figures/suite3_100_0_runtime/expt_alias_dyndyck_mixed.dat};
\legend{Offline,DDlog,Dynamic}

\end{axis}
\end{tikzpicture}
}%
\vspace{7pt}
\caption{\label{subfig:alias_all_mixed_suite3_100_0}
Mixed updates.
}
\end{subfigure}
\caption{\label{fig:experiments_suite3_100_0}
The average time to handle a single update on on-the-fly data-dependence analysis (\subref{subfig:datadependence_all_incr_suite3_100_0}, \subref{subfig:datadependence_all_decr_suite3_100_0}, \subref{subfig:datadependence_all_mixed_suite3_100_0}) and on-the-fly alias analysis (\subref{subfig:alias_all_incr_suite3_100_0}, \subref{subfig:alias_all_decr_suite3_100_0}, \subref{subfig:alias_all_mixed_suite3_100_0}) for sequence files generated with 100-0 split of original graph.
Note that all results are in log-scale.
}
\end{figure}

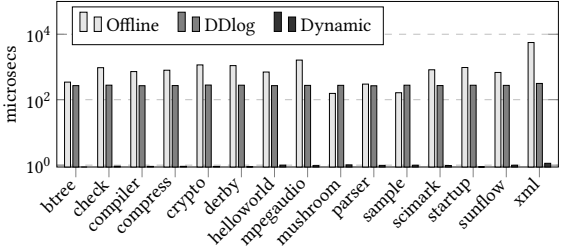
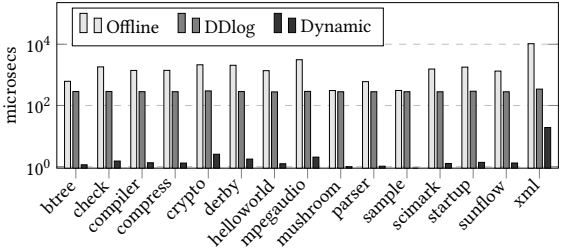
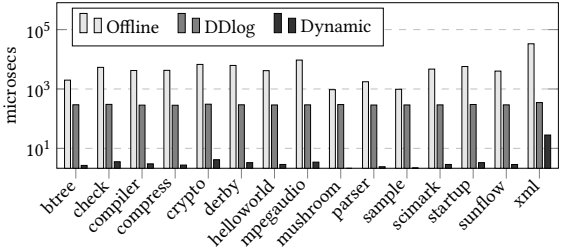
\end{document}